\documentclass[12pt, oneside]{amsart}
\pdfoutput=1 
\usepackage{geometry}
\usepackage{lmodern}
\usepackage{amsmath}
\usepackage{amssymb}
\usepackage{amsthm}
\usepackage{comment}
\usepackage{color}
\usepackage[usenames, table, dvipsnames]{xcolor}
\usepackage{graphicx}
\usepackage{caption}
\usepackage{subcaption}
\usepackage{soul}
\usepackage[countmax]{subfloat}
\usepackage{float}
\usepackage{rotfloat}
\usepackage{setspace}
\usepackage{esint}
\usepackage{pdflscape}
\usepackage{tikz}
\usetikzlibrary{shapes,arrows,plotmarks}

\allowdisplaybreaks
\usepackage[authoryear]{natbib}
\definecolor{MyDarkBlue}{rgb}{0,0.08,0.45}
\usepackage[unicode=true,pdfusetitle, bookmarks=true,bookmarksnumbered=false,bookmarksopen=false, breaklinks=false,pdfborder={0 0 1},backref=section,colorlinks=true, linkcolor = MyDarkBlue, citecolor = MyDarkBlue]{hyperref}
\usepackage{breakurl}
\usepackage{verbatim}
\usepackage{bbm}
\usepackage[para]{threeparttable}
\usepackage[ruled,vlined]{algorithm2e}

\PassOptionsToPackage{normalem}{ulem}
\usepackage{ulem}

\makeatletter

\numberwithin{equation}{section}

\hypersetup{colorlinks=true, linkcolor=MyDarkBlue}

\newtheorem{ass}{{\bf \sc Assumption}}

\newtheorem{lemma}{{\bf\sc Lemma}}
\newtheorem{sublemma}{{\bf\sc Sublemma}}
\newtheorem{example}{{\bf\sc Example}}

\newtheorem{rem}{{\bf \sc Remark}}

\newtheorem{obs}{{\bf\sc Observation}}

\newtheorem{lem}{{\bf \sc Lemma}}[section]

\newtheorem{prop}{{\bf\sc Proposition}}[section]

\DeclareMathOperator*{\plim}{plim}
\DeclareMathOperator*{\argmin}{argmin}
\DeclareMathOperator*{\argmax}{argmax}

\providecommand{\Prob}{\mathrm{P}}

\renewcommand{\Pr}{\Prob}

\renewenvironment{proof}[1][Proof]{\noindent\text{#1.} }{\ \rule{0.5em}{0.5em}}

\makeatother

\makeatletter
\newcommand*{\centerfloat}{%
  \parindent \z@
  \leftskip \z@ \@plus 1fil \@minus \textwidth
  \rightskip\leftskip
  \parfillskip \z@skip}
\makeatother

\begin{document}

\vspace*{-1cm}

\title[Treatment Variant Aggregation to Select Policies]{Selecting the Most Effective Nudge: Evidence from a Large-Scale Experiment on Immunization}

\author{Abhijit Banerjee$^{\dagger,*,\S}$ }
\author{Arun G. Chandrasekhar$^{\ddagger,*,\S}$}
\author{Suresh Dalpath$^{\#}$}
\author{Esther Duflo$^{\dagger,*,\S}$ }
\author{John Floretta$^{*}$ }
\author{Matthew O. Jackson$^{\ddagger,\star}$ }
\author{Harini Kannan$^{*}$ }
\author{Francine Loza$^{*}$ }
\author{Anirudh Sankar$^{\ddagger}$ }
\author{Anna Schrimpf$^{*}$ }
\author{Maheshwor Shrestha$^{\S}$ }
\date{\today}

\thanks{We are particularly grateful to the the Haryana Department of Health and Family Welfare for taking the lead on this intervention and allowing the
evaluation to take place. Among many other, we acknowledge the tireless support of Rajeev Arora, Amneet P. Kumar, Sonia Trikha, V.K. Bansal, Sube
Singh, and Harish Bisht. We are also grateful to Isaiah Andrews, Adel Javanmard, Adam McCloskey, and Karl Rohe  for  helpful discussions. We thank Emily Breza, Denis Chetvetrikov,  Paul Goldsmith-Pinkham, Ben Golub, Nargiz Kalantarova,  Shane Lubold, Tyler McCormick, Douglas Miller,  Francesca Molinari,  Suresh Naidu, Eric Verhoogen,  and participants at various seminars for suggestions. Financial support from USAID DIV, 3iE, J-PAL GPI, and MIT-Givewell, and NSF grant SES-2018554 is gratefully acknowledged.  Chandrasekhar is grateful to the Alfred P. Sloan foundation for support. We thank  Rachna Nag Chowdhuri, Diksha Radhakrishnan, Mohar Dey, Maaike Bijker, Niki Shrestha, Tanmayta Bansal, Sanjana Malhotra, Aicha Ben Dhia, Vidhi Sharma, Mridul Joshi, Rajdev Brar, Paul-Armand Veillon, and Chitra Balasubramanian, as well as Shreya Chaturvedi,  Vasu Chaudhary, Shobitha Cherian, Louis-Mael Jean, Anoop Singh Rawat, Elsa Trezeguet, and Meghna Yadav  for excellent research assistance.
}
\thanks{$^{\dagger}$Department of Economics, MIT}
\thanks{$^{\ddagger}$Department of Economics, Stanford University}
\thanks{$^{\#}$Public Health Planning, Policy, and M\&E, Health Department, Govt. of Haryana}
\thanks{$^{*}$J-PAL}
\thanks{$^{\S}$NBER}
\thanks{$^{\star}$Santa Fe Institute}

\begin{abstract}
Policymakers often choose a policy bundle that is a combination of different interventions in different dosages. We develop a new technique---\emph{treatment variant aggregation} (TVA)---to select a policy from a large factorial design.  TVA pools together policy variants that are not meaningfully different and prunes those deemed ineffective. This allows us to restrict attention to aggregated policy variants, consistently estimate their effects on the outcome, and estimate the best policy effect adjusting for the winner's curse.  We apply TVA to a large randomized controlled trial that tests interventions to stimulate demand for immunization in Haryana, India. The policies under consideration include reminders, incentives, and local ambassadors for community mobilization. Cross-randomizing these interventions, with different dosages or types of each intervention, yields 75 combinations. The policy with the largest impact (which combines incentives, ambassadors who are information hubs, and reminders) increases the number of immunizations by 44\% relative to the status quo. The most cost-effective policy (information hubs, ambassadors, and SMS reminders but no incentives) increases the number of immunizations per dollar by 9.1\% relative to status quo.
\end{abstract}

\maketitle

\setcounter{page}{0}
\newpage 
\section{Introduction}

In many settings, policymakers have to select the best policy bundle that combines several interventions that can come in different dosages and varieties, generating a large sets of potential policies.
For example, to reduce school dropout in a low income country, a policymaker might want to introduce a cash
transfer, which may be conditional or not, be more or less generous, and come in combination with a remedial education program or not.
Similarly in medicine, a particular treatment regimen may combine several drugs in different potential dosages. For example, the management
of HIV-AIDS was revolutionized in the mid 1990 by the combination of two or three drugs in subtle dosages, the famous ``AIDS Cocktail.''

The ideal strategy, if time and implementation capacities were not constraints,
would be to  experiment iteratively in the context until the best bundle is found. There is a growing literature
on how to conduct and analyze adaptive trials \citep{hadad2021confidence,kasy2021adaptive,zhan2021policy}. However, it is often not possible to conduct such sequential trials:  the window for experimentation may be
short before a policy must be chosen, or a decision on a treatment regimen or vaccine must be made quickly because of an urgent health threat. This was certainly true in the case of HIV-AIDs, where there were huge pressures to rapidly identify and approve a treatment.
In such cases, the only option may be to conduct a large scale experiment that simultaneously tests many different policy bundles. For example, Saquinavir, the first protease inhibitor, was approved by the FDA in 1995. As soon as 1996, the seminal ACTG-320 trial evaluated the three-drug cocktail \citep{hammer1997controlled}.

However, with several alternative interventions and multiple possible dosages, there is an enormous number of potential combinations, each of which is a unique policy bundle. For example, with three candidate interventions and five possible dosages for each one, the policymaker must select among 125 unique policy bundles. There is no clear guidance in the literature on how to  design and analyze such trials when the number of potential options is large. For some, this is a fundamental limitation of using randomized controlled trials to inform policy or even treatment choices \citep{easterly2014tyranny,ravallion2009should}. There are so many potentially relevant options that one is caught between oversimplifying the context by assuming that only a few of these options are relevant, or having to deal with too many variables for any finite sample.

We fill this gap by developing a methodology for  \emph{treatment variant aggregation} (TVA): a principled algorithm that pools together policy variants that have similar impact and prunes policy options that are ineffective. This reduces the dimensionality of the problem and enables identification of the best overall combination and the consistent estimation of its impact, accounting for the winner's curse.
We then apply this methodology to an urgent policy problem: how to increase the take up of childhood immunization in India. We analyze a large-scale experiment that we carried out in partnership with the government in the state of Haryana, India, where three candidate interventions and several different dosages generated 75 potential unique policies spread over 915 villages. We identify the most effective and the most
cost-effective combination of interventions.

Our methodology allows us to solve two problems that arise when evaluating large numbers of candidate policy bundles.

First, the researcher must decide how many and which potential policies to include in an experimental design, and how to analyze the trial.
One approach, recommended in \cite{mcckenzie_crosscut}, is to only include a limited number of bundles, for example by omitting all policies that are combinations of each other, or alternatively by varying only one marginal.  For example, the ACTG-320 trials compared the three drug cocktail (protease inhibitor plus two nucleosides), in one specific dosage, to the two nucleosides in the same dosage \citep{hammer1997controlled}.
However, the optimality of such an approach presumes an ``oracle property:'' that the researcher or policymaker already knows which policies are worth comparing.
We consider situations in which any of many policy bundles could turn out to be optimal.
For such situations, one can include all unique policy bundles in a fully-saturated factorial design.
This reduces power since each policy bundle may only be observed on a small sample.
To increase precision, researchers often attempt to pool policy bundles ex-post based on observed outcomes.
Without specific structure to the problem,  this can be misleading in finite samples, especially when interaction effects are small (so that a test against zero has limited power), but not quite zero (``local-to-zero'', \cite{muralidharan2019factorial}).
We need to find minimal and realistic assumptions on the inferential environment that enable a principled, data driven approach to reducing the dimensionality of the problem.

The second problem is that the impact of a policy that is estimated to be the best out of a set of $K$ unique policies (henceforth we use ``policy'' as a short-hand for ``policy bundle'') is likely to be over-estimated, due to the ``winner's curse'' \citep*{andrews2019inference}.  Some policy $k^\star$ could have the highest estimated impact partially due to a high true effect, but could also partially be due to randomness. Conditional on being the best in the data, some of the estimated impact is likely due to randomness. As a result, the expected impact of policy $k^\star$  is overestimated and the statistician must adjust for it.

The statistical setting that we analyze is as follows. There are $M$ possible intervention arms, $R$ possible ``dosages'' per arm (including zero) and therefore $K = R^M$ possible policies.
The policymaker is uncertain about which policies are effective.
However, it can be that in some circumstances an \emph{incremental} dosage change on some arm does not have a meaningful effect on the outcome,
for some other combination of other arms. For example, if there are three drugs $(x,y,z)$, and two adjacent dosages $z$ and $z'$ for the third drug (e.g., 10mg and 20mg),
then it is possible that $(x,y,z) = (x,y,z')$ for some particular choices of $x$ and $y$.  We can thus pool those two policies, $(x,y,z)$ and $(x,y,z')$, and treat them as one for all practical purposes.
The policymaker conducts \emph{treatment variant aggregation} (TVA). This \emph{pools} together policy variants that are not meaningfully different (e.g., $(x,y,z)$ and $(x,y,z')$ are pooled together as above) and \emph{prunes} all the combinations that are ineffective (those that pool with the null policy).
 TVA allows us to restrict attention to aggregated policy variants and only those that matter, which can improve estimation.   We discuss how we can use TVA to consistently estimate policy effects and estimate the best policy effect adjusting for the winner's curse. Our main contribution is to develop the TVA procedure, which proceeds in several steps.

The first step is to re-represent the fully-saturated factorial regression of the outcome on unique policies in terms of another, equivalent specification that tracks the effects of incremental dosages. TVA utilizes a Hasse diagram lattice of policy variants to deduce how zeros in the marginal effects determines pruning and pooling of variants.

To fix ideas, consider a simple example with two arms ($M = 2$) and two non-zero dosages for each ($R = 3$), yielding $K = 9$ unique policies.  So each arm can either be used or not, and used in either a low or high dose.
Let us represent these by $[T_1,T_2]=[0,0]$, $[0,1]$, $[0,2]$, $[1,0]$, etc., where the entries are the corresponding treatment levels with 0 being not used, and 1 being low and 2 high dosage.
A standard regression would just have a dummy variable for each particular policy combination  $[T_1,T_2]$, and then a corresponding coefficient  $\beta_{[T_1,T_2]}$.
Instead, it is useful to break these into marginal effects, in terms of a different regression:
{\small
\begin{align}\label{eq:simple-alpha}
y = \alpha_{[0,0]} &+ \alpha_{[1,0]} \cdot 1\{T_1 \geq 1\} \cdot 1\{T_2 = 0\}  + \alpha_{[0,1]} \cdot 1\{T_1 = 0\} \cdot 1\{T_2 \geq 1\}  \\ \nonumber
&+ \alpha_{[2,0]} \cdot 1\{T_1 = 2\} \cdot 1\{T_2 = 0\}+ \alpha_{[0,2]} \cdot 1\{T_1 = 0\} \cdot 1\{T_2 = 2\} \\ \nonumber
&+ \alpha_{[1,1]} \cdot 1\{T_1 \geq 1\} \cdot 1\{T_2 \geq 1\} + \alpha_{[2,1]} \cdot 1\{T_1 =2\}\cdot 1\{T_2 \geq 1\} \\ \nonumber
&+ \alpha_{[1,2]} \cdot  1\{T_1 \geq 1\} \cdot 1\{T_2 = 2\}
+ \alpha_{[2,2]}\cdot 1\{T_1 = 2\} \cdot 1\{T_2 = 2\} + \epsilon. \nonumber
\end{align}
}

In this specification, the $\alpha_{[r_1,r_2]} $ are all marginal effects, and hence, inspecting the vector $\alpha$  and checking which $\alpha_{[r_1,r_2]} = 0$ tells us which adjacent policies can be pooled together, and which ones can be pruned (pooled with the null policy; for instance, if $\alpha_{[1,0]}=0 $).
In a general factorial design of $K$ unique policies, we have regressors of the form $1\{T_1 \geq r_1, \ T_2 \geq r_2, \ldots, \ T_M \geq r_M\}$ for treatment arm intensities $T_m$ and thresholds $r_m$ for arm $m$, with $K-1$ regressors plus an intercept. At every stage we ask whether an incremental increase in dosage for a given arm of some policy causes a marginal change. That is, we check for zero effects: $\alpha_{[r_1,\ldots,r_M]} = 0$ for some or multiple $[r_1,\ldots,r_M]$.

 This approach makes use of the researchers' a priori knowledge of which policies can be pooled: these are policies that are dosage variations of the same treatment \emph{profile}, or underlying policy type. Therefore, it places discipline on the problem. It ensures we are not mis-naming pooled choices by pooling non-comparable policy bundles, which is the issue implicitly raised in \cite{muralidharan2019factorial}.\footnote{Specifically, \cite{muralidharan2019factorial} take issue with ``short models'' such that, for example, what is claimed as the effect of $(x,0,0)$---e.g.,  protease inhibitor on its own---actually also includes some of the effect of $(x,y,0)$---e.g.,  where a nucleoside is given in combination with the protease inhibitor. In this sense the treatment is ``mis-named''. In TVA, the policy $(x,y,0)$ is considered to be a categorically different treatment type from $(x,0,0)$ for $x,y>0$. More generally, the pooled policy names always unambiguously indicate which unique policy combinations are pooled together.
 }
 We assume that when there are non-zero marginal effects, those effect sizes are large enough---assuming away the local-to-zero range---so that we may actually discover and make inferences about the best policy combinations.\footnote{In practice, we show in Section \ref{subsec: sparisty robustness} through simulations that we may relax the local-to-zero assumption in several directions (we study five regimes) and still retain strong performance for this final objective.}$^{,}$\footnote{Our approach works when these assumptions to allow regularization are palatable. When the assumptions, or reasonable relaxations, cannot be justified, sequential testing, rather than simultaneous testing with post-selection processing, is unavoidable.}

Our goal is to  identify the support (the set of non-zero coefficients) of the regression equation \eqref{eq:simple-alpha}. Under our maintained assumptions, a natural way to do this is to use LASSO. This requires an extra step. The regressors
in equation  \eqref{eq:simple-alpha} are typically strongly correlated. For instance, $1\{T_1 = 2\} \cdot 1\{T_2 = 0\}$ implies $1\{T_1 \geq 1\}\cdot 1\{T_2 = 0\}$. This means that the marginal effects specification fails the necessary and sufficient condition for LASSO support consistency, of ``irrepresentability,'' which requires that regressors are not too correlated \citep{zhao2006model}. Thus, the second step is to apply the \emph{Puffer transformation} to the variables to which LASSO is being applied \citep{rohe2014note,jia2015preconditioning}. This de-correlates the design matrix that comes from \eqref{eq:simple-alpha}. We show that the
the specific structure of the RCT makes it particularly suitable for this technique.

Once LASSO has been applied on the Puffer-transformed variables to consistently estimate the marginal effects support, the third step is to reconstruct a set of unique policies taking into account the pooling and pruning implied by the LASSO results.

The fourth step is to  estimate OLS on the new set of unique policies, post-selection. Using an argument adapted from \cite{javanmard2013model},  we show that this estimator is consistent and asymptotically normally distributed.\footnote{The important detail here is that the convergence in distribution is not uniform (in the parameter space) \cite{leeb2005model}. Nevertheless, asymptotic normality holds pointwise (in the parameter space)---essentially, in our setting, the non-uniformity does not have much bite since incorrect selection of the high-effect policies happens with probability tending rapidly to zero.}

Finally, the fifth step is to estimate the effect of the best pooled and pruned policy, adjusting for the winner's curse \citep{andrews2019inference}. There are three advantages of conducting this adjustment post-TVA. First, when there are fewer potential alternatives to the best policy $k^\star$, the odds of picking the best by chance are lower. Second, with fewer alternatives, it is less likely that the second-best alternative has an effect that is similar to the $k^\star$ effect. Had it been more likely that the second-highest alternative had an effect similar to the highest effect, the adjustment for bias would require a more stringent shrinkage penalty. Our procedure avoids this.  Third, there is the benefit of coherence: if two policies are essentially the best and can  be pooled, there is no point in applying a strong shrinkage penalty.

We apply this framework to a large-scale immunization experiment that we conducted in Haryana, India, from December 2016 to November 2017. Immunization is recognized as one of the most effective and cost-effective ways to prevent illness, disability, and death. Yet nearly 20 million children do not receive critical immunizations each year \citep*{unicef2019progress}. Though resources directed towards immunization have increased steadily, there is mounting evidence that insufficient parental demand has contributed to stagnating immunization rates \citep*{who2019imm}.  This has motivated experimentation with ``nudges,'' including conventional ones such as small cash or in-kind incentives,\footnote{See  \cite{banerjee2010improving,bassani2013financial,wakadha2013feasibility,johri2015strategies,oyo2016interventions,gibson2017mobile}.}  SMS reminders,\footnote{See \cite{wakadha2013feasibility,domek2016sms,uddin2016use,regan2017randomized}.
} as well as more novel interventions such as symbolic social rewards\footnote{See \cite{karing2018social}.} or using influential individuals in a social network as ``ambassadors."\footnote{See \cite{alatas2019celebrities,banerjeeusing}.} We conducted this experiment in partnership with the Government of Haryana, which was interested in selecting the best policy for full scale adoption in the state.  We cross-randomized three arms with different nudges that had shown some promise in earlier work: (1) monetary incentives, (2) SMS reminders, and (3) seeding ambassadors. Incentives came in two types (linear and convex) with two dosages each (low and high in terms of value). SMS reminders had two dosages. Either 33\% or 66\% of caregivers received SMS reminders (and voicemails) about the next scheduled vaccination. Ambassadors where either randomly selected or selected through a nomination process. The nomination process was done in three ways, one of which came in two dosages (Information Hub).
All together, we have 75 unique policies and 915 villages were at risk for all three treatments taken together.

Applying TVA and the winner's curse adjustment, we find that the best policy is to use information hubs and either low or high SMS coverage, in combination with convex incentives that can be either low or high. This increases the number of immunizations by 44\% ($p < 0.05$). Choosing the cheapest among these suggests that the policymaker should chose low convex incentives, send SMS to 33\% of caregivers, and identify information hubs to relay the message. To maximize the number of immunization per dollar spent, the best policy is using information hubs along with SMS reminders at 33\% or more of caregivers covered. Adjusting for the winner's curse, it  increases the number of immunizations per dollar by 9.1\% ($p < 0.05$) compared to the status quo with no additional intervention. It is the only policy that strictly increases the number of immunizations per dollar spent.

The results highlight the importance of complementarities that may get lost had a factorial design not been used. Information hubs magnify the effect of other interventions and spark diffusion: neither incentives nor reminders are selected on their own, but are selected when combined with information hubs, which spark diffusion. Similarly, information hubs are not selected on their own, but are selected when combined with the conventional strategies. This suggests that in cases where there are no strong reasons to rule out interactions a priori,  it is important to accommodate them in the design and the statistical analysis.


\section{Treatment Variant Aggregation}\label{sec:Estimation}

This section describes the setup and proposes the methodology for treatment variant aggregation and for estimation of the impact of the best policy after adjusting for the winner's curse.

\subsection{Overview and Setup}

We have a randomized controlled trial of $M$ arms and $R$ ordered
dosages (\{none, intensity 1,..., intensity $R-1$\}). This yields $K:=R^{M}$ unique treatment combinations or \emph{unique policies}. Let $T_{ik}$
be a dummy variable indicating that unit $i$  is assigned to unique policy $k$. Unique policies are described as \emph{variants} of each other when they differ only in the (non-zero) dosages of the treatments applied.\footnote{This implies in particular that two unique policies differing only in whether some arm is active or inactive (dosage is zero) are not considered variants, as formalized below. }

Assuming the same number of dosages per arm is for notational ease and without substantive loss of generality. In practice the number of dosages per arm can vary. 

The \emph{unique policy regression} is given by
\begin{equation}
y_{i}=T\beta^{0}+\epsilon_{i}.\label{eq:unique_policy}
\end{equation}
The support of this regression is given by the set of unique policies
that have non-zero effect relative to control,
\[
S_{\beta}:=\left\{ k\in\left[K\right]:\ \left|\beta_{k}^{0}\right|\neq0\right\} .
\]

Some of the variants have equivalent effects and ought to be considered as one policy. Others may be altogether ineffective and ought to be pruned (i.e., pooled with control). We construct a method of \emph{treatment variant aggregation} (TVA) in order to \emph{pool and prune variants} systematically.

\subsubsection{Treatment Profiles and Policy Variants} 

A fundamental concept is a \emph{treatment profile}. This clarifies which unique policies are \emph{variants} of each other; i.e., could potentially be pooled with one another (without being pooled with the control).

The treatment profile $P(k)$ of a unique policy $k$ simply designates which of the $M$ arms are active (dosages are positive), without regards to how high the dosage is. Two unique policies $k,k'$ are mutual variants if and only if $P(k) = P(k')$, i.e. exactly the same arms are active for both policies. Thus, $K$ unique policies are categorized into $2^M$ treatment profiles.

\begin{example}
In the case of our experiment, consider observation $i$ that has an assigned policy $k=$(No Ambassador,
33\% SMS, low-value flat incentives) and observation $j$ that has an assigned assigned policy $k'=$(No Ambassador,
66\% SMS, low-value flat incentives). Though $k$ and $k'$ are distinct
treatment combinations, they share the same treatment profile--- $P\left(k\right)=P\left(k'\right)$---of
(No Ambassador, Some SMS, Some incentives). Therefore $k$ and $k'$ are variants.  They would not be variants if instead
$k'=$(No Ambassador,
66\% SMS, \emph{No incentives}).
\end{example}

\subsubsection{Treatment Variant Aggregation: Pooling and Pruning}

Increasing the dosage in a treatment arm may be inconsequential after a point, and more generally 
minor policy variants may have the same impact.
Here, we consider a re-specification of \eqref{eq:unique_policy} that explicitly tracks the marginal effect of increasing dosages by grouping together policy variants that have the same effect on the outcome, described below.\footnote{While sometimes what is ``dosage'' and ``dosage ordering'' is readily apparent from the arm, as in the SMS arm of our intervention with saturation levels 33\% and 66\%, in other cases the researcher has to decide this (of course this can be pre-specified). For example, in the seeds arm of our intervention, we decided that the information hub ambassador comes in two dosages, with those that are trusted for health advice as the higher dosage.} When these marginal effects are zero, this means that a set of variants are to be  either pooled or pruned.

Letting $Z_{S_{TVA}}$ denote the matrix of indicator variables for the pooled policies, our goal is to estimate the \emph{pruned and pooled policy regression}:

\begin{equation}
y=Z_{S_{TVA}}\eta^{0}_{S_{TVA}}+\epsilon\label{eq: pooled_pruned}.
\end{equation}

Comparing \eqref{eq: pooled_pruned} with \eqref{eq:unique_policy}, $\eta^0$ are the projection coefficients of $T \beta^0$ onto $Z_{S_{TVA}}$, that corresponds to simply grouping certain policies, and
estimating the parameters for the grouped policies. Notice that some policies may end up pooled with the control and we refer to this as pruning.

Let $\mathcal{P}$  denote the set of all partitions of the $K$ policies. Elements of $\mathcal{P}$ index every conceivable pooling designation of the $K$ policies. A specific pooling designation is given by $\Pi$, some specific partition of $K$ policies. Whether two given policies $k$ and $k'$ are pooled corresponds to whether they are members of the same part of the partition, $\pi \in \Pi$.

The universe of all conceivable poolings is disciplined by a set of rules $\Lambda$ governing the \emph{admissible} poolings. The admissible pools are then a strict subset $\mathcal{P}_{\mid \Lambda} \subset \mathcal{P}$. The target $S_{TVA} \in  \mathcal{P}_{\mid \Lambda}$ is defined to be the maximally admissible pooled and pruned set of policies (i.e., the coarsest partition). 

We now define the admissible pooling rules $\Lambda$. It is useful to represent policies in terms of its assigned intensities in each arm. Let us say $k:= [r_1,...,r_M]$ where $r_i$ is the intensity assigned to in arm $i$. There is a natural partial ordering of policies with respect to their intensities, whereby if $k := [r_1,...,r_M]$ and $k' := [s_1,...,s_M]$ in the intensity representation, then we say $k\geq k'$ if $\forall i \in [M], r_i \geq s_i $; i.e., the intensities of $k$ \emph{in each arm} weakly dominate those of $k'$. 
Finally, let $k_0$ denote the control or null policy, with intensity representation $[0,...,0]$.

\begin{ass}\label{assu: admissible} $\Pi$ is an admissible pooling respecting $\Lambda$---i.e., $\Pi \in \mathcal{P}_{\mid \Lambda}$---if and only if
\begin{enumerate}
\item  $k, k' \in \pi$ implies  $\beta^{0}_k = \beta^{0}_{k'}$. Only policies with equal treatment effects may be pooled.
\item  $k, k' \in \pi$ implies $P(k) = P(k')$ or $k_0 \in \pi$. Only variants may be pooled, or if non-variant policies are being pooled then they must also be pooled with the control (null policy).
\item  $k, k' \in\pi$  and $k \leq k^{\prime \prime} \leq k'$ (in the intensity ordering) imply that $k^{\prime \prime} \in \pi$.  Pooling can only involve contiguous variants, and cannot skip any intermediate policies.
\item Let $k, k' \in \pi$ have intensity representations $k := [r_1,...,r_i,...,r_M]$ and $k' := [r_1,...,r_i-1,...,r_M]$. Then, for any two $j, j'$ such that $P(j)=P(j')=P(k)=P(k')$, $j < k$, $j'< k'$ and with intensity representations $j:= [s_1,...,r_i,...,s_M]$, $j' := [s_1,...,r_i-1,...,s_M]$, it must be true that $j, j' \in \pi'$ (where $\pi'$ possibly equals $\pi$).\footnote{We use $j<k$ to indicate that $j\leq k$ and $j\neq k$.}
\end{enumerate}
\end{ass}

Assumption \ref{assu: admissible} (part 1) is the basic requirement on pooling.\footnote{Note that a finite sample pooling procedure admits a pool only if it cannot reject that estimated treatment effects are equal.} Assumption \ref{assu: admissible} (part 2) imposes some essential structure for interpretability. It ensures that we avoid pooling non-comparable bundles---a problem identified by \cite{muralidharan2019factorial}.  Limiting the pooling to variants cleaves the combinatorially vast universe of unstructured, conceivable pools to those that are more policy relevant. So pooling comparisons are only made within-profile, with the sole exception being when policies are pruned (pooled with control).

Assumption \ref{assu: admissible} (parts 3 and part 4) capture the implications running a regression of outcomes on profile dummies as well as marginal variables that capture main effects and complementarities as one climbs dosages, as in equation \eqref{eq:simple-alpha} in the Introduction and \eqref{eq: TVA} below. Both fall out of equation \eqref{eq:alpha_beta_rel} that relates coefficients between the unique policy and marginal effects regressions and tells one how to add up marginal effects to retrieve policy effects.

Conceptually, part 3 stems from the fact that marginals track whether \emph{incremental} dosages matter or not. Thus, contiguous variants are pooled so long as all incremental dosage effects in the relevant directions are zero. However, when an incremental dosage that tracks a complementarity between arms has nonzero effect, it cleaves apart policies not just ``locally'' but in an entire region of the treatment profile. Part 4 expresses this ``externality.'' The intuition and function of this condition is presented below in an example in Figure \ref{fig:hasse_concatenations}, panels B and E.

We can depict typical (non pruning) pooling choices for a treatment profile in a Hasse diagram. In a Hasse diagram, a line upwards from variants $k$ to $k'$ implies $k' > k$, and there is no variant $k''$ such that $k' > k'' > k$ (in the partial order). The running example is the case of a 2 arm treatment of 4 intensities (3 nonzero intensities ``low", ``medium", ``high"); i.e., $M= 2, R = 4$ and the treatment profile where both arms are ``on". Figure \ref{fig:hasse} depicts the Hasse for this treatment profile. Here unique policies are named per their intensity representations; i.e., $[r_1,r_2]$ where $r_i \in \{1,2,3\}$ is the (nonzero) dosage in arm $i$.

Pooling choices respecting $\Lambda$ for this treatment profile show up as particular ``concatenations" of adjacent policies on the original Hasse diagram. The top panels of Figure \ref{fig:hasse_concatenations} (panels A-C) depict key examples for this treatment profile. In panel A, where arm 2 is at low intensity, the exact intensity that arm 1 is set at is irrelevant for treatment effect. That is, $\beta_{[1,1]} = \beta_{[2,1]} = \beta_{[3,1]}$ and so $\{[1,1], [2,1],[3,1]\}$ can be concatenated to generate a policy we can call $[1:3,1]$.

In panel B, for low and medium intensities in both arms, increasing the intensity on arm 1 from 1 to 2 makes a difference but the intensity of arm 2 is irrelevant. That is $\beta_{[1,1]} = \beta_{[1,2]}$ and $\beta_{[2,1]} = \beta_{[2,2]}$ but $\beta_{[1,1]} \neq \beta_{[2,1]}$ and the concatenated pooled policies are $[1,1:2]$ and $[2,1:2]$.

In panel C, again for the low and medium intensities in the two arms, the exact intensities of both arms are irrelevant. That is, $\beta_{[1,1]} = \beta_{[1,2]} = \beta_{[2,1]} = \beta_{[2,2]}$ and the concatenated pooled policy is $[1:2,1:2]$.

We now illustrate the connection between these policy ``concatenations" respecting $\Lambda$ and marginal effects. For a treatment combination $k$, $\alpha^{0}_k$ is the marginal effect of the dosages in $k$ within its treatment profile relative to incrementally lower dosages. Formally, the marginal effects $\alpha^0$ may be defined implicitly so that a policy's effects are the sum of marginal effects from increasing dosages up to its particular dosage profile:

\begin{equation}\label{eq:alpha_beta_rel}
\beta_{k}^{0}=\sum_{k^{\prime} \leq k ; P\left(k^{\prime}\right)=P(k)} \alpha_{k^{\prime}}^{0}.
\end{equation}

Equation \eqref{eq:alpha_beta_rel} can be inverted to recover $\alpha^0$ in terms of $\beta^0$. An explicit expression for its terms $\alpha^{0}_k$ is more unwieldly in its full generality, but depending on the policy $k$, it can be a difference between two variants' effects or reflect a complementarity, i.e., the interaction effect from combining dosages in different arms. This is consistent with the interpretation that a policy's effects are the main effects of the highest dosages in each arm, considered separately, plus the relevant interaction effects.

We can now see how $\Lambda$-respecting policy concatenations in the Hasse show up as zeroes in the marginals $\alpha^0$ by interpreting Figure \ref{fig:hasse_concatenations},  panels D-F, in light of equation \eqref{eq:alpha_beta_rel}. In panel D, $\alpha_{[2,1]} = \alpha_{[3,1]} = 0$, i.e., keeping the intensity fixed as low in arm 2, there is no marginal contributions of increasing the intensity in arm 1. Since $\alpha_{[2,1]} = \beta_{[2,1]} - \beta_{[1,1]}$, and $\alpha_{[3,1]} = \beta_{[3,1]} - \beta_{[2,1]}$, the concatenations on panel A ensue. 

In panel E, since $\alpha_{[2,2]} = (\beta_{[2,2]} - \beta_{[2,1]}) - (\beta_{[1,2]} - \beta_{[1,1]})$, it reflects a `complementarity' in increasing the dosages from low to medium in both arms. Since this is zero, it makes arm 2's relevance parallel in for both low and medium intensities of arm 1; i.e., $\beta_{[2,2]} - \beta_{[2,1]}  = \beta_{[1,2]} - \beta_{[1,1]}$. Since arm 2 is irrelevant when arm 1 is low---because $\alpha_{[1,2]} = 0$---it follows that  $\beta_{[2,2]} = \beta_{[2,1]}$ and $\beta_{[1,2]} = \beta_{[1,1]}$, permitting the concatenations in the above panel B. No further concatenations ensue since $\alpha_{[2,1]} \neq 0 \implies \beta_{[2,1]}  \neq \beta_{[1,1]}$.    This illustrates the role of Assumption \ref{assu: admissible} (part 4), since the creation of pooled policy $[2, 1:2]$ also implies the creation of policy $[1, 1:2]$. Indeed, there is a simple depiction of Assumption \ref{assu: admissible} (part 4) in the Hasse diagrams for treatment profiles: for any parallelogram that one can draw in the Hasse, if the ``top'' segment is pooled together, so must the ``bottom'' segment (in a possibly distinct pooled policy).

On the other hand, in panel F, the small change of setting the marginal $\alpha_{[2,1]} = 0$ concatenates both low and medium intensities in both arms.

As illustrated through these examples, zeros in $\alpha^0$ thus show up as policy concatenations in Hasse diagrams, and more formally, as pooling decisions respecting $\Lambda$. This motivates the \emph{marginal effects regression}:

\begin{equation}
y=X\alpha^{0}+\epsilon.\label{eq: TVA}
\end{equation}

This is an invertible transformation of \eqref{eq:unique_policy}. $X$ can be interpreted as indicators

\[
X_{i\ell}:={\bf 1}\left\{ k\left(i\right)\geq\ell\ \cap\ P\left(k\left(i\right)\right)=P\left(\ell\right)\right\} .
\]
In other words, $X$ assigns for unit $i$ a ``1'' for all policy variants
that share $k\left(i\right)$'s treatment profile and are weakly dominated
in intensity by $k\left(i\right)$ and a ``0'' otherwise.

Estimating \eqref{eq: TVA} is informative of the zeros in $\alpha^0$. Knowledge of the zeros is of course equivalent to the knowledge of its complement, the support of \eqref{eq: TVA}:

\[
S_{\alpha}:=\left\{ j\in\left[K\right]:\ \left|\alpha_{j}^{0}\right|\neq0\right\} .
\]

The idea is to apply a model selection procedure to estimate $S_{\alpha}$. In Appendix \ref{sec:pooling}, we show how to construct the unique maximally pooled and pruned set $S_{TVA} \in \mathcal{P}_{\mid \Lambda}$ from $S_{\alpha}$.\footnote{Following this same procedure with any estimate $\widehat{S}_\alpha$ leads to an estimate $\widehat{S}_{TVA}$ of pooled and pruned policies.} The maximality ensures that no contiguous set of intensities thought to have the same treatment effects are left un-pooled.

In sum, we are interested in the following \emph{treatment variant aggregation} and best-policy effect estimation procedures:
\begin{enumerate}
\item Consistently select support: $\Pr\left(\widehat{S}_{\alpha}=S_{\alpha}\right)\rightarrow1$.
\item Consistently estimate effects after pooling and pruning:  $\left\Vert \hat{\eta}_{\widehat{S}_{TVA}}-\eta_{S_{TVA}}^{0}\right\Vert _{\infty}\rightarrow_{p}0$.\footnote{We write  (2) somewhat informally here because (1) happens with high probability (exponentially in $n$)  tending to one as will be formalized in detail in Proposition \ref{prop: TVA_consistency}}
\item Find the best policy $\hat\kappa^{\star}=\text{argmax}_{\kappa \in \widehat{S}_{TVA}}\hat{\eta}_{\widehat{S}_{TVA},\kappa}$

and generate an unbiased estimate of best policy effect $\eta^{0}_{\hat\kappa^{\star}}$.
\end{enumerate}

Our object of interest in (3), following \cite{andrews2019inference}, is the \textit{estimated} best policy effect rather than the \textit{true} best policy effect. As a consequence the impossibility result of \cite{hirano2012impossibility} on inference does not apply in our setting.\footnote{The reason their impossibility result arises is because the max-operator, necessary to retrieve the \textit{true} best policy effect, is nondifferentiable.  For the \textit{estimated} best mean, however, we do not need to apply the max-operator after conditioning on $\hat\kappa^\star$, hence usual inference applies.}

Note the two levels of estimation in $\hat{\eta}_{\widehat{S}_{TVA}}$. The random estimand $\eta_{\widehat{S}_{TVA}}^{0}$ is defined even when  $\widehat{S}_{TVA}$ is misspecified.\footnote{This estimand is obtained from the orthogonal projection of $T \beta^0$ onto the subspace generated by $\widehat{S}_{TVA}$.} $\hat{\eta}_{\widehat{S}_{TVA}}$ are estimated coefficients in the (possibly misspecified) model. The best policy effect estimand $\eta^{0}_{\widehat{S}_{TVA},\hat\kappa^{\star}}$ is an element of $\eta_{\widehat{S}_{TVA}}^{0}$.

\subsection{Pooling and Pruning for Support Selection}

The next step is to identify the support $S_{\alpha}$. The natural
place to start would be to apply LASSO directly to (\ref{eq: TVA}). However,
this approach fails sign consistency because
the marginal effects matrix $X$ fails an ``irrepresentability criterion'' which
is necessary for consistent estimation \citep{zhao2006model}. Irrepresentability
bounds the acceptable correlations in the design matrix. Intuitively,
it requires that variables that are not in the
support are not too strongly correlated with those that are. Otherwise, an irrelevant
variable is ``representable'' by relevant variables, which makes
LASSO erroneously select it with non-zero probability irrespective
of sample size. We show with a proof by construction that irrepresentability
fails in our setup in Appendix \ref{sec:simulations}. The idea is higher dosage variants will end up
being representable by lower dosage treatment combinations, violating
the requirement.

A way out is provided by \cite{jia2015preconditioning}. They show that, under some conditions,
 one can estimate the LASSO support by transforming
the data to recover irrepresentability. They demonstrate
that a simple left-multiplication (pre-conditioning) can de-correlate
the data (at the expense of inflating variance in the error).

In Proposition \ref{prop: support} we demonstrate that in the specific instance of the crossed
RCT design with ordered intensities, the pre-conditioning strategy of
 \cite{jia2015preconditioning} can be applied because the relevant sufficient conditions are met. Specifically, with an RCT, we can exactly characterize the design matrix and therefore the inflation factor. We can
 show that the variance inflation cost is tolerable, in the sense that we can consistently recover the support and the treatment effects.

The weighting is constructed as follows. Let us take the singular
value decomposition of $X:=UDV'$ where $U$ is an $n\times K$ unitary
matrix, $D$ is a $K\times K$ diagonal matrix of singular values,
and $V$ is a $K\times K$ unitary matrix. The \emph{Puffer transformation}---so
named for the fish whose shape is suggested by the geometry of this transformation---is $F:=UD^{-1}U'$. The regression of interest
is now
\begin{equation}
Fy=FX\alpha+F\epsilon\label{eq: puffer}
\end{equation}
where $F\epsilon\sim\mathcal{N}\left(0,UD^{-1}\Sigma D^{-1}U'\right)$.
As  \cite{jia2015preconditioning} note, this satisfies irrepresentability since
$\left(FX\right)'\left(FX\right)=I$, which is sufficient (\cite{jia2015preconditioning}, \cite{bickel2009simultaneous}).

To understand why this works, recall that the matrices $U$ and $V'$
can be thought of as rotations and $D$ as a rescaling of the principal
components. So, the transformation $F$ preserves the rotational elements
of $X$ without the rescaling by $D$ and $FX=UV'$ as its singular
value decomposition (with singular values of 1).

The reason this is useful is because when a matrix $X$ has correlation,
then the $i$th singular value of $X$ captures the residual variance
of $X$ explained by the $i$th principal component after partialling
out the variance explained by the first $i-1$ principal components.
So, when there is high correlation within $X$, less than $K$ principal
components effectively explain the variation in $X$ and so the later
(and therefore lower) singular values shrink to zero. $F$ inflates
the lowest singular values of $X$ so that each of the principal components
of the transformed $FX$ explains the variance in $FX$ equally. In
that sense, $FX$ is de-correlated and, for $K<n$, is mechanically
irrepresentable. The cost is that this effective re-weighting of the
data also amplifies the noise associated with the observations that
would have had the lowest original singular values.  Of course if the amplification is too strong,
it can can hinder efficiency of LASSO in finite sample and even prevent the sign consistency of  LASSO, in the worst case.\footnote{In $K>n$ cases--not  studied  here and not having a full characterization in the literature--even irrepresentability is not immediate and the theory developed is only for special cases (a uniform distribution on the  Stiefel manifold) and a collection of empirically relevant simulations \citep{jia2015preconditioning}.}

The reason why LASSO is particularly amenable to the Puffer transformation
in our specific setting of the cross-randomized experiment with varying
dosages is that the marginal effects design matrices are highly structured.
In particular, the assignment probabilities to the various unique
treatments are given, and as a result, the correlations with $X$
are bounded away from 1. This has the implication that the minimum
singular value is bounded below so that under standard assumptions
on data generation, LASSO selection is sign consistent. While this
is guaranteed for a sample size that grows in fixed $K$, the more
important test is whether it works when $K$ goes up with $n$; we
need to show that the Puffer transformation does not destroy the sign
consistency of LASSO selection as the minimal singular value of $X$
goes to zero as a function of $K$. In Lemma  \ref{lem:cmin}, we bound the rate
at which the minimal singular value of $X$ can go to zero as a function
of $K$ in a crossed RCT such as ours and Proposition \ref{prop: TVA_consistency} below relies
on this lemma to then prove that the Puffer transformation ensures
irrepresentability and consistent estimation by LASSO in our context.

We make the following additional assumptions and discuss their restrictiveness below.
\begin{ass}[RCT design gowth]
\label{assu: design}$R\geq3$, $K<n$, and $K=O\left(n^{\gamma}\right)$
for some $\gamma\in[0,1/2)$.\footnote{This assumption can be slightly weakened to $K = o(n^{\frac{1}{2}})$ for the results up to support consistency (Proposition \ref{prop: support}) to hold. However, one would then need a further assumption that $K^2 \log(K) = o(n)$ for  post-LASSO inference under a normal distribution (Proposition \ref{prop:normalinf}) to hold. The intuition is that the growth rate of $K$ must be tempered for the central limit theorem to operate in this growing parameter regime. Our Assumption \ref{assu: design} satisfies this requirement automatically. }
\end{ass}

\begin{ass}[Minimal marginal effect size]
\label{assu: support}$\left|S_{\alpha}\right|<K$ and $\min_{k\in S_{\alpha}}\left|\alpha_{k}\right|>c>0$
for $c$ fixed in $n$.
\end{ass}

\begin{ass}[Homoskedasticity]\label{assu: noise}$\epsilon_i\overset{\mathrm{iid}}{\sim}\mathcal{N}\left(0,\sigma^2 I_n\right)$, with $\sigma^2 > 0$ fixed in $n$.


\end{ass}

\begin{ass}[Penalty sequence]
\label{assu: penalty}Take a sequence $\lambda_{n}\geq0$ such that
$\lambda_{n}\rightarrow0$ and $\lambda_{n}^{2}n^{1-2\gamma}=\omega\left(\log n\right)$.
\end{ass}
\
Assumption \ref{assu: design} restricts the growth of the problem, preventing settings with too many treatments relative to  observations. Without this assumption, the correct support may not be estimated with probability tending to one, and the post-estimators may not necessarily be asymptotitcally normally distributed.  In practice, it means that the RCT cannot have cells in the fully saturated treatment design with very few units assigned to that unique treatment combination. Assumption \ref{assu: support} is the conventional LASSO-sparsity assumption applied to the marginal effects formulation. It imposes that \textit{adjacent policy variants} are either appreciably different or have no difference (i.e., the so-called ``beta-min'' assumption in the literature). We do not handle the case of local alternatives among adjacent variants — i.e., very small yet non-zero differences,  but policies that are not variants of each other or are nowhere adjacent are allowed to be local alternatives as discussed in Section \ref{subsub: wc_adjustment}. Assumption \ref{assu: noise} places our theory under homoskedastic errors following the literature on Puffer transformation. Extension to heteroskedasticity is left for future work.  Finally, Assumption \ref{assu: penalty} imposes a restriction on the LASSO-penalties, standard in the regularization literature.

\begin{prop}
\label{prop: support}Assume \ref{assu: admissible}-\ref{assu: penalty}. Let $\widetilde{\alpha}$
be the estimator of (\ref{eq: puffer}) by LASSO:
\[
\widetilde{\alpha}:=\text{argmin}_{a\in\mathbb{R}^{K}}\left\Vert Fy-FXa\right\Vert _{2}^{2}+\lambda_{n}\left\Vert a\right\Vert _{1}.
\]
Then $\Pr\left(\mathrm{sign}\left(\widetilde{\alpha}\right)=\mathrm{sign}\left(\alpha^{0}\right)\right)\rightarrow1$.
\end{prop}

All proofs are in Appendix \ref{sec:proofs} unless otherwise noted.

\subsection{Consistency of the TVA Estimator}

Having constructed an estimator  $\widehat{S}_{\alpha}$ of the support $S_\alpha$,  the next step is to use Algorithm \ref{alg:pooling} in Section \ref{sec:simulations} to construct $\widehat{S}_{TVA}$, the estimated set of pooled and pruned unique policies, and  then estimate policy  effects.  The regression of interest is (\ref{eq: pooled_pruned}). We show this estimator is consistent.\footnote{We thank Adel Javanmard for a helpful discussion of the proof.}

\begin{prop}\label{prop: TVA_consistency}
Assume \ref{assu: admissible}-\ref{assu: penalty}. Let $\hat{\eta}_{\widehat{S}_{TVA}}$ be the post-Puffer
LASSO   OLS estimator of (\ref{eq: pooled_pruned})
on support $\widehat{S}_{TVA}$. Then, with probability at least $1-2e^{-\frac{n^{1-2\gamma}\lambda^2}{2\sigma^{2}}+\gamma\log n}\rightarrow1$,
\[
\left\Vert \hat{\eta}_{\widehat{S}_{TVA}}-\eta_{S_{TVA}}^{0}\right\Vert _{\infty}\leq C\cdot\sqrt{\frac{\log n}{n^{1-\gamma/2}}}
\]
for any $C>0$ fixed in $n$.
\end{prop}

\subsection{The Effect of the Best Policy}

Another policy relevant issue is the recommendation of a ``best policy''
together with an estimate of the effect of the best policy. Intuitively, to select
the best policy, we can scan the post-LASSO estimates of policies
in $\widehat{S}_{TVA}$. As \cite{andrews2019inference} note, the maximum of the set of estimated policies is subject to the winner's curse. In order to correct this, \cite{andrews2019inference} make use of the asymptotic normality of the estimators in question that are to be compared. So, we proceed in two steps. First, we demonstrate that $\hat{\eta}_{\widehat{S}_{TVA}}$ is indeed asymptotically normally distributed. This allows us to use the winner's curse procedure, that we then apply as a second step.

\subsubsection{Asymptotic Normality}

The post-Puffer LASSO estimators are asymptotically normally distributed (pointwise) for the following reason. If the correct support, $S_{TVA}$ were always selected, mechanically the estimators are asymptotically normal.

So, in practice, we need to worry about two errors: (a) the asymptotic distribution of the estimator with some incorrect support being selected and (b) the asymptotic distribution of the true estimator when the incorrect support is selected. We show in Appendix \ref{sec:proofs} that both of these terms are small in our setup.

Intuitively, the second term can be ignored. After all, the true estimator itself is asymptotically normally distributed, so given the very unlikely event of incorrect selection, this term is asymptotically negligible. The first term requires more work. But again, one can show that the amount of potential bias accumulated due to selecting the wrong support is slow relative to the rate of actually estimating the wrong support.\footnote{An entirely different approach would be to use a recent focus in the literature on exact post-selection inference using the observation that the LASSO procedure to select a model generates a polyhedral conditioning set \cite{lee2016exact}. This generates a parameter estimator distribution that is a truncated, rather than complete, normal. In our special environment---a correctly specified linear model, sparse parameters, restrictions on shrinkage rate of minimal values of parameters on the support---the truncation points diverge when conditioning on the event that the true model is the estimated model. In the winner's curse context an analogous point is made in \cite{andrews2019inference}, Proposition 3. This means that the distribution returns to the usual Gaussian. However, we provide a simpler, direct argument where we can calculate the distribution when the correct support is selected and bound the problematic terms in the event of poor selection.}  Therefore, we have the following result.\footnote{We again thank Adel Javanmard for a helpful discussion of the proof.}

\begin{prop}\label{prop:normalinf}
Assume  \ref{assu: design},  \ref{assu: support}, \ref{assu: noise},
 and  \ref{assu: penalty}. Let $H := \plim Z_{S_{TVA}}'Z_{S_{TVA}}/n$ and $J := \plim Z_{S_{TVA}}' \Sigma Z_{S_{TVA}}/n$ exist. Then,  $\hat{\eta}_{\widehat{S}_{TVA}}$, the post-Puffer LASSO selection OLS estimator of (\ref{eq: pooled_pruned})
performed on support $\widehat{S}_{TVA}$, is asymptotically distributed
\[
\sqrt{n}\left(\hat{\eta}_{\widehat{S}_{TVA}}-\eta_{\widehat{S}_{TVA}}^{0}\right)\rightsquigarrow\mathcal{N}\left(0,H^{-1}JH^{-1}\right).
\]
\end{prop}

It is well-known that one cannot uniformly (over the parameter space) build post-selection asymptotic distributions  \citep{leeb2005model,leeb2008sparse}. This is the subject of much discussion of a larger literature on post-selection inference---interpretations of the post-estimation procedures and practical function \citep{berk2013valid,tibshirani2016exact,lee2016exact}. In our context, several remarks are worth making. First, our  claim is about pointwise inference, not uniformity over the parameter space. Second, we have nothing to say conditional on incorrect selection, hence the non-uniformity. Still, no matter what model is selected---even if an incorrect one---since in our setting the regressors are always orthogonal, there is some valid post-selection interpretation in the sense of \cite{berk2013valid}, but we do not characterize what occurs in the vanishing probability events. Third, as we recover the support with probability tending to one, and at an exponential rate, in a practical sense the non-uniformity occurs only for very small local alternatives.\footnote{We are grateful to Adam McCloskey for pointing this out.} Loosely, recall that the non-uniformity comes up when the probability of correct selection does not go to one, or along the sequence is local to the event of failed selection. Given the very high rate of correct selection (tending to one exponentially fast in $n$), these local alternatives must be exceedingly close to the true parameter (the sequence of alternatives converging to the true parameter at very fast rate in $n$). See analogous discussion in \cite{mccloskey2020asymptotically} and the discussion of (A.1) in that paper.

Indeed, consistent with the theoretical results, as we will show in Section \ref{subsec: properties_TVA}, the estimators look normal in practice indicating that the non-uniformity concerns are likely to not be large in at least many practical cases, in our specific setting. Further, there is an interesting subtlety that arises in our case in particular. In our  setting which concerns best policies, since the elements with the highest effects tend to be selected first, and because of orthogonality, in practice the large parameter estimates almost always perform well.

\subsubsection{Adjusting for the Winner's Curse}\label{subsub: wc_adjustment}

The next step is to apply the result of \cite{andrews2019inference}. We leverage
the results of their Proposition 9.\footnote{We thank Isaiah Andrews and Adam McCloskey for helpful discussions.}

We are interested in picking the best estimated policy:
\[
\hat\kappa^{\star}=\text{argmax}_{\kappa \in \widehat{S}_{TVA}}\hat{\eta}_{\widehat{S}_{TVA},\kappa}
\]
With this we have the estimate of the best estimated policy's effect $\hat{\eta}_{\widehat{S}_{TVA},\hat \kappa^\star}$. \cite{andrews2019inference} demonstrate that this naive estimate will be biased and  how to adjust for it.

Since all estimated policy effects, $\hat{\eta}_{\kappa}$ for $\kappa \in \widehat{S}_{TVA}$, are regression coefficients and therefore random variables, there is some chance in any sample that $\hat{\eta}_{\widehat{S}_{TVA},\kappa} > \hat{\eta}_{\widehat{S}_{TVA},\kappa'}$ even though $\eta^0_{\widehat{S}_{TVA},\kappa} < \eta^0_{\widehat{S}_{TVA},\kappa'}$ and therefore the ordinal ranking would be incorrect. For a fixed amount of noise in the regression model \eqref{eq: pooled_pruned}, if the policies were better separated --- so $\lvert \eta^0_{\widehat{S}_{TVA},\kappa} - \eta^0_{\widehat{S}_{TVA},\kappa'}\rvert $ were larger --- then the probability of incorrectly ordinally ranking the two policies $\kappa$ and $\kappa'$ would be smaller.

As we are doing asymptotic analysis and not exact finite-sample analysis, the asymptotic version of this finite-sample problem is when policy
effects differ by order $1/\sqrt{n}$ . That is, consider two policy
effects
\[
\eta_{S_{TVA},\kappa}^{0}=\eta_{S_{TVA},\kappa'}^{0}+\frac{r_{\kappa \kappa'}}{\sqrt{n}}
\]
with $\eta^0_{S_{TVA},\kappa'}$ well-separated from 0 and $r_{\kappa \kappa'}$ a constant
fixed in $n$. In this case, the asymptotic distribution of the difference
between the estimates of these two policy effects will be non-vanishing,
capturing the finite-sample problem.

In this setup, despite the fact that the two policies $\kappa$
and $\kappa'$ have effects that differ by a $\Theta\left(1/\sqrt{n}\right)$ term, they
are still well-separated from 0 and therefore satisfy the minimum
amplitude condition (Assumption \ref{assu: support}, relative to control). To see this, observe that if $\eta_{S_{TVA},\kappa'}^0 > c > 0$ then for at most all but finitely many $n$, $\eta_{S_{TVA},\kappa}^0 > c$ as well.

To satisfy our Assumption \ref{assu: support}, such policies $\kappa$, $\kappa'$ must either not be policy variants (i.e., have different treatment profiles), or, if they do have the same profile, not be too similar
in terms of dosage. If not, the difference between them, were they to cause the finite sample problem necessitating a winner's curse adjustment, would also contradict Assumption \ref{assu: support} and therefore affect the  LASSO estimation. This is formalized in the following lemma.

\begin{lemma}\label{lem: localalter}
Assume \ref{assu: admissible} and suppose $\kappa, \kappa' \in S_{TVA}$ are local alternatives, i.e., $\eta_{S_{TVA},\kappa}^{0}=\eta_{S_{TVA},\kappa'}^{0}+\frac{r_{\kappa \kappa'}}{\sqrt{n}}$ for some $r_{\kappa \kappa'}$ fixed in $n$. Then, for $\alpha^0$ to respect Assumption \ref{assu: support}, one of the following has to hold:

\begin{enumerate}
\item $P(\kappa) \neq P(\kappa')$ (i.e., $\kappa$ and $\kappa'$ have different treatment profiles \footnote{This is a slight abuse of notation, since $P(\cdot)$ was defined originally over treatment combinations, not pooled policies. So, $P(\cdot)$ here is the simply the well defined extension to the latter.}), or
\item $P(\kappa) = P(\kappa')$, and $\kappa$ and $\kappa'$ are nowhere adjacent in the Hasse diagram.\footnote{Formally, this is the condition that for any treatment combinations $k, k'$ such that $\kappa$ pools $k$ and $\kappa'$ pools $k'$, either $k$ and $k'$ are incomparable or there is a treatment combination $z$ pooled by neither $\kappa$ nor $\kappa'$ such that $\min\{k,k'\} < z < \max\{k,k'\}$.}
\end{enumerate}
These conditions are also sufficient to allow for local alternatives, in that for any $S_{TVA} \in \mathcal{P}_{\vert \Lambda}$, all pairs $\kappa, \kappa' \in S_{TVA}$ meeting condition (1) or (2) can be simultaneously made local alternatives by some choice of $\alpha^0$ satisfying Assumption \ref{assu: support}.
\end{lemma}

\begin{ass}
There are at most $q<\infty$, independent of $n$, pairs of local alternatives, i.e., pairs $\kappa, \kappa' \in S_{TVA}$ such that $\eta_{S_{TVA},\kappa }^{0}=\eta_{S_{TVA}, \kappa'}^{0}+\frac{r_{\kappa \kappa'}}{\sqrt{n}}$ where $r_{\kappa \kappa'}$ fixed in $n$. $\alpha^0$ respects Assumption \ref{assu: support}.\label{assu: finite_support}
\end{ass}

Under this assumption, there are still pairs of policies
that are distinct but may have similar effects. Since they might be selected as the first and second best in the data, may therefore exhibit finite-sample
bias due to the winner's curse.

Now define $X\left(\kappa\right):=\hat{\eta}_{\widehat{S}_{TVA},\kappa}$. From the arguments presented in the previous part of the paper $X\rightsquigarrow\mathcal{N}\left(\mu,\Omega\right)$, where $\mu$ can be written in terms of the mean of a reference policy $\kappa_0$, i.e., $\mu=  \eta^{0}_{S_{TVA},\kappa_0}+ c + \frac{r_{\kappa \kappa_0}}{\sqrt{n}}$, where $c$ is some constant of separation and $r_{\kappa \kappa_0}$ is either $0$ or a constant depending on whether $\kappa$ and $\kappa_0$ are local alternatives or not. \\

Before continuing, we adapt our setting to that of \cite{andrews2019inference}. The main text of the paper focuses on the case where the estimators in question, $X(\kappa)$ are \emph{exactly} jointly normally distributed. While two extensions are presented, one for a conditioning event such as model-selection (addressed in their Appendix A) and another for the case of asymptotic normality which is required for practical settings  such as regression (in their Appendix D), the paper does not formally work out the case with both issues present.

We have both in our setting, under Assumptions \ref{assu: admissible}-\ref{assu: finite_support}. We have a conditioning event ($\widehat{S}_{TVA} = S_{TVA}$) occurring with probability tending to one and we are in a regression setting with asymptotic normality. So while the theoretical properties of the estimator in our setting are highly plausible and coherent with our simulations below,  the extension to the nested case of model selection remains to be proven.  It is beyond the scope of our present paper here to nest both of their extensions, and we leave it for future work.

So, to continue, we assume that the distribution is exact. This allows us to focus on how local alternatives in our Hasse diagram  may impact the problem. Therefore, we assume
\[
X\sim\mathcal{N}\left(\mu,\Omega\right).
\]
Then we can {exactly} apply the results of Proposition 6 in \cite{andrews2019inference}, and
build the hybrid confidence set. In particular, we can construct the hybrid estimator based on chosen significance levels $\alpha, \beta$ such that $\beta < \alpha$. The hybrid estimator $\hat{\eta}^{hyb}_{\widehat{S}_{TVA},\hat \kappa^\star}$ is approximately median unbiased (with absolute median bias bounded by $\frac{\beta}{2}$)  with confidence intervals with coverage $\frac{1-\alpha}{1-\beta}$, conditional on $\eta^0_{\widehat{S}_{TVA},\hat \kappa}$ falling within a simultaneous confidence interval (with respect to the post-LASSO estimates $\hat{\eta}^0$) of level $1-\beta$.  A summary of the overall procedure is presented in Algorithm \ref{alg:TVA_overall}.

\begin{algorithm}[!h]
	\SetAlgoLined
	\SetKwInOut{Input}{input} \SetKwInOut{Output}{output}
	\begin{enumerate}
		\item Given treatment assignment matrix $T$, calculate the  treatment profile and marginal dosage intensity matrix $X$.
		\item Estimate $\widehat{S}_{\alpha}:=\left\{ j\in\left[K\right]:\ \left|\widetilde{\alpha}_{j}\right|\neq0\right\}$ by estimating  \eqref{eq: TVA} through a Puffer transformed LASSO.
		\item Calculate marginal effects support $\widehat{S}_{TVA}$ from $\widehat S_\alpha$ using Algorithm \ref{alg:pooling} in Appendix \ref{sec:pooling}.
		\item Estimate pooled and pruned treatment effects of unique (relevant) policies, $\hat{\eta}_{\widehat{S}_{TVA}}$, using regression \eqref{eq: pooled_pruned}.
		\item To estimate the best policy in $\widehat{S}_{TVA}$, select $\hat{\kappa}^\star = \argmax_{\kappa \in \widehat{S}_{TVA}} \hat{\eta}_{\widehat{S}_{TVA},\kappa}$ , construct the hybrid estimator  $\hat \eta_{\widehat{S}_{TVA},\hat \kappa^\star}^{\text{hyb}}$ with nominal size $\alpha$ and median bias tolerance $\beta/2$.
	\end{enumerate}
	\caption{Estimating Treatment Effects by Treatment Variant Aggregation}
	\label{alg:TVA_overall}
\end{algorithm}

\

Under the setting described above, it is important to note that even when wrong policies are selected the error is negligible.

\begin{rem} Conditional on correct support selection $\widehat{S}_{TVA} = S_{TVA}$, when there are local alternatives near the true best policy, the wrong policy may be selected with non-vanishing probability. (When there are no local alternatives around the true best policy, then with high probability the correct winner is always selected.) But the post-selection estimates of this wrong policy are $\Theta(1/\sqrt{n})$ different from the true best-policy effect.
\end{rem}


\section{Simulation Performance}\label{sec: simulations}

Here we run simulations in the environment described in Section \ref{sec:Estimation} -- namely, when a sparse set of policies have meaningful and meaningfully different impacts. In Section \ref{sec: performance_TVA} we show that TVA has strong performance in this environment.  In Section \ref{subsec: alternatives} we show it outperforms several other standard approaches; the relative deficiencies of other estimators also highlight the features that give TVA its edge. Finally, in Section \ref{subsec: sparisty robustness} we relax sparsity and lower bounded marginal effect sizes in the environment and show that TVA is still a strong candidate for these settings.

\subsection{Performance of TVA}\label{sec: performance_TVA}

Below, we describe the simulation setup and performance indicators. The idea is to generate simulated design matrices from marginal specifications (\ref{eq: TVA}) that resemble the data, score these on certain metrics, and aggregate these scores into measures of performance for sample size $n$.

\subsubsection{Simulation Setup}\label{sec: simulation_setup}

Throughout our simulations we use the following common setup:
\begin{enumerate}
\item Fix $R = (5,5,3), M = 3$ and $\sigma:= \sqrt{var(\epsilon)} = 2.3$: parameters are chosen to loosely mimic our experiment where $3$ treatment arms have asymmetric intensities leading to $75$ unique policies and where $\sigma$ is chosen such that the $R^2$ of the post-LASSO regression matches the experiment for a similar sample size.

\item The simulation results are plots of performance $\hat{m}(n)$ against sample size $n$ where $n$ ranges between 1,000 and 10,000.\footnote{For some computationally intensive simulations $n$ is logarithmically spaced.}
\item These scores $\hat{m}(n)$ are generically computed as follows.
\begin{enumerate}
\item  A set $\mathcal{C}$ of true supports of the marginal specification \eqref{eq: TVA} is randomly chosen. Each member $S_\alpha^i\in \mathcal{C}$ is a particular support or ``configuration" and each configuration has fixed support size $|S_\alpha^i| = M$. Specifically, each configuration $\mathcal{C}$ is constructed by randomly sampling $M$ covariates of $X$. Furthermore, if $S_\alpha^i =(k_1, k_2,..,k_M)$ in some given order, we assign coefficients $\alpha_{k_j} = 1 + 4\cdot\frac{j-1}{M-1}$. That is, these nonzero coefficients are linearly spaced between $1$ and $5$. Thus each configuration fully specifies the set of coefficients $\alpha$ for \eqref{eq: TVA}.
\item For each $S_\alpha^i \in \mathcal{C}$, a set $\mathcal{S}_{S_\alpha^i} (n)$ of simulations (design matrices) is generated based on the coefficients specified by the configuration, and the Gaussian noise, with sample size $n$. For each simulation $\hat{s}(n) \in \mathcal{S}_{S_\alpha^i}(n)$, it is scored by a metric $m(\hat{s}(n))$ that will be specified.
\item These scores are aggregated over simulations $\mathcal{S}_{S_\alpha^i}(n)$, and then aggregated again over configurations $\mathcal{C}$, to produce the aggregated performance score $\hat{m}(n)$.
\end{enumerate}
\end{enumerate}

\subsubsection{Performance Measures}\label{sec: performance_measures}

Denoting by $\widehat{S}^i_{\alpha} (\hat{s}(n))$ the model selection estimator for $S_\alpha^i$ for simulation $\hat{s}(n)$, we use the following performance metrics throughout our simulations:\\
\\
\textbf{Support selection accuracy:}
\[m(\hat{s}(n)) := \frac{|\widehat{S}^i_{\alpha} (\hat{s}(n)) \cap S^i_{\alpha}|}{|\widehat{S}^i_{\alpha} (\hat{s}(n)) \cup S^i_{\alpha}|}.\]
This is a value between $0$ and $1$ that increases with support coverage, and is $1$ if and only if the support is correctly selected.  To construct the  aggregated metric $\hat{m}(n)$ it is averaged over the simulations per configuration, and then averaged again over configurations.\\
\\
\textbf{``Some" best policy inclusion accuracy:}
\[m(\hat{s}(n))= \begin{cases}
1 & \text{ if } \hat \kappa^{i\star}(\hat{s}(n)) \cap \kappa^{i\star} \neq \emptyset \\
0 & \text{ otherwise }
\end{cases}\]
where $\kappa^{i\star}=\text{argmax}_{\kappa\in S_{TVA}^i}\eta_{S_{TVA},\kappa}$ denotes the true best pooled policy in the marginal effects support $S_{TVA}^i$ (uniquely determined from $S_\alpha^i$). This measure is again averaged over simulations per configuration, and then averaged over configurations. The final metric therefore gives the share of simulations per $n$ where at least one true best policy was pooled into the estimated best pooled policy.\\
\\
\textbf{Minimum dosage best policy inclusion accuracy:}
\[m(\hat{s}(n))= \begin{cases}
1 & \text{ if } k^{i\star\text{min}} \in \hat \kappa^{i \star}(\hat{s}(n)) \\
0 & \text{ otherwise }
\end{cases}\]
where $k^{i\star\text{min}}$ denotes the true minimum dosage best policy.\footnote{The ``intersection'' and ``inclusion'' operators for the best policy inclusion measures are to be understood in the following way: suppose the true best policy $\kappa^{i\star} \in S_{TVA}^i$ pools together $m$ policies as per $S_\alpha^i$ that we can organize into a set $S_1 = \{k^{i\star}_1, \cdots,k^{i\star\text{min}},\cdots,k^{i\star}_m\}$. Equivalently we organize into $S_2$ the $n$ policies composing the estimated best pool as per $\widehat S_\alpha^i(n)$. Then $\hat \kappa^{i \star}(\hat{s}(n))$ stands for $S_2$ and $\kappa^{i\star}$ for $S_1$.} Once aggregated this measure captures the share of simulations, per $n$ where the minimum dosage best policy was included in the estimated best pool.\\
\\
\textbf{Mean squared error (of best policy effect):} For each simulation $\hat{s}(n)$,  the estimated best policy treatment effect is scored by its error with respect to the true treatment effect:
\[m(\hat{s}(n)) :=  \hat \eta_{\widehat{S}_{TVA},\hat \kappa^\star}^{\text{hyb}} -   \eta_{S_{TVA}, \kappa^\star} .\]
And thus $\hat{m}(n)$ is simply the estimated MSE:
\[\hat{m}(n) := \frac{1}{|\mathcal{C}|}\sum_{\mathcal{C}} \frac{1}{|\mathcal{S}_{S_\alpha^i} (n)|} \sum_{\mathcal{S}_{S_\alpha^i} (n)} m^2(\hat{s}(n)). \]

\subsubsection{Properties of TVA} \label{subsec: properties_TVA}

Simulation performance of TVA attests to its main theoretical properties: support consistency, best policy estimation consistency, and normally distributed coefficient estimates.

Figure \ref{fig:puffer_vs_ols} depicts these results.  For TVA consistency, consider the blue and green performance points in panels B and D. Panel B shows that even for low $n$, TVA includes some of the best policies in $\widehat{S}_\alpha$ as well as rapidly (in $n$) and consistently including the \emph{minimum} dosage best policy, that is of particular interest to the policymaker. Because TVA pools policies that comprise the best policy, redundancies to winner's curse attenuations are minimized; Panel D shows that MSE of the best policy starts small and quickly falls to 0 (blue dots). Finally, besides best policy consistency, a more global concern is whether the TVA estimator is support consistent in the first place. This is explicitly verified in further simulations of our Online Appendix (e.g., Figure \ref{fig:puffer_vs_lasso2}). As elaborated in Section \ref{subsec: alternatives}, this overall performance is in marked contrast to other estimators.

Besides being consistent, TVA estimates are also distributed asymptotically normally,, which permits reliable inference in the usual manner.  This is demonstrated in Figure \ref{fig:puffer_vs_ols},  Panel A which depicts a close match between the empirical CDF of standardized marginal policy estimates $\frac{\hat\eta_\kappa - \eta_\kappa}{se(\hat\eta_\kappa)}$ (for $\kappa \in \widehat{S}_{TVA}$, entered as a mixture distribution with equal weights) and the CDF of a standard normal distribution for a series of sample sizes. First, abstracting away from model misspecification complications, consider the case of a large sample size ($n=10,000$) where the empirical support of TVA is always correct (blue empirical CDF). The almost perfect match to the theoretical CDF speaks to normally distributed estimates of all $M=3$ components of the mixture distribution.\footnote{For this exercise we use a single randomly chosen configuration $\mathcal{C}$ that determines the coefficients $\eta_\kappa$, $\kappa \in \{1,\cdots M\}$ against which to compare our estimates.}  The cases where $n < 10,000$ relax insistence on correct model specification and demonstrate that this is not an artifact of a very large sample size. Here, $\eta_{\widehat{S}_{TVA},\kappa}$ is the population pseudo-parameter in the (potentially) mispecified model. Even for rather moderate sample sizes ($n = 3000$), we see normally distributed estimates.

\subsection{Alternative Estimators}\label{subsec: alternatives}

In what follows, we stick to the simulation framework presented in Section \ref{sec: performance_TVA}. We primarily make the case for the strong performance of TVA relative to its most straightforward alternative, a direct application of OLS, by comparing the performance of both estimators on a range of measures outlined in Section \ref{sec: performance_TVA}. Looking at the various measures gives insight into what gives TVA its edge.  In an extended simulation section of our Online Appendix \ref{subsec: alternatives_appendix} we provide comparisons for further LASSO-based alternative estimation strategies, the results of which are summarized below.

\subsubsection{Direct OLS}

An intuitive route for inference in this setting is estimating the unique policy specification (\ref{eq:unique_policy}) using only an application of OLS and nothing else (a strategy we call ``direct OLS"). Since this is a fully saturated regression, this estimator has no theoretical issues with convergence nor with interpretation. Rather, this is about performance in the finite sample in the environment we describe. Most obviously, there is a loss of power in estimating, separately, the impact of 75 distinct interventions. Moreover, on the question of selecting and estimating the best policy effects, it faces the following inadequacies: (1) it fails at consistently identifying the minimum dosage best policy (2) the estimates of best policy exhibit a stronger winner's curse (3) the attenuations from applying \cite{andrews2019inference} are large (and relative to TVA, overly severe).

Figure \ref{fig:puffer_vs_ols}, panels B-D, which compare TVA and OLS on the relevant metrics, documents patterns (1)-(3) for best policy estimation. Panel B plots inclusion accuracies as a function of sample size, and speaks to pattern (1). Direct OLS (orange) does almost as well as TVA in estimating as best policy \emph{some} policy that is part of the true best pooled policy, but it effectively picks the dosage at random; thus selection of the \emph{minimum} dosage hovers at around $36.67\%$.

Panels C, and D document patterns (2) and (3). The direct OLS results are in orange. Panel C exhibits a strong winner's curse, which is to be expected in a situation with numerous candidates for ``best policy'', since the odds that a particularly large shock was drawn and thrust one to the top is quite high.  The attenuation resulting from the application of \cite{andrews2019inference}, as a result of this neck-and-neck competition, are also large; this is verified in Panel C, and in fact shrinkage goes up to 89\% in many individual simulations. These are actually over-attenuated; indeed, as the Panel D plots of attenuated best policy MSE over sample size show, a large MSE (of 0.52) persists even for large $n$.\footnote{Recall that while \cite{andrews2019inference} estimators are consistent, this assumes some local separation of parameters, which is not guaranteed in these neck-and-neck competitions.} In  contrast in all these panels, the winner's curse attenuations for TVA are much more modest, because of reduction of the number of competing policies and therefore greater separation between them.

Besides the specific issue of best policy estimation, direct OLS has low power. This is depicted in Figure \ref{fig:power_comparison}, where simulated OLS estimates of all the unique policies \eqref{eq:unique_policy} contrast with the pruned and pooled estimation \eqref{eq: pooled_pruned}. As expected, the estimated effects of the pooled policies are less dispersed  (panel A). In this visualization, we deliberately choose a configuration where the effects of different policies are similar, so that these histograms overlap. This makes the task of discovering the correct way to pool an interesting challenge  (further exemplified in the direct OLS estimates of a single simulated draw of data in Panel B, where 95\%-confidence intervals of policies from distinct pools strongly overlap), and it highlights the need of a disciplined procedure.

\subsubsection{Naive LASSO}\label{subsubsec: alternatives}

There are two ways we could ``naively" apply LASSO. The first is to disregard pooling, and apply LASSO on the unique policy specification (\ref{eq:unique_policy}) because sparse dosages might also mean a sparse set of policies. While there is no theoretical issue with this procedure in terms of model consistency, using this for policy estimation leads to much the same performance limitations as direct OLS with regards to best policy estimation, namely a persistently high best policy MSE stemming from overly severe correction from \cite{andrews2019inference}'s winner's curse adjustment. Figure \ref{fig:puffer_vs_lasso1} from our Online Appendix,  contrasts this ``No pooling, only pruning" version of LASSO to TVA on best policy estimation and documents these patterns in detail.

The second way to ``naively" apply LASSO is to consider both pooling and pruning as important, but adopt a sign inconsistent model selection procedure by applying LASSO directly on (\ref{eq: TVA}) without a Puffer transformation. As expected, simulations attest to inconsistent support selection though MSE on the best policy is comparable to TVA. Importantly, it fails to select the minimum dosage best policy with substantial probability relative to TVA (refer to Figure \ref{fig:puffer_vs_lasso2} in Online Appendix \ref{sec:simulations_appendix}, and related discussion for more details).

\subsubsection{Debiased LASSO}

Because we are interested in high dimensional inference\footnote{Albeit, we are still in a $K< n$ regime (mechanically since the number of treatments cannot exceed the number of units), sometimes called low dimensional with diverging number of parameters.} one alternative to a two step process of model selection and inference is the so-called ``debiased LASSO'' (\cite{zhang2014confidence}, \cite{javanmard2014confidence}, \cite{javanmard2018debiasing}, \cite{van2019asymptotic}). The basic idea is that since the downward bias in LASSO is estimable, we can reverse it.  A feature, however, is that these debiased coefficients are almost surely never exactly zero, so that there is no question of sparsity. We thus only need to consider applying debiased LASSO to (\ref{eq:unique_policy}).

In Figure \ref{fig:puffer_vs_dlasso} of our Online Appendix, we show that the debiased LASSO procedure suffers from the same limitations as direct OLS estimation, especially with regards to best policy estimation (high MSE due to over-attenuation of the winner's curse).

\subsubsection{``Off the Shelf" Bayesian approaches: Spike and Slab LASSO \citep{nie2022bayesian}}

The rules governing admissible pooling encodes the econometrician's prior about the environment. This raises the possibility of a Bayesian framework. Indeed, LASSO estimates have a Bayesian interpretation in terms of Laplace priors. One can ask whether a more sophisticated, ``explicitly" Bayesian approach can address our final objectives. This paradigmatically different route is the topic of future work. In Section \ref{subsubsec:bbssl} of our Online Appendix, we just show that``off the shelf" Bayesian approaches are unlikely to help. In particular, we show that a direct application of spike and slab formulations -- the most intuitively relevant method -- underperforms relative to our TVA procedure with a performance pattern similar to that of applying Naive LASSO to the marginal specification  \eqref{eq: TVA}.

\subsection{Performance Under Five Sparsity Relaxation Regimes}\label{subsec: sparisty robustness}

Our main theoretical guarantees in Section \ref{sec:Estimation} hold in an environment with exact sparsity and with marginal effect sizes uniformly bounded away from 0. Here we explore \emph{practical} performance relaxing this in several plausible ways. Although performance of TVA suffers, it is still strong; moreover,  as elaborated in section \ref{subsec: sparsity_appendix} of our Online Appendix, TVA does better than the next best practical alternative of applying naive LASSO to the marginal specification \eqref{eq: TVA}.

We consider five regimes of sparsity and effect size relaxations. Although the support configurations are no longer necessarily of cardinality $M$, they are still randomly chosen as in Section \ref{sec: performance_TVA}. In Regimes 1 and 2 we relax exact zeros in the non-primary marginals to small effect sizes; these are either rapidly diminishing as $\Theta(\frac{1}{n})$ (Regime 1) or moderately diminishing as $\Theta(\frac{1}{\sqrt{n}})$ (Regime 2).\footnote{Note that only a diminishing rate retains a threat of misspecified model selection at any $n$} In Regime 3 we further relax sparsity in the first regime by expanding the true support to include marginal effects of both large and medium sizes. In Regime 4 we relax the lower bound on marginal effect sizes for the primary marginals, diminishing at a rate $\Theta(\frac{1}{n^{0.2}})$ between moderate and rapid. In the fifth regime we further relax sparsity in Regime 4 by expanding the support with rapidly diminishing coefficients. A summary of the regime configurations is described below:

\begin{enumerate}
\item \textbf{Regime 1}: $M$ constant marginal effect sizes in $[1,5]$ \& $M$ rapidly diminishing remaining marginal effect sizes in $[1,5]/n$.
\item \textbf{Regime 2}: $M$ constant marginals in $[1,5]$ \& $M$ moderately diminishing remaining marginals ($[1,5]/\sqrt{n}$).
\item \textbf{Regime 3}: $M$ large constant marginals in $[5,10]$, $M$ medium marginals in $[1,2]$ \& $M$ rapidly diminishing remaining marginals in $[1,5]/n$.
\item \textbf{Regime 4}: $M$ decreasing marginals in $[1,5]/n^{0.2}$ (and zero marginals everywhere else).
\item \textbf{Regime 5}: $M$ decreasing marginals in $[1,5]/n^{0.2}$ \&  $M$ moderately diminishing remaining marginals  in $[1,5]/\sqrt{n}$.
\end{enumerate}

The main finding is that support accuracy of TVA is generally strong. Even in the case of model misspecification, the MSE of the best policy is still low for moderate sample size. Furthermore its distinct advantage relative to alternatives with regards to best policy estimation --- the much more reliable selection of the minimum dosage best policy --- remains equally strong.

Figure \ref{fig:regime_plots}, panel A speaks to the first pattern: even if support accuracy suffers from relaxing sparsity requirements performance remains generally high. For regimes 1-3 for example, support accuracy converges to 100\% albeit more slowly than in the exactly sparse environment (and convergence in R2 is, as expected, slower than in R1).  For regimes 4 and 5, even though TVA is support inconsistent, this does not imply a higher MSE of the best policy demonstrating that model misspecification is not generally threatening with regards to this final objective.  Panel B shows this very clearly where MSE is steadily decreasing and comparable across all regimes.
Performance on some best policy selection --- the easier task --- seems also particularly unaffected by the sparsity relaxations, with accuracy ranging between 80\%-90\% (panel C). Finally TVA's distinctive advantage over all alternatives explored in Section \ref{subsec: alternatives} is the reliable selection of the minimum dosage best policy. Performance on minimum best policy selection remains generally strong across regimes (panel D) with a steady increase as sample size grows. Notably the minimum best inclusion rate is above 90\% across all regimes for even moderate sample sizes ($n = 3000$). \\

Taken together, these simulation results make the case that TVA is both a powerful and robust candidate for our setting.

\section{Context, Experimental Design, and Data}\label{sec:context}

We now apply this method to a large-scale experiment conducted in collaboration with the government of Haryana to help them select the most effective policy bundle to stimulate demand for immunization.
The objective of the experiment was explicitly to select the best policy to scale up, after one year-long experiment with 75 potentially distinct treatments, making it an excellent setting for this method.

\subsection{Context}

This study took place in Haryana, a populous state in North India, bordering New Delhi.
In India, a child between 12 and 23 months is considered  to be fully immunized if he or she receives one dose of BCG, three doses of Oral Polio Vaccine (OPV), three doses of DPT,  
and at least one dose of a measles vaccination.
India is one of the countries  where immunization rates are puzzlingly low. According to the 2015-2016 National Family Health Survey, only 62\% of children were fully immunized \citep{nfhs4haryana}.
This is not due to lack of access to vaccines or health personnel. The Universal Immunization Program (UIP) provides all vaccines free of cost to beneficiaries, and vaccines are delivered in rural areas--even in the most remote villages. Immunization services have made considerable progress over the past few years and are much more reliably available than they used to be. During the course of our study we found
that the monthly scheduled immunization session were almost always run in each village.

The central node of the UIP is the Primary Health Centre (PHC). PHCs are health facilities that provide health services to an average of 25 rural and semi-urban villages with about 500 households each. Under each PHC, there are approximately four sub-centres (SCs). Vaccines are stored and transported from the PHCs to either sub-centers or villages on an appointed day each month, where there is a mobile clinic where the Auxiliary Nurse Midwife (ANM) administers vaccines to all eligible children. A local health worker, the Accredited Social Health Activist (ASHA), is meant to help map eligible households, inform and motivate parents, and take them to the immunization session. She receives a small fee for each shot given to a child in her village.

Despite this elaborate infrastructure, immunization rates are particularly low in North India, especially in Haryana.  According to the District Level Household and Facility Survey, the full immunization coverage among 12-23 months-old children in Haryana fell from 60\% in 2007-08 to 52.1\% in 2012-13  \citep{dlhs4haryana}.

In the district where we carried out the study, a baseline study revealed even lower immunization rates (the seven districts that were selected were chosen because they have low immunization). About 86\% of the children (aged 12-23 months) had received at least three vaccines. However, the share of children whose parents had reported they received the measles vaccine (the last in the sequence) was 39\%, and only 19.4\% had received the vaccine before the age of 15 months, while the full sequence is supposed to be completed in one year.

After several years focused on improving the supply of immunization services, the government of Haryana was interested in testing out strategies to improve household take-up of
immunization, and in particular, their completion of the full immunization schedule. With support from USAID and the Gates Foundation, they entered into a partnership with J-PAL
to test out different interventions. The final objective was to pick the best policy possibly scale up throughout the state.

Our study took place in seven districts where immunization was particularly low. In four districts, the full immunization rate in a cohort of children older than the ones we consider, was below 40\%, as reported by parents (which is likely a large overestimate of the actual immunization rate, given that children get other kinds of shots and parents often find it hard to distinguish between them, as noted in   \cite{report3ie2021}). Together, the districts cover a population of more than 8 million (8,280,591) in more than 2360 villages, served by 140 PHCs and 755 SCs. The study covered all these PHCs and SCs, and are thus fully representative of the seven districts. Given the scale of the project, our first step was to build a platform to keep a record of all immunizations. Sana, an MIT-based health technology group, built a simple m-health application that the ANMs used
to register and record information about every child who attended at least one camp in the sample villages. Children were given a unique ID that made it possible to track them across visits and centers. Overall, 295,038 unique children were recorded in the system, and 471,608 vaccines were administered. Data from this administrative database is our main source of information on immunization and we discuss its reliability below.  More details on the implementation are provided in the publicly available progress report \citep{report3ie2021}.

\subsection{Interventions}

The study evaluates the impact of several nudges on the demand for immunization: small incentives, targeted reminders, and local ambassadors.

\subsubsection{Incentives}

When households are indifferent or have a propensity to procrastinate, small incentives can offset any short term cost of getting to an immunization camp and lead to a large effect on immunization. \cite{banerjee2010improving} shows that small incentives for immunization in  Rajasthan (a bag of lentils for each shot and a set of plates for completing the course) led to a large increase in the rates of immunization. Similar results were subsequently obtained in other countries, suggesting that incentives tend to be effective {\citep{bassani2013financial,gibson2017mobile}}. In the Indian health system, households receive incentives for a number of health behavior, including hospital delivery, pre-natal care visits, and, in some states (like Tamil Nadu), immunization.

The Haryana government was interested in experimenting with incentives. The incentives that were chosen were mobile recharges for pre-paid phones, which can be done cheaply and reliably on a very large scale. Almost all families have at least one phone and the overwhelming majority of the phones are pre-paid. Mobile phone credits are of uniform quality and fixed price, which greatly simplify procurement and delivery.

A small value of mobile phone credit  was given to the caregivers each time they brought their child to get immunized. Any child under the age of 12 months receiving one of the five eligible shots (i.e., BCG, Penta-1, Penta-2, Penta-3, or Measles-1), was considered eligible for the incentives intervention. Mobile recharges were delivered directly to the caregivers' phone number that they provided at the immunization camp. Seventy (out of the 140) PHCs were randomly selected to receive the incentives treatment.

In \cite{banerjee2010improving}, only one reward schedule was experimented with.  It involved a flat reward for each shot plus a set of plates for completing the immunization program.
This left many important policy questions pending: does the level of incentive make a difference? If not, cheaper incentives could be used. Should the level of rewards increase with each immunization to offset the propensity of the household to drop
out later in the program?

To answer these questions, we varied the level of incentives and whether they increased over the course of the immunization program.
The randomization was carried out within each PHC, at the subcenter level. Depending on which sub-center the caregiver fell under, she would either receive a:

\begin{enumerate}
	\item Flat incentive, high: INR 90 (\$1.34 at the 2016 exchange rate, \$4.50 at PPP) per immunization (INR 450  total).
	\item	Sloped incentive, high: INR 50 for each of the first three immunizations, 100 for the fourth, 200 for the fifth (INR 450 total).
	\item	Flat incentive, low: INR 50 per payment (INR 250 total).
	\item	Sloped incentive, low: INR 10  for each of the first three immunizations, 60 for the fourth, 160 for the fifth (INR 250 total).
\end{enumerate}

Even the high incentive levels here are small and therefore implementable at scale, but they still constitute a non-trivial amount for the households. The ``high'' incentive level was chosen to be roughly equivalent to the level of incentive
chosen in the Rajasthan study: INR 90 was roughly the cost of a kilogram of lentils in Haryana during our study period. The low level was meant to be half of that (rounded to INR 50 since the vendor could not deliver recharges that
were not multiple of 10). This was meaningful to the households: INR 50 corresponds to 100 minutes of talk time on average.
The provision of incentives was linked to each vaccine. If a child missed a dose, for example Penta-1, but then came for the next vaccine (in this case, measles), they would receive both Penta-1 and measles and get the incentives for both at once, as per the schedule described above.

To diffuse the information on incentives, posters were provided to ANMs, who were asked to put them up when they set up for each immunization session.
The  village ASHAs and the ANMs were also supposed to inform potential beneficiaries of the incentive structure and amount in the relevant villages. However, there was no systematic large scale information campaign, and it is possible that not everybody was aware of the presence or the schedule of the incentives, particularly if they had never gone to a camp.

\subsubsection{Reminders}

Another frequently proposed method to increase immunization is to send text message reminders to parents. Busy parents have limited attention and reminders can put the immunization back at the ``top of the mind.'' Moreover, parents do not necessarily understand that the last immunization in the schedule (measles) is for a different disease and is at least as important as the previous ones. SMSs are also extremely cheap and easy to administer in a population with widespread access to cell phones. Even if not everyone gets the message, the diffusion may be reinforced by social learning, leading to  faster adoption.\footnote{See, e.g., \cite*{rogers1995,krackhardt1996,kempekt2003,jackson2008b,iyengarvv2010,hinzetal2011,katonazs2011,jacksony2011,banerjeecdj2013,blochjt2016,jackson2017,akbarpour}.}

The potential for SMS reminders is recognized in India. The Indian Academy of Pediatrics rolled out a program in which parents could enroll to get reminders by providing their cell phone number and their child's date of birth.
Supported by the Government of India, the platform planned to enroll 20 million children by the end of 2020.

Indeed, text messages have already been shown to be effective to increase immunization
in some contexts. For example, a systematic review of five RCTs finds that reminders for immunization increase take up on average \citep{mekonnen2019effect}.
However, it remains true that text messages could  have no effect or even backfire if parents do not understand the information provided and feel they have no one to ask \citep{banerjee2018less}.
Targeted text and voice call reminders were sent to the caregivers to remind them that their child was due to receive a specific shot. To identify any potential spillover to the rest of the network, this intervention followed a two step randomization. First, we randomized the study sub-centers into three groups: no reminders, 33\% reminders, and 66\% reminders. Second, after their first visit to that sub-center, children's families were randomly assigned to either get the reminder or not, with a probability corresponding to the treatment group for their sub-centers. The children were assigned to receive/not receive reminders on a rolling basis.

The following text reminders were sent to the beneficiaries eligible to receive a reminder. In addition, to make sure that the message would reach illiterate parents, the same message was sent through an automated voice call.

\begin{enumerate}
	\item Reminders in incentive-treatment PHCs:
	\begin{quote}
	\small{``Hello! It is time to get the <<name of vaccine>> vaccine administered for your child <<name>>. Please visit your nearest immunization camp to get this vaccine and protect your child from diseases. You will receive mobile credit worth <<range for slope or fixed amount for flat>>  as a reward for immunizing your child.''}
	\end{quote}
	
	\item Reminders in incentive-control PHCs:
	\begin{quote}
	\small{``Hello! It is time to get the <<name of vaccine>> vaccine administered for your child. Please visit your nearest immunization camp to get this vaccine and protect your child from diseases.''}
	\end{quote}
\end{enumerate}

\subsubsection{The Immunization Ambassador: Network-Based Seeding}

The goal of the  ambassador  intervention was to leverage the social network to spread information. The objective was to identify influential individuals who could relay to villagers both the information on the existence of the immunization camps, and, wherever relevant, the information that incentives were available.
Existing evidence shows that people who have a high centrality in a network (e.g., they have many friends who themselves have many friends) are able to spread information more widely in the community \citep{katzl1955,aralw2012,banerjeecdj2013,beaman2018can,banerjeeusing}.
Further,  members in the social network are able to easily identify individuals, whom we call information hubs, who are the best placed to diffuse information as a result of their centrality as well other personal characteristics (social mindedness, garrulousness, etc.)\citep{banerjeeusing}.

This intervention took place in a subset of 915 villages where we collected a full census of the population (see below for data sources).
Seventeen respondents in each village were randomly sampled from the census to participate in the survey, and were asked
to identify people with certain characteristics (more about those later). Within each village, the six people nominated most often by the group of 17 were  recruited
to be ambassadors for the program. If they agreed, a short survey was conducted to collect some demographic variables, and they were then formally asked to become program ambassadors. Specifically, they agreed to receive
one text message and one voice call every month, and to relay it to their friends. In villages without incentives, the text message was a bland reminder of the value of immunization. In villages with incentives, the text message further
reminded the ambassador (and hence potentially their contacts) that there was an incentive for immunization.

While our previous research had shown that villagers can reliably identify information hubs, a pertinent question for policy unanswered by previous work is whether the information hubs can effectively transmit messages about health, where trust in the messengers may be more important than in the case of more commercial messages.

There were four groups of ambassador villages, which varied in the type of people that the 17 surveyed households were asked to identify. The full text is in Appendix \ref{sec:hub_questions}.

\begin{enumerate}
	\item \emph{Random} seeds: In this treatment arm, we did not survey villages. We picked six ambassadors randomly from the census.
	
	\item 	\emph{Information hub} 	seed: Respondents were asked to identify who is good at relaying information.

	\item \emph{Trusted} seed: Respondents were asked to identify those who are generally trusted to provide good advice about health or agricultural questions

	\item \emph{Trusted information hub} seed:
	Respondents were asked to identify who is both trusted and good at transmitting information

\end{enumerate}

\subsection{Experimental Design}

The government was interested in selecting the best policy, or bundle of policies, for possible future scale up.
We were agnostic as to the relative merits of the many available variants. For example, we did not know
whether the incentive level was going to be important, nor did we know if the villagers would be able to identify trusted people effectively and hence, whether the intervention to
select trusted people as ambassadors would work. However,  we believed that there could be significant interactions between different policies. For example, our prior was that the
ambassador intervention was going to work more effectively in villages with incentives, because the message to diffuse was clear.  We therefore implemented a completely cross-randomized design, as illustrated in our Online Appendix Figure \ref{fig:experimental_design}.

We started with 2,360 villages, covered by 140 PHCs, and 755 sub-centers. The 140 PHCs were randomly divided into 70 incentives PHCs, and 70 no incentives PHCs (stratifying by district).
Within the 70 incentives PHCs, we randomly selected the sub-centers to be allocated to each of the four incentive sub-treatment arms. Finally, we only had resources to conduct a census and a baseline exercise
in about 900 villages. We selected about half of the villages from the coverage area of each subcenter, after excluding the smallest villages. Only among the 915 villages did we conduct the ambassador
randomization: after stratifying by sub-center, we randomly allocated the 915 villages to the control group (no ambassador) or one of the four ambassador treatment groups.

In total, we had one control group, four types of incentives interventions, four types of ambassador interventions, and two types of SMS interventions. Since they were fully cross-randomized (in the sample of 915 villages), we had 75 potential policies, which is large  even in relation to our relatively large sample size.  Our goal is to identify the most effective and cost-effective policies
and to provide externally valid estimates of the best policy's impact, after accounting for the winner's curse problem. Further, we want like to identify other
effective policies  and answer the question of whether different variants of the policy had the same or different impacts.

\subsection{Data}

\subsubsection{Census and Baseline}

In the absence of a comprehensive sampling frame, we conducted a mapping and census exercise across 915 villages falling within the 140 sample PHCs. To conduct the census, we visited 328,058 households, of which 62,548 households satisfied our eligibility criterion (children aged 12 to 18 months). These exercises were carried out between May and November 2015. The data from the census was used to sample eligible households for a baseline survey. We also used the census to sample the respondent of the ambassador identification survey (and to sample the ambassadors in the ``random seed'' villages). Around 15 households per village were sampled, resulting in data on 14,760 households and 17,000 children. The baseline survey collected data on demographic characteristics, immunization history, attitudes and knowledge and was conducted between May and July 2016. {A village-level summary of baseline survey data is given in Appendix  \ref{table:baseline}}.

\subsubsection{Outcome Data}
Our outcomes of interest are the number of vaccines administered for each vaccine every month, and the number of fully immunized children every month.
The main analysis of this paper focuses on the number of children who received the measles vaccines in each village every month. The measles vaccine is the last vaccine in the immunization schedule, and the ANMs check the immunization history and administer missing vaccines when a child is brought in for this vaccine. As a result, it is a good proxy for a child being fully immunized.

For our analysis, we use administrative data collected by the ANM using the e-health application on the tablet,  stored on the server, to measure immunization. At the first visit, a child was registered using a government provided ID (or in its absence, a program-generated ID) and past immunization history, if any. In subsequent visits, the unique ID was used to pull-up the child's details and update the data. Over the course of the program, about 295,038 children were registered, yielding a record of 471,608 immunizations.  We use the data from December 2016 to November 2017. We do this because of a technical glitch in the system--the SMS intervention was discontinued from November 2017, although the incentives and information hub interventions were continued a little longer, through March 2018.

Since this data was also used to trigger SMS reminders and incentives, and for the government to evaluate the nurses' performance,\footnote{Aggregated monthly reports generated from this data replaced  the monthly reports previously compiled by hand by the nurses.} it was important to assess its accuracy. Hence, we conducted a validation exercise, comparing the administrative data with random checks, as described in Appendix \ref{sec:data_validation}.  The data quality appears to be excellent. Finally, one concern (particularly with the incentive program) is that the intervention led to a pattern of substitution, with children who would have been immunized elsewhere (in the PHC or at the hospital) choosing to be immunized in the camp instead. To address this issue, we collected data immediately after the intervention on a sample of children who did not appear in the database (identified through a census exercise), to ascertain the status of their immunization. In Appendix \ref{sec:substitution}, we show that there does not appear to be a pattern of substitution, as these children were not more likely to be immunized elsewhere.

Below, the dependent variable is the number of measles shot given in a village in a month (each month, one immunization session is held at each site). On average, in the entire sample, 6.16 measles shot were delivered per village every month (5.29 in the villages with no intervention at all). In the sample at risk for the ambassador intervention (which is our sample for this study) 6.94 shots per village per month were delivered.

\subsection{Interventions Average Effects}\label{subsec: average_effects}

In this section, we present the average effects of the interventions using a standard regression without interactions.

We focus on the sample of census villages used throughout our analysis - which are the villages where the ambassador intervention was also administered - and run the following specification:
\begin{align*}
	y_{dsvt}&=\alpha+\beta^\prime \text{Incentive}_s +\gamma^\prime \text{SMS}_s  +\delta^\prime \text{Ambassador}_v+ \upsilon_{dt}+\epsilon_{dsvt}.
\end{align*}
We weight these village-level regressions by village population, and standard errors are clustered at the SC level.\footnote{{This is the highest level at which a treatment is administered, so clustering at this level should yield the most conservative estimate of variance. In practice clustering at the village level or SC level does not make an appreciable difference.}}

{The results (already reported in \cite{banerjeeusing}) are depicted graphically in Figure \ref{fig:agg-policy-effects} and show that, on average,  using information hubs (``gossips'' in that paper) as ambassadors has positive effects on  immunization: 1.89 more children receive a measles vaccine on a base of 7.32 in control  in this sample ($p = 0.04$). This is nearly identical to the effect of the high-powered, sloped incentive, though this intervention is considerably cheaper. In contrast, none of the other ambassador treatments--random seeding, seeding with trusted individuals, or seeding with trusted information hubs--have benefits statistically distinguishable from zero ($p = 0.42, \  p = 0.63, \text{ and } p = 0.92$ respectively) and the point estimates are small,  as well. To ensure that conclusions are not simply an artifact of this particular subsample, we show in Appendix \ref{sec:fig-table-appendix} that these results are robust to running the analysis on the full sample, .

The conclusion from this analysis is that financial incentives can be effective to boost demand for immunization, but only if they are large enough and increase with each immunization. Of the two cheaper interventions, the SMS interventions, promoted widely in India and elsewhere, seem disappointing. In contrast,
leveraging the community by enrolling local ambassadors, selected using the cheap procedure of asking a few villages who are good information hubs, seems to be as effective as using incentives. It leads to an increase of 26\% 
in the number of children who complete the schedule of immunization every month. This alone could increase full immunization rate in those districts from 39\% (our baseline full immunization rate, as reported by parents) to nearly 49\%.
This analysis does not fully answer the policymaker's question, however. It could well be that the interventions have powerful interactions with each other, which has two implications. First, the main effect, as estimated, does not tell us what the impact of the policy would be in Haryana if implemented alone (because as it is, they are a weighted average of a complicated set of interacted treatments). Second, it is possible that the government could do better by combining two (or more) interventions. For example, our prior in designing the information hub  ambassador intervention (described in our proposal for the project)\footnote{\url{https://doi.org/10.1257/rct.1434-4.0}} was that it would have a positive interaction effect with incentives, because it would be much easier for the information hubs  to relay hard information (there are incentives) than a vaguer message that immunization is useful.  The problem, however, is that there are a large number of interventions and interactions: we did not---nor was it feasible to---think through ex-ante all of the interactions that should or should not be included, which is why in \cite{banerjeeusing}, we only reported the average effects of each different type of seeds in the entire sample, without interactions. In the next section, we adapt our disciplined approach to select which ones to include, and to then estimate the impact of the ``best'' policy.



\section{Results}\label{sec:results}

\subsection{Identifying effective policies}

\subsubsection{Method}

We adapt the TVA procedure for our case. We allow only some pooling within arms depending on the nature of the sub-treatment. In the incentive arms, slope and flat incentives are not allowed to pool, but the amount of money (high or low) is considered to be
a dosage.  In the ambassador arms, we do not allow pooling between random selection of ambassadors, trusted ambassador, and information hub. Within information hubs, however,
we consider that the ``trusted information hub" is an increased dosage of information hub,  so these may pool with one another.

To summarize, interventions ``information hubs," ``slope," ``flat," and ``SMS" are found in two intensities. The marginal specification \eqref{eq: TVA} therefore looks like

{\small
\begin{align*}\label{eq:smart_pool_haryana}
	y_{dsvt} &= \alpha_0 + \alpha_{SMS} \text{SMS}_{s} + \alpha_{H,SMS} \text{High SMS}_{s}\\
	&+ \alpha_{Slope} \text{Slope}_{s} + \alpha_{H,Slope}\text{High Slope}_{s}
	+ \alpha_{Flat} \text{Flat}_{s} + \alpha_{H,Flat}\text{High Flat}_{s}\\
	&+ \alpha_{R} \text{Random}_{v} + \alpha_{H}\text{Info Hub (All)}_{v} + \alpha_{T}\text{Trust}_{v} + \alpha_{TH}\text{Trusted Info Hub}_{v} \\
	& + \alpha_{X}^\prime X_{sv} + v_{dt} + \epsilon_{dsvt},
\end{align*}
}
where we have explicitly listed the variables in ``single arm" treatment profiles. $X_{sv}$ is a vector of the remaining 64 marginal effects variables in ``multiple arm" treatment profiles, and $v_{dt}$ is a set of district-time dummies. Here $\text{SMS}$ refers to ``any SMS''.

Our model selection estimation follows the recommended implementation in  \cite{rohe2014note}, which uses a sequential backward elimination version of LASSO (variables with $p$-values above some threshold are progressively deselected)   on the  $\text{Puffer}_N$ transformed variables (this aids in correcting for the heteroskedasticity induced by the Puffer transformation).
 We select penalties $\lambda$ for both regressions (number of immunizations and immunizations per dollar) to minimize a Type I error, which is particularly important to avoid in the case of policy implementation.\footnote{ \cite{rohe2014note} notes a bijection between a backwards elimination procedure based on using Type I error thresholds and the penalty in LASSO. We take $\lambda = 0.48$ and $\lambda = 0.0014$ for the number of immunizations and immunizations per dollar outcomes, respectively. Both of these choices map to the same Type I error value ($p = 5 \times 10^{-13}$) used in the backwards elimination implementation of LASSO selected to essentially eliminate false positives. Appendix \ref{sec: robustness} elaborates on this choice.} This makes sense because it is extremely problematic to have a government introduce a large policy based on a false positive. This reasoning is elaborated in Appendix \ref{sec: robustness}.

This gives $\widehat{S}_{\alpha}$, an estimate of the true support set $S_\alpha$ of the marginal effects specification (\ref{eq: TVA}). We then generate a use of unique pooled policy set $\widehat{S}_{TVA}$ (following the procedure we outline in Algorithm \ref{alg:pooling} in Appendix \ref{sec:pooling}). Next, we run the pooled specification \eqref{eq: pooled_pruned} to obtain post-LASSO estimates $\hat{\eta}_{\widehat{S}_{TVA}}$ of the pooled policies as well as $\hat{ \eta}_{\widehat{S}_{TVA},\hat{\kappa}^\star}^{\text{hyb}}$, the winner's curse adjusted estimate of the best policy.

\subsubsection{Results}

The results are presented in Figure \ref{fig:postlasso_combined}. Panel A presents the  post-LASSO estimates where the outcome variable is the number of measles vaccines per month in the village. Panel B presents the  post-LASSO estimates where the outcome variable is the number of measles vaccines per dollar spent. In each, a relatively small subset of policies is selected as part of $\widehat{S}_{TVA}$ out of the universe of 75 granular policies (16\% of the possible options in Panel A and 35\% in Panel B).


In Figure \ref{fig:postlasso_combined}, Panel A, two of the four selected pooled policies  are estimated to do significantly better than control:  information hubs seeding with sloped incentives (of both low and high intensities) and SMS reminders (of both 33\% and 66\% saturation) are estimated to increase the number of immunizations by 55\% 
relative to control ($p = 0.001$), while trusted seeds with high-sloped incentives and SMS reminders (of both saturation levels) are estimated to increase immunizations by 44\% 
relative to control ($p = 0.009$). 

These two effective policies increase the number of immunizations, relative to the status quo,
at the cost of a greater cost for each immunization (compared to standard policy). These policies induce 36.0 immunizations per village per month per \$1,000 allocation (as compared with 43.6 immunizations per village per month in control). The reason is that that the gains from having incentives in terms of immunization rates is smaller than the increase in costs (e.g.,  the incentives must be paid to all the infra-marginal parents).

Two things are worth noting to qualify those results, however. First, in \cite{chernozhukov2018generic}, we show that in the places where the full package treatment is predicted to be the most effective (which tends to be the places with low immunization), the number of immunizations per dollar spent is not statistically different in treatment and control villages. Second, immunization is so cost-effective, that this relatively small increase in the cost of immunization may still mean a much more cost-effective use of funds than the next best use of dollars on policies to fight childhood disease \citep{ozawa2012cost}.

Nevertheless, a government may be interested in the most cost-effective policy, if they have a given budget for immunization.  We turn to policy cost effectiveness in Figure \ref{fig:postlasso_combined}, Panel B. The most cost-effective policy (and the only policy that reduces per immunization cost) compared to control is the combination of information hub
seeding (trusted or not) with SMS reminders (at both 33\% or 66\% saturation) and no incentives, which leads to a 9.1\% 
increase in vaccinations per dollar ($p = 0.000$).

\subsection{Estimating the Impact of the Best policy}

To estimate the impact of the best policy, we first select the best policy from $\widehat{S}_{TVA}$ based on the post-LASSO estimate. Then, we attenuate it using the hybrid estimator with $\alpha = 0.05$ and $\beta = \frac{\alpha}{10} = 0.005$, which this is the value used by \cite{andrews2019inference} in their simulations. The hybrid confidence interval has the following interpretation: conditional on policy effects falling within a 99.5\% simultaneous confidence interval, the hybrid confidence interval around the best policy has at least 95\% coverage. It also has at least 95\% coverage unconditionally.\footnote{Per Proposition 6 of \cite{andrews2019inference}, it has unconditionial coverage between $1-\alpha = 95\%$ and $\frac{1-\alpha}{1-\beta} = 95.58\%$.}

Table \ref{tab:max} presents the results.  In column 1, the outcome variable is the number of measles vaccines given every month in a given village. We find that for the best policy in the sample (information hub seeds with sloped incentives at any level and SMS reminders at any saturation) the hybrid estimated best policy effect relative to control is 3.26 with a 95\% hybrid confidence interval of [0.032,6.25].  This is lower than the original post-LASSO estimated effect of 4.02. The attenuation is owing to a second best policy (trusted seeds with high sloped incentives with SMS reminders at any saturation), chasing the best policy estimate somewhat closely.\footnote{The increased attenuation from a more closely competing second-best policy emerges from the formulas for conditional inference given in Section 3 of \cite{andrews2019inference}.} Nevertheless, even accounting for winner's curse through the attenuated estimates and the adjusted confidence intervals, the hybrid estimates still reject the null. Thus, the conclusion is that accounting for winner's curse, this policy increases immunizations by 44\% 
relative to control.

While policymakers may chose this policy if they are willing to bear a higher cost to increase immunization, there may be settings where cost effectiveness is an important consideration.
In column 2, the outcome variable is the number of vaccinations per dollar.  Accounting for winner's curse through hybrid estimation, for the best policy of information hubs (all variants) and SMS reminders (any saturation level), the hybrid estimated best policy effect relative to control is 0.004 with a 95\% hybrid confidence interval of [0.003,0.004]. Notably, this appears almost unchanged from the naive post-LASSO. This is because no other pooled policy with positive effect is ``chasing" the best policy in the sample; the second-best policy is the control (status quo), which is sufficiently separated from the best policy so as to have an insignificant adjustment for winner's curse. Thus, adjusting for winner's curse, this policy increases the immunizations per dollar by 9.1\% relative to control. 

One concern with these estimates may be that they are sensitive to the implied LASSO penalty $\lambda$ chosen. To check the robustness of our results, we consider alternative values of $\lambda$. However, we also need a criteria for evaluating results under various $\lambda$ since a marginal effects support will never be robust for the whole range of $\lambda$.  Appendix \ref{sec: robustness} spells out this criteria, which amounts to formulating a set of ``admissible" $\lambda$ for the practitioner. In a nutshell, the criterion is to avoid including in $\widehat{S}_\alpha$ first and second best policies that are very likely to be false positives. Including a false positive as the first best is obviously the most serious error in the context of policy advising, but it also matters for the second best, since including these in the support may overly attenuate the best policy estimate for winner's curse. In our case, we find that for both immunizations and immunizations per dollar, the winner's curse adjusted estimates are robust for their respective sets of admissible $\lambda$. To exemplify this robustness, we can take the union of confidence intervals within their admissible sets. This is [0.32,6.25] for immunizations and [0.001,0.006] for immunizations per dollar. Neither is much wider than the single confidence interval for the choice of $\lambda$ we highlight.

Though admissible $\lambda$ are on a different scale for the two outcomes, there is a sense in which the admissible set is larger for immunizations per dollar. This suggests a different kind of robustness concern which is more about the relative fragility of the TVA estimator for each of the outcomes.   We can speak to this fragility using a bootstrapping analysis described in detail in our Online Appendix \ref{sec: extended_robustness}. Intuitively, it captures stability of best policy estimation in terms of observation leverage, where conclusions driven by outliers will fare worse. In this analysis, the best policy for cost-effectiveness holds for 96\% with highly concentrated estimates around the main one in actual data. Meanwhile, the best policy for immunization holds for 77\% of bootstrapped samples with estimates more widely dispersed. This speaks to the relative stability of the best policy for cost effectiveness over that for immunizations.


\section{Conclusion} \label{sec:conclusion}

Despite immunization being one of the most effective and cost-effective methods to prevent disease, disability, and mortality, millions of children each year go unvaccinated. The COVID-19 epidemic has made the situation worse: vaccine coverage has dipped to levels not seen since the 1990s \citep{gates2020}. Swift policy action is critical to ensure that this dip is temporary and children who missed immunizations during the pandemic get covered soon.

In   rural India, there was a priori reason to believe that nudges may work. After all, many children get their first vaccines but caregivers rarely follow through. This is consistent with the vast majority of caregivers reporting that vaccines are helpful. Yet, it was a priori unclear as to which nudge, let alone which policy bundle out of the 75 candidates, would be effective.

Respecting this genuine uncertainty was critical. If we had simply done parallel treatments of incentives, reminders, and ambassadors, we might have found no effects. Our key finding is that combined interventions work better than each in isolation. Although there is temptation of paring down the number of treatments a priori for the sake of power, there is a danger in not doing this in a data-driven way. The suggestion of avoiding all interactions in this setting (made in \cite{muralidharan2019factorial}) would  have led to the conclusion that nothing is effective.

Additionally, the interaction effects identified by TVA teach us something about the world. From the point of view of public health policy it tells us that it is valuable to add network-based insights (information hubs), which are not in a typical policymaker's toolkit, to catalyze the effects of conventional instruments. From a basic research perspective, it also suggests that the information hubs, i.e. the person best placed in a village to circulate information, may be more effective when they have something concrete to talk about, such as incentives or something to explain such as SMSs. Such questions merit future research.

The method suggested here is applicable to many domains where policymakers have several arms with multiple potential doses, do not have the time or capacity to adaptively experiment, or have genuine uncertainty about which policy bundles should be effective. Rather than guessing, or pretending to be an oracle, we suggest that policymakers consider the data-driven approach of treatment variant aggregation which may apply to their setting. The proposed method relies on strong assumptions that rule out some of the cases where model-selection leads to invalid inferences. Provided these assumptions are palatable, our findings show that  TVA prunes and pools effectively and, thispays dividends when the policymaker wishes to adjust for the winner's curse without falling into the trap of over-conservatism. The algorithm can be easily pre-specified and does not require the researcher to take a stance on the possible effects of myriad interactions which are likely difficult to predict in advance.


\bibliographystyle{ecta}
\bibliography{networks2,gossips}

\clearpage

\section*{Figures}



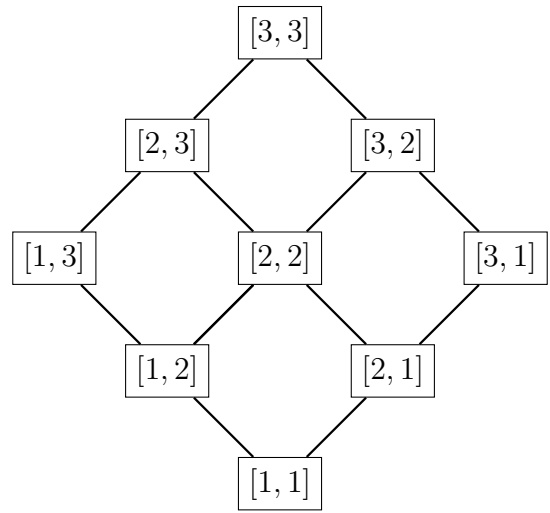
\begin{figure}[!ht]
\begin{centering}
 \begin{center}
 \vfill
 \begin{tikzpicture}[scale = 0.75] \def \n {5} \def \radius {2cm} \def \margin {8}
 \node[draw, rectangle,  minimum size=19pt] at (0,-4) (v1){$[1,1]$};
 \node[draw, rectangle, minimum size=19pt] at (-2,-2) (v2){$[1,2]$};
 \node[draw, rectangle,  minimum size=19pt] at (-4,0) (v3){$[1,3]$};
 \node[draw, rectangle,  minimum size=19pt] at (2,-2) (v4){$[2,1]$};
 \node[draw, rectangle,  minimum size=19pt] at (4,0) (v5){$[3,1]$};
 \node[draw, rectangle,  minimum size=19pt] at (0,0) (v6){$[2,2]$};
 \node[draw, rectangle, minimum size=19pt] at (-2,2) (v7){$[2,3]$};
  \node[draw, rectangle, minimum size=19pt] at (2,2) (v8){$[3,2]$};
   \node[draw, rectangle,  minimum size=19pt] at (0,4) (v9){$[3,3]$};

 \draw[line width = 0.3mm, >=latex] (v1) to (v2);
 \draw[line width = 0.3mm,  >=latex] (v1) --(v4);
 \draw[line width = 0.3mm,  >=latex] (v2) to (v3);
 \draw[line width = 0.3mm,  >=latex] (v2) to (v6);
 \draw[line width = 0.3mm,  >=latex] (v4) -- (v6);
 \draw[line width = 0.3mm, >=latex] (v2) to (v6);
 \draw[line width = 0.3mm, >=latex] (v3) to (v7);
 \draw[line width = 0.3mm, >=latex] (v4) -- (v5);
   \draw[line width = 0.3mm, >=latex] (v6) to (v7);
  \draw[line width = 0.3mm, >=latex] (v6) to (v8);
  \draw[line width = 0.3mm, >=latex] (v5) to (v8);
 \draw[line width = 0.3mm, >=latex] (v7) to (v9);
  \draw[line width = 0.3mm, >=latex] (v8) to (v9);
\end{tikzpicture}
\caption{Hasse diagram for $M=2, R=4$ for the treatment profile where both arms are active. A line upwards from treatment combinations $[r_1,r_2]$ to $[r_1',r_2']$ means that $[r_1,r_2] \leq [r_1',r_2']$ and $[r_1,r_2] \neq [r_1',r_2']$ in the intensity ordering.}
  \label{fig:hasse}
 \vfill
 \end{center}
\par\end{centering}
\end{figure}


\begin{figure}[!ht]
\centerfloat
\includegraphics[scale=0.5]{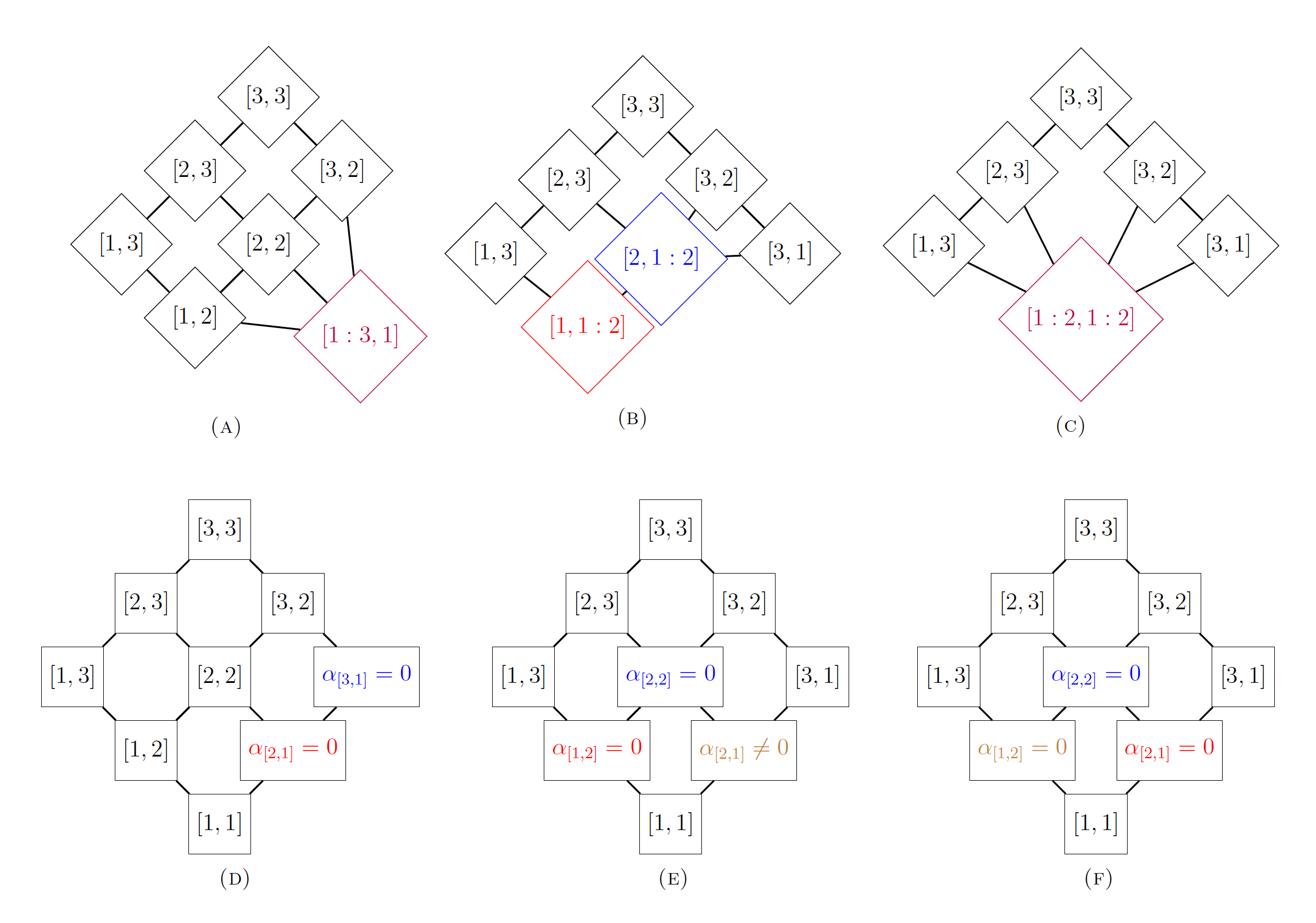}

\caption{Panels A, B, C show examples of policy-concatenations in the $\eta$-space while panels D, E, F show the zeros in the marginal effects space, corresponding to these concatenations. On \textit{panel A}, policies $[1,1], [2,1], [3,1]$ are concatenated since $\beta_{[1,1]} = \beta_{[2,1]} = \beta_{[3,1]}$. Correspondingly, panel D shows $\alpha_{[2,1]} = \beta_{[2,1]} - \beta_{[1,1]}$ and $\alpha_{[3,1]} = \beta_{[3,1]} - \beta_{[2,1]}$. On \textit{panel B}, policies $\{[1,1],[1,2]\}$ and $\{[2,1], [2,2]\}$ are concatenated since $\beta_{[1,1]} = \beta_{[1,2]}$ and $\beta_{[2,1]} = \beta_{[2,2]}$. The equivalent marginal space (panel E) shows, $\alpha_{[2,1]} = \beta_{[2,1]} - \beta_{[1,1]}$, $\alpha_{[1,2]} = \beta_{[1,2]} - \beta_{[1,1]}$, and $\alpha_{[2,2]} = \left(\beta_{[2,2]} - \beta_{[2,1]}\right) - \left(\beta_{[1,2]} - \beta_{[1,1]}\right)$. On \textit{panel C}, $\beta_{[1,1]} = \beta_{[1,2]} = \beta_{[2,1]} = \beta_{[2,2]}$. In panel F note that the only change relative to panel B is that here $\alpha_{[2,1]} = 0$.}\label{fig:hasse_concatenations}
\end{figure}


\begin{figure}[!h]
\begin{center}
	\begin{subfigure}[Panel A]{.4\linewidth}
		\includegraphics[scale=0.43]{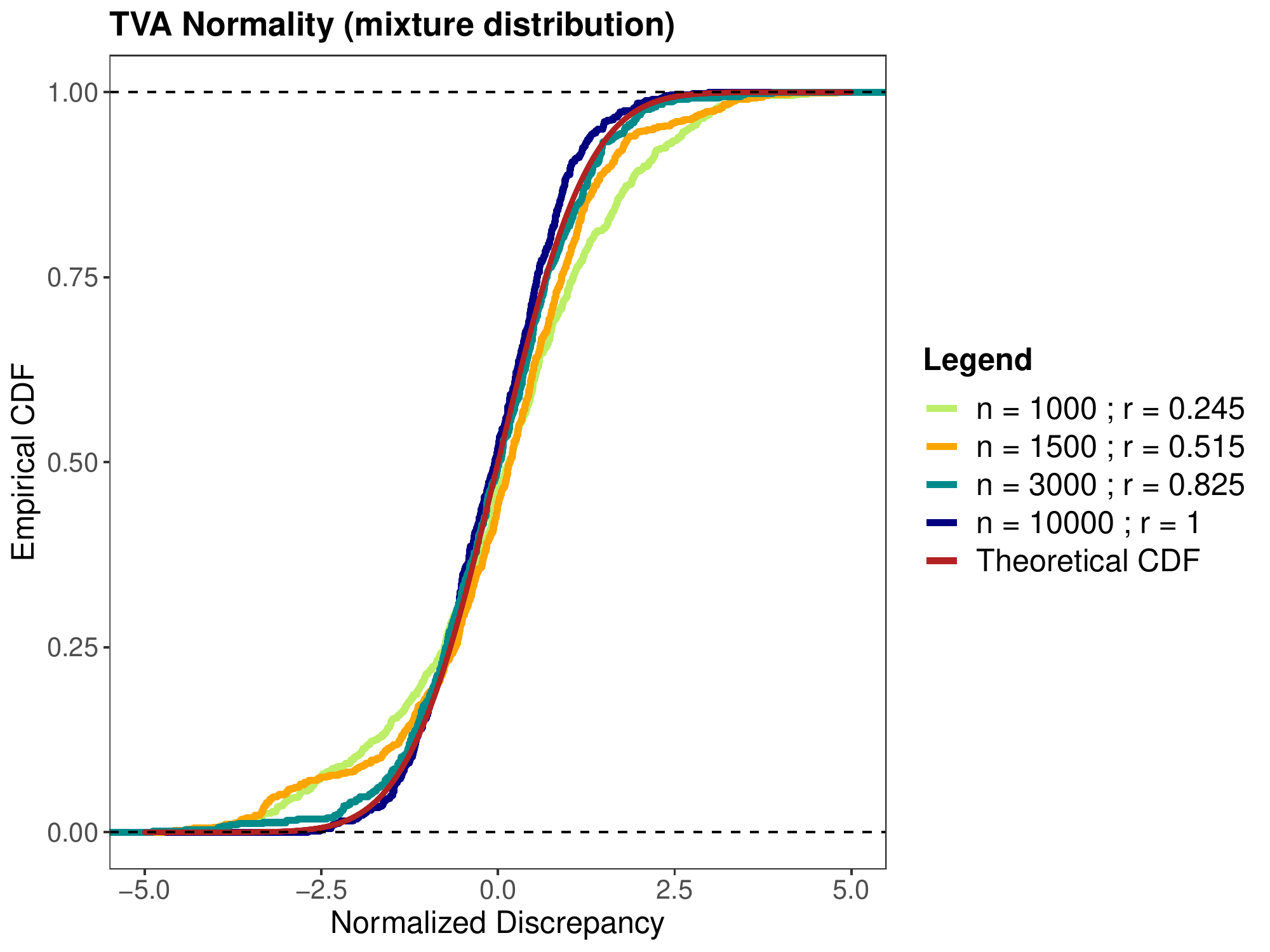}
		\caption{TVA: Normality}
	\end{subfigure}
	\qquad
	\hfill
	\begin{subfigure}[Panel B]{.4\linewidth}
		\includegraphics[scale=0.43]{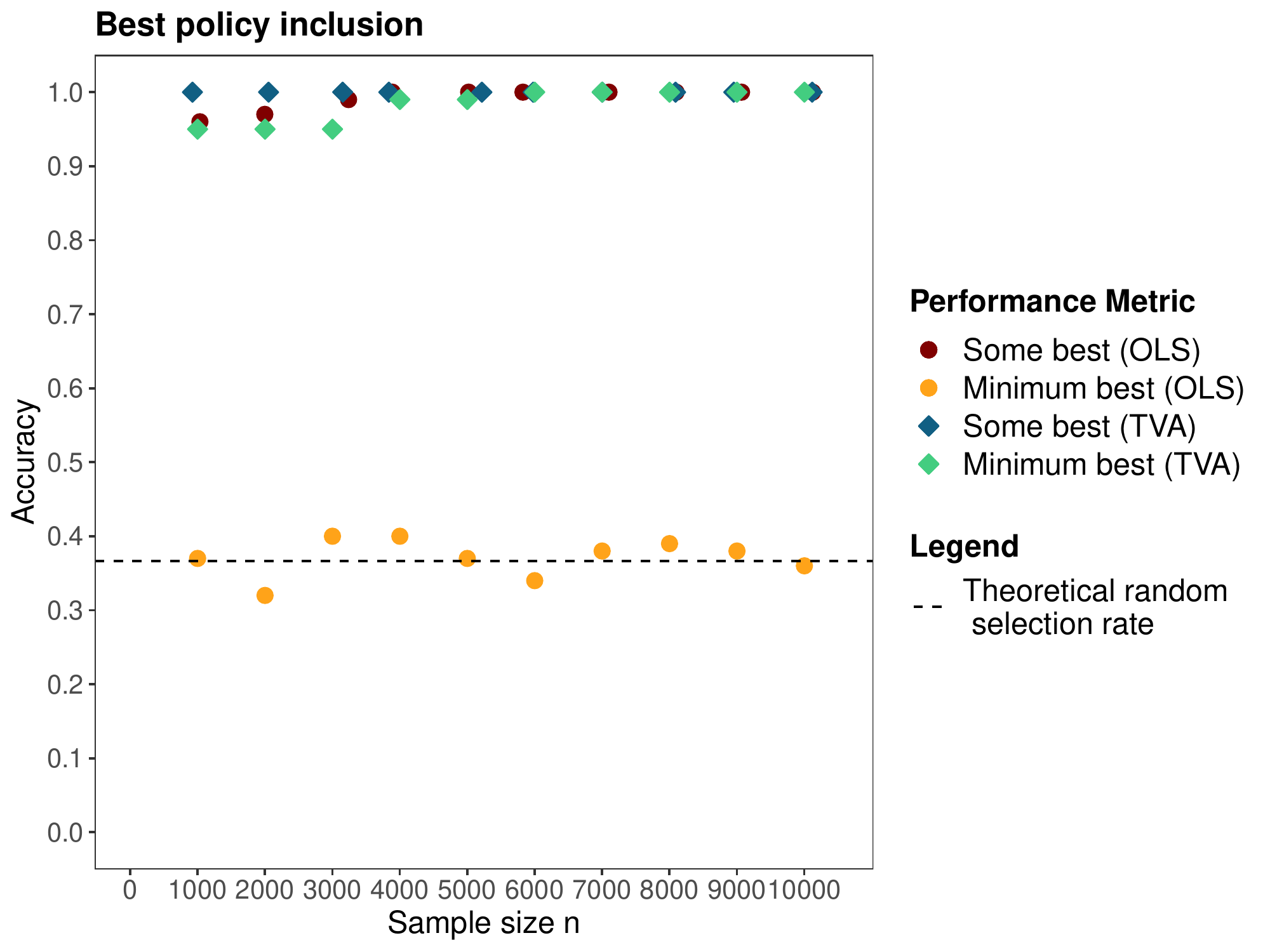}
		\caption{TVA vs OLS: Best Policy Inclusion}
	\end{subfigure}
\end{center}
\begin{center}	
	\begin{subfigure}[Panel C]{.4\linewidth}
		\includegraphics[scale=0.43]{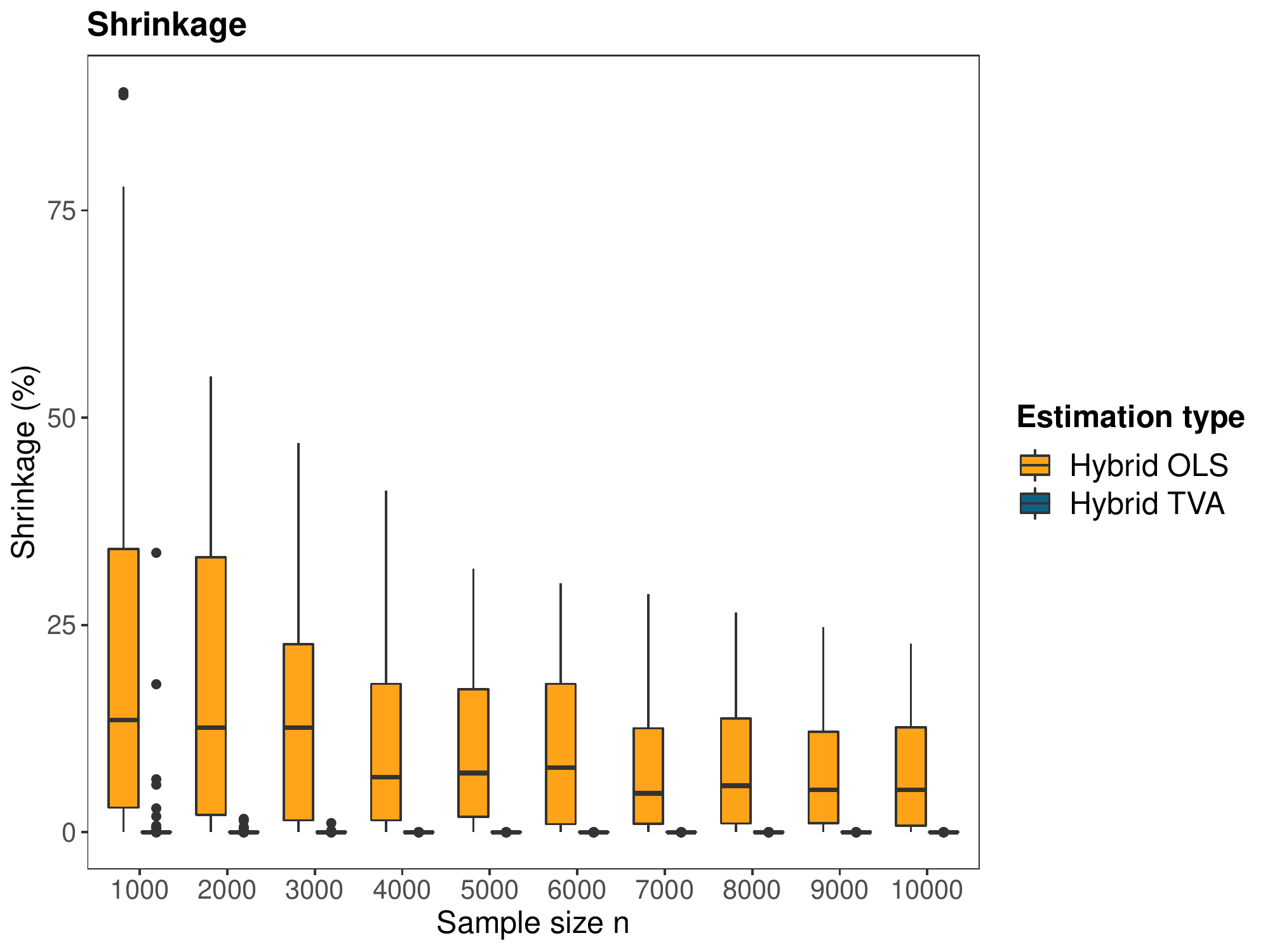}
		\caption{TVA vs OLS: Shrinkage}
	\end{subfigure}
	\qquad
	\hfill
	\begin{subfigure}[Panel D]{.4\linewidth}
		\includegraphics[scale=0.43]{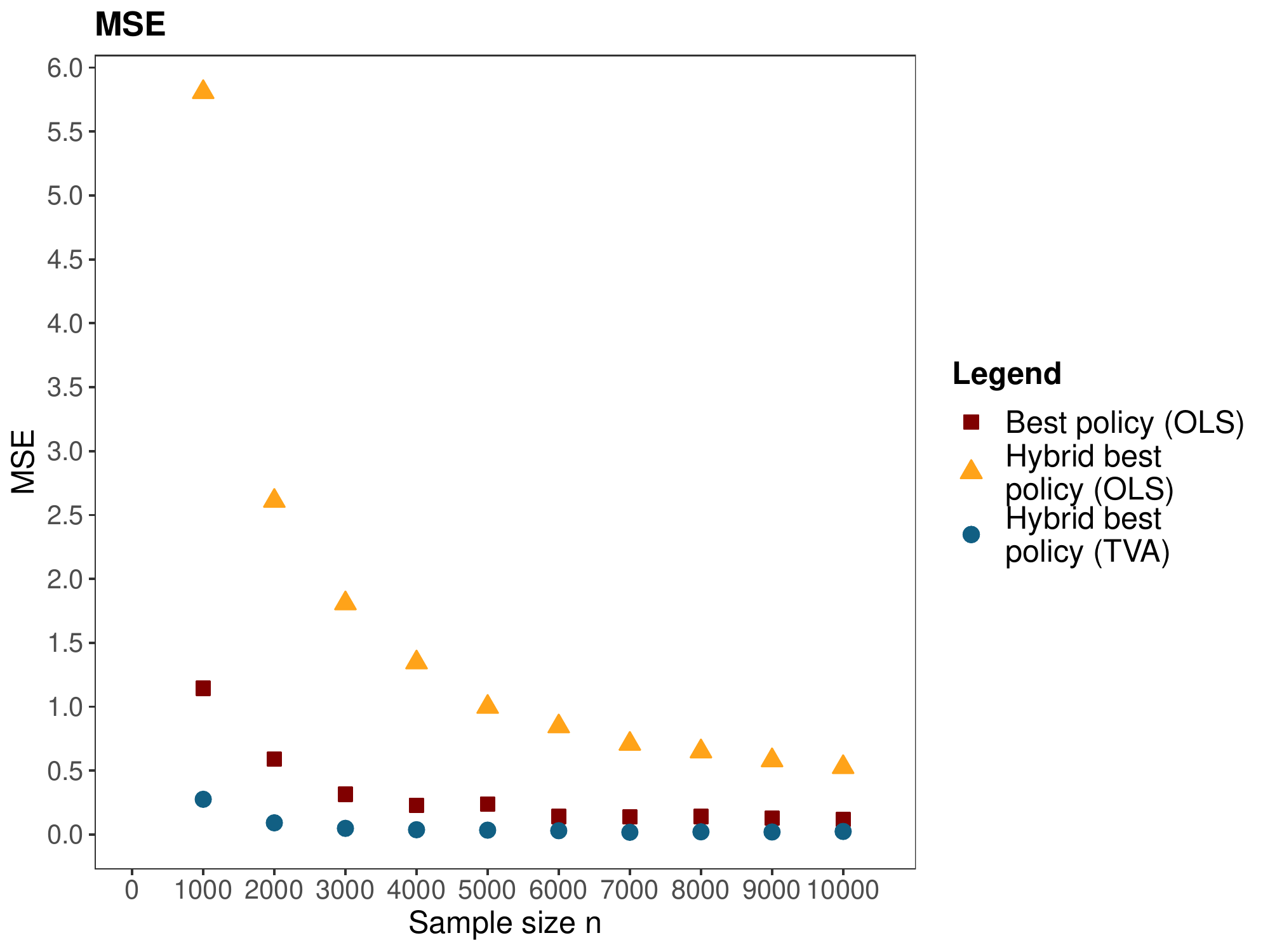}
		\caption{TVA vs OLS: MSE}
	\end{subfigure}
	
\caption{A plot comparing the performance of the TVA estimator to applying OLS on the unique policy specification (\ref{eq:unique_policy}) for a range of measures.  On panel A, we first expose the normality of TVA estimates ($r$ is correct support selection rate). Panel B then uses the best policy inclusion measures defined in subsection \ref{sec: performance_measures} and points are slightly jittered for better readability. For OLS, this measure is set to $1$ whenever the highest treatment effect policy is part of the true best pool (some best) or equal to the minimum dosage best policy (min best).  Panel C compares the amount of shrinkage imposed by the winner's curse adjustment as percentage of the initial coefficient. Panel D compares the MSE of the best policy estimation, between the TVA  and the OLS estimator of the unique policy specification (\ref{eq:unique_policy})  before and after adjusting for the winner's curse. In all panels, there are 20 simulations per configuration and 5 configurations per $n$. \label{fig:puffer_vs_ols}\label{fig:penalty_sensitivity}}
\end{center}
\end{figure}

\begin{center}
\begin{figure}[!h]

	 \begin{subfigure}[OLS_density]{\linewidth}
	 \centerfloat
		\includegraphics[scale=0.42]{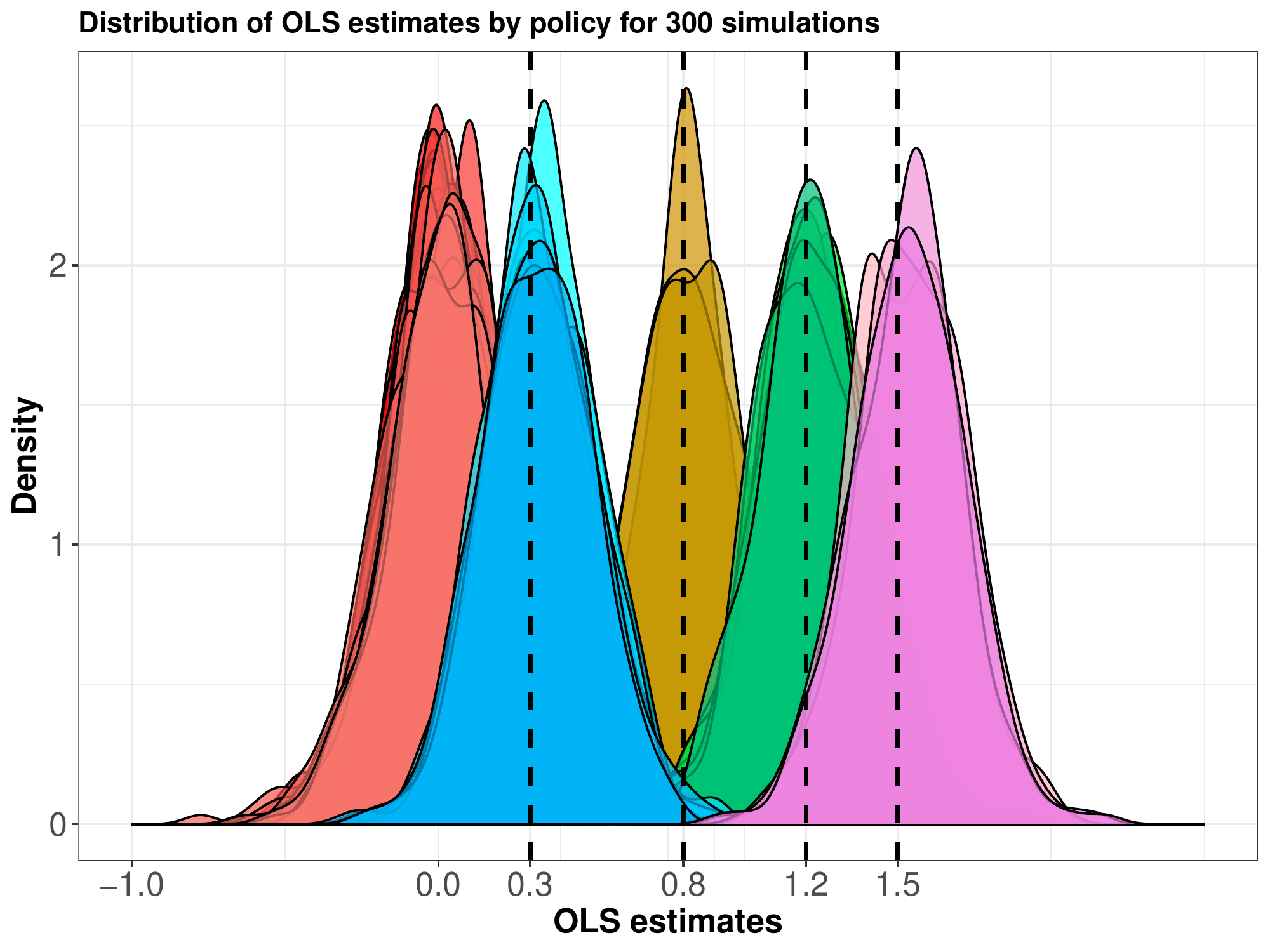}
		\qquad
		\hspace{-.9cm}
		\includegraphics[scale=0.42]{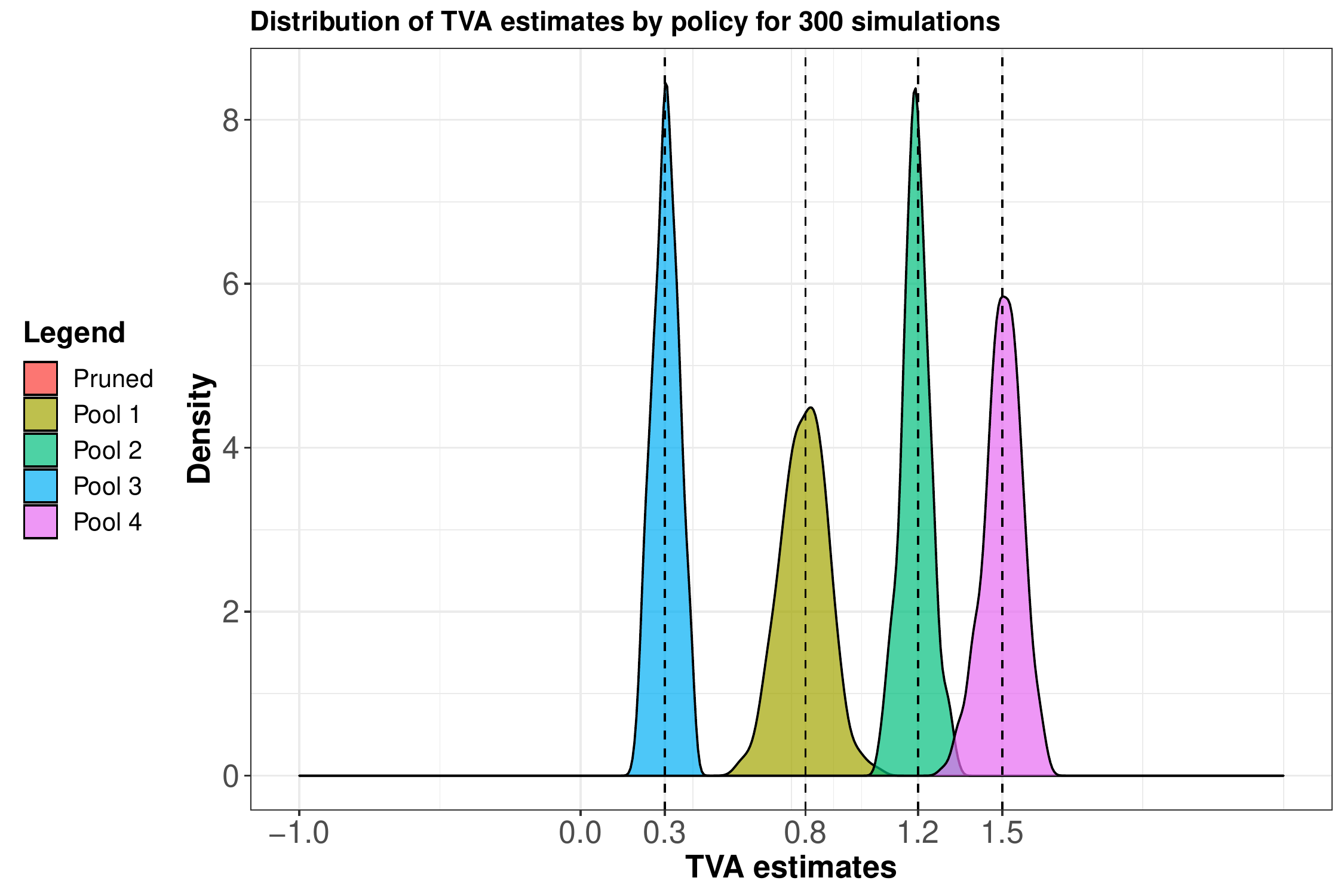}
	\caption{Densities (300 simulations)}
	\end{subfigure}

	\vspace{1cm}
	\begin{subfigure}[OLS_coefplot]{\linewidth}
	\centerfloat
		\includegraphics[scale=0.38]{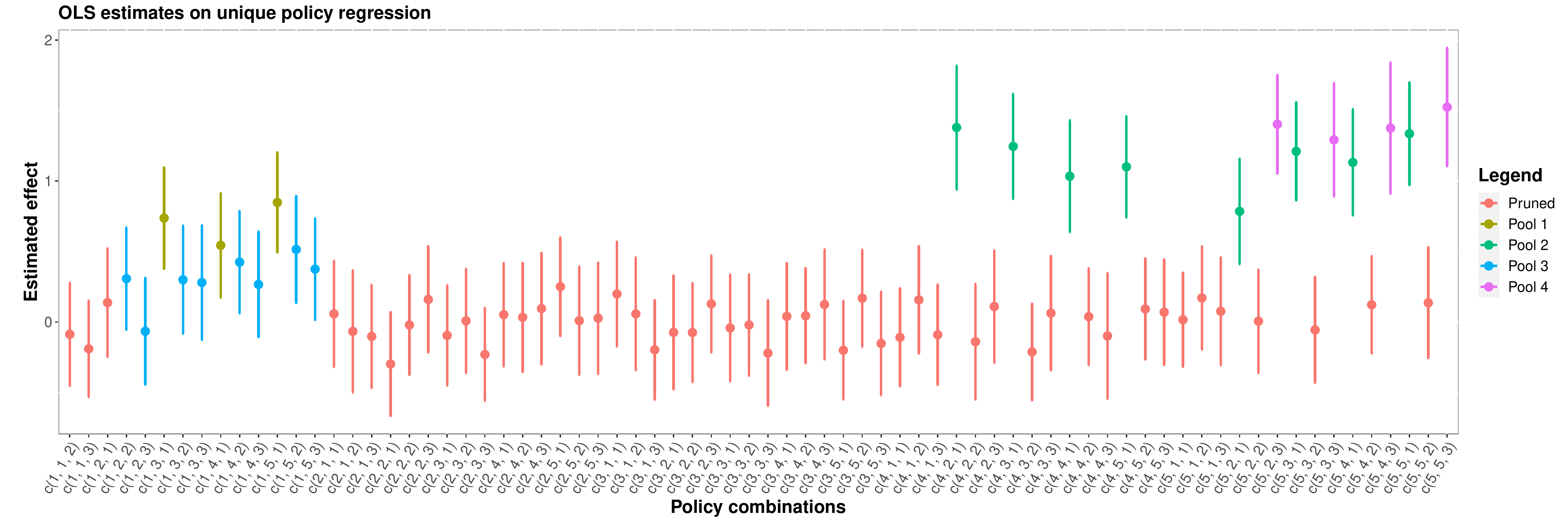}
		
	\end{subfigure}
	
	\begin{subfigure}[puffer_coefplot]{\linewidth}
	\centerfloat
		\includegraphics[scale=0.38]{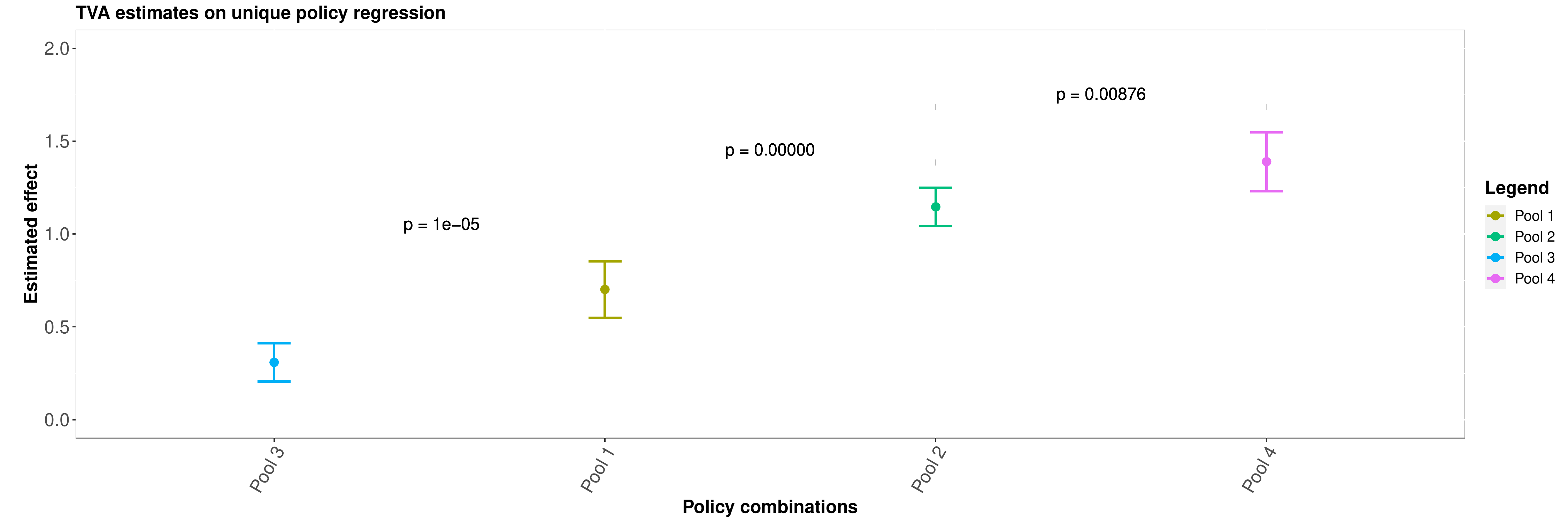}
		\caption{Estimated Coefficients}
	\end{subfigure}
\captionsetup{width=\linewidth}
\caption{This figure presents results from a simulation setting with $n = 4000$ and where 4 pooled policies (each composed of 3-8 unique policies) are non-zero with effects 0.3, 0.8, 1.2 and 1.5 respectively.  For OLS applied to (\ref{eq:unique_policy}) and TVA, we show both the distribution of policy estimates across 300 simulations (panel A) and the estimated policy coefficients for one representative simulated draw (panel B).  Note that the color labeling for OLS plots corresponds to the true underlying pooled and pruned policies. For the TVA density plots, we condition on the event that the TVA estimator has selected the correct support (mean support accuracy is $92.2\%$ across simulations). \label{fig:power_comparison}}
\end{figure}
\end{center}

\begin{center}
\begin{figure}[!h]
	\begin{subfigure}[Panel A]{.45\linewidth}
		\includegraphics[scale=0.39]{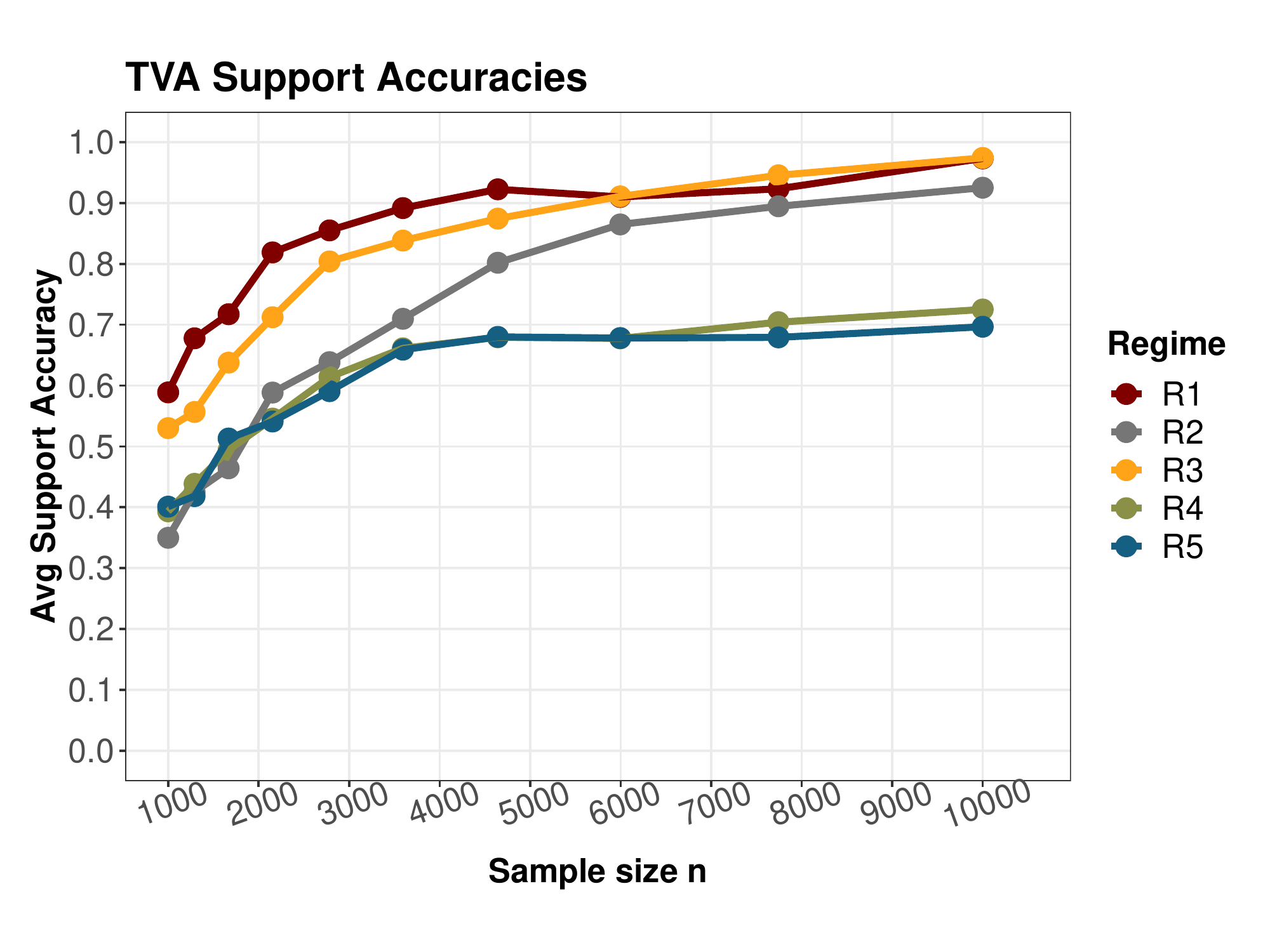}
		\caption{Support Accuracy}
	\end{subfigure}
\qquad
\hfill
	\begin{subfigure}[Panel A]{.45\linewidth}
		\includegraphics[scale=0.39]{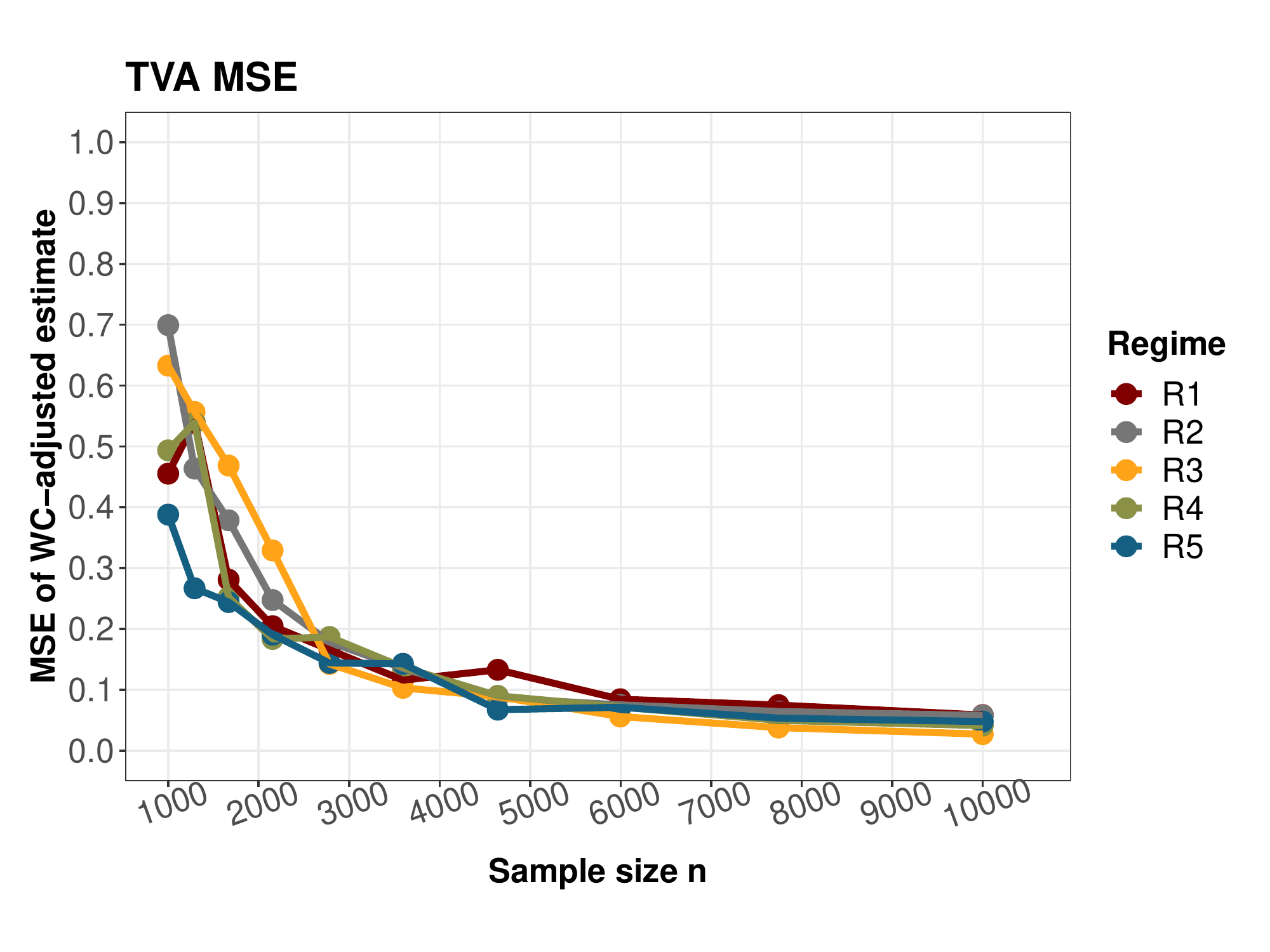}
		\caption{MSE of Best Policy}
	\end{subfigure}

	\begin{subfigure}[Panel A]{.45\linewidth}
	\vspace{2\baselineskip}
		\includegraphics[scale=0.39]{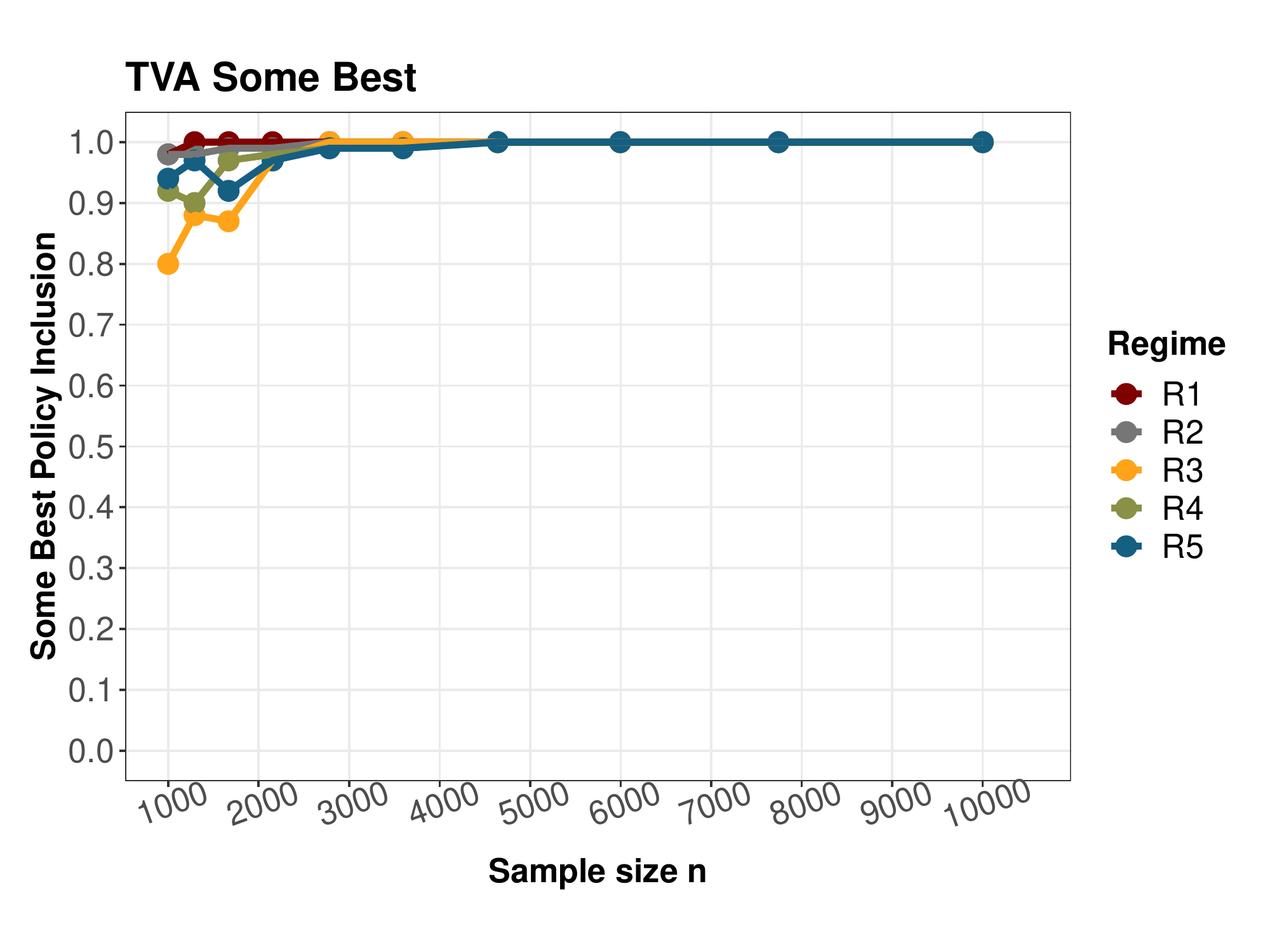}
		\caption{Some Best Policy Inclusion}
	\end{subfigure}
\qquad
\hfill
	\begin{subfigure}[Panel A]{.45\linewidth}
	\vspace{2\baselineskip}
		\includegraphics[scale=0.39]{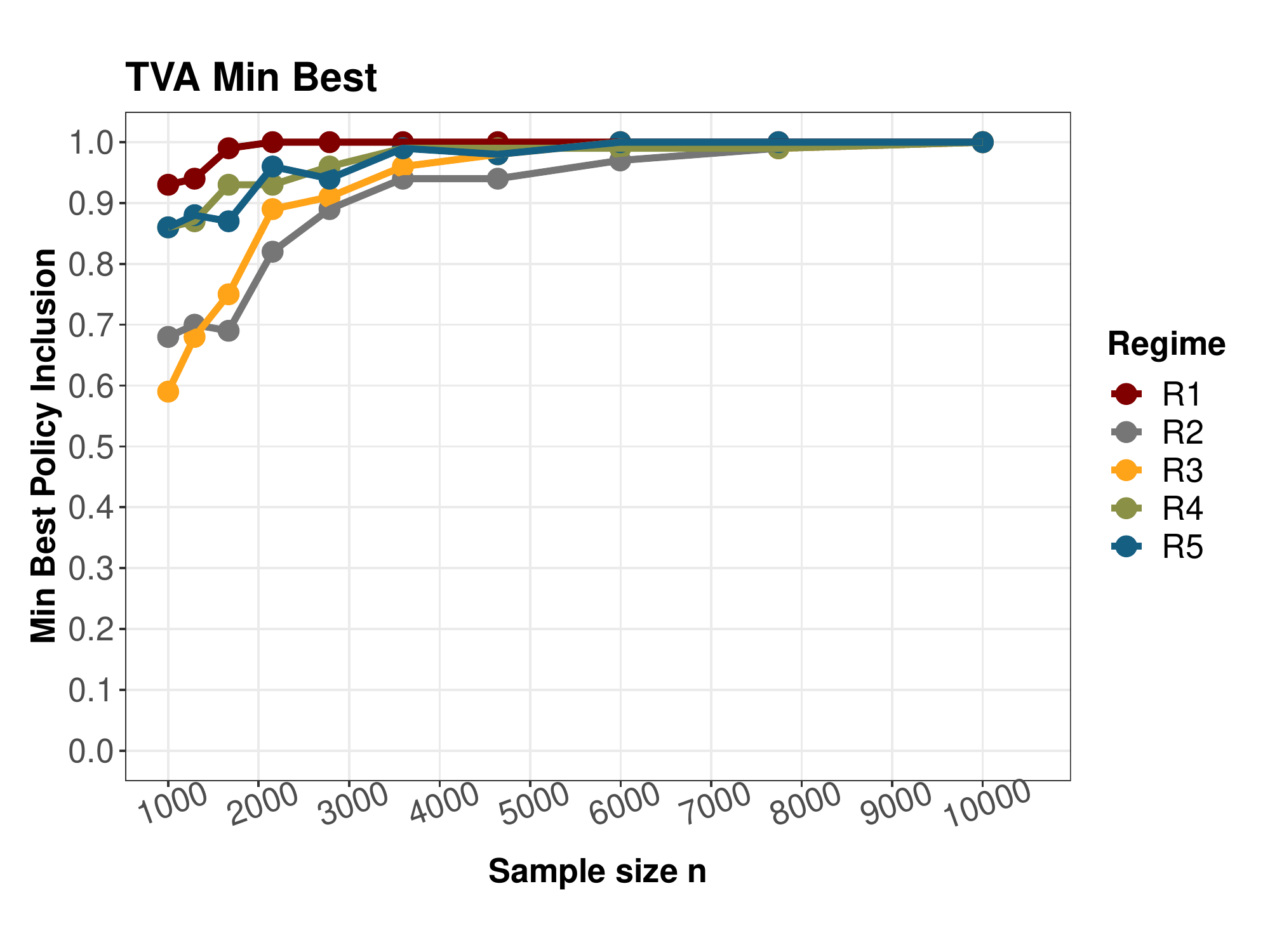}
		\caption{Min Best Policy Inclusion}
	\end{subfigure}
\end{figure}

\begin{figure}[!h]
\ContinuedFloat
\caption{This figure shows the performance of the TVA estimator under relaxations of the sparsity assumptions. We show support accuracy (panel A), MSE of the best policy (panel B), Some Best policy inclusion rate (panel C) and Minimum Best policy inclusion rate (panel D) as a function of $n$ for five different levels of violating sparsity requirements (R1-R5). Per regime, there are $20$ simulations per support configuration per $n$, for five support configurations. }\label{fig:regime_plots}
\end{figure}
\end{center}

\clearpage


\begin{figure}[!h]
\centerfloat
	\includegraphics[scale=0.61]{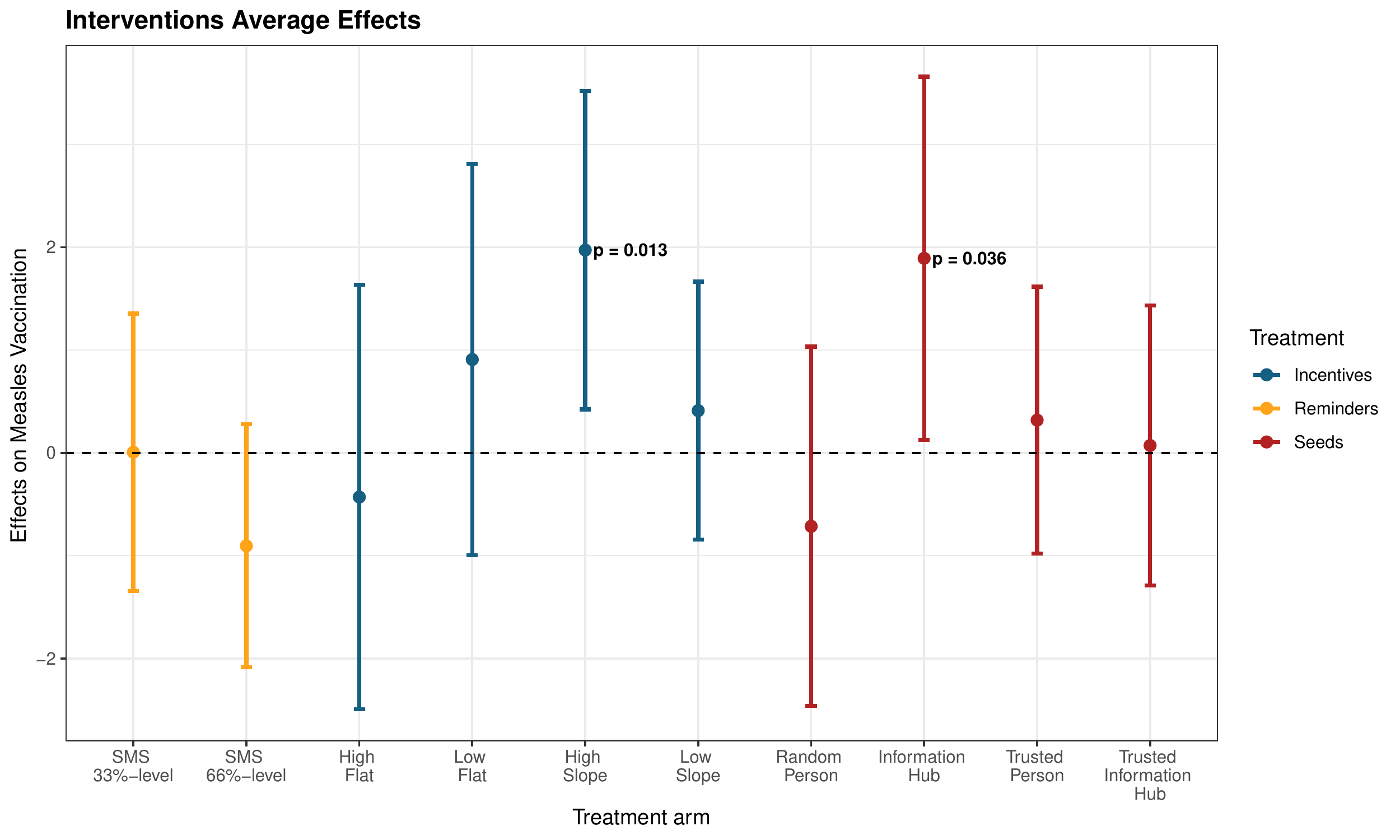}
\caption{Effects on the number of measles vaccinations relative to control (7.32) by reminders, incentives, and seeding policies, restricted to the villages were the ambassador intervention was administered. The specification is weighted by village population, controls for district-time fixed effects, and clusters standard errors at the sub-center level.\label{fig:agg-policy-effects}
}
\end{figure}

\clearpage

\begin{figure}[!h]

	 \begin{subfigure}[Panel A]{\linewidth}
		\includegraphics[scale=0.12]{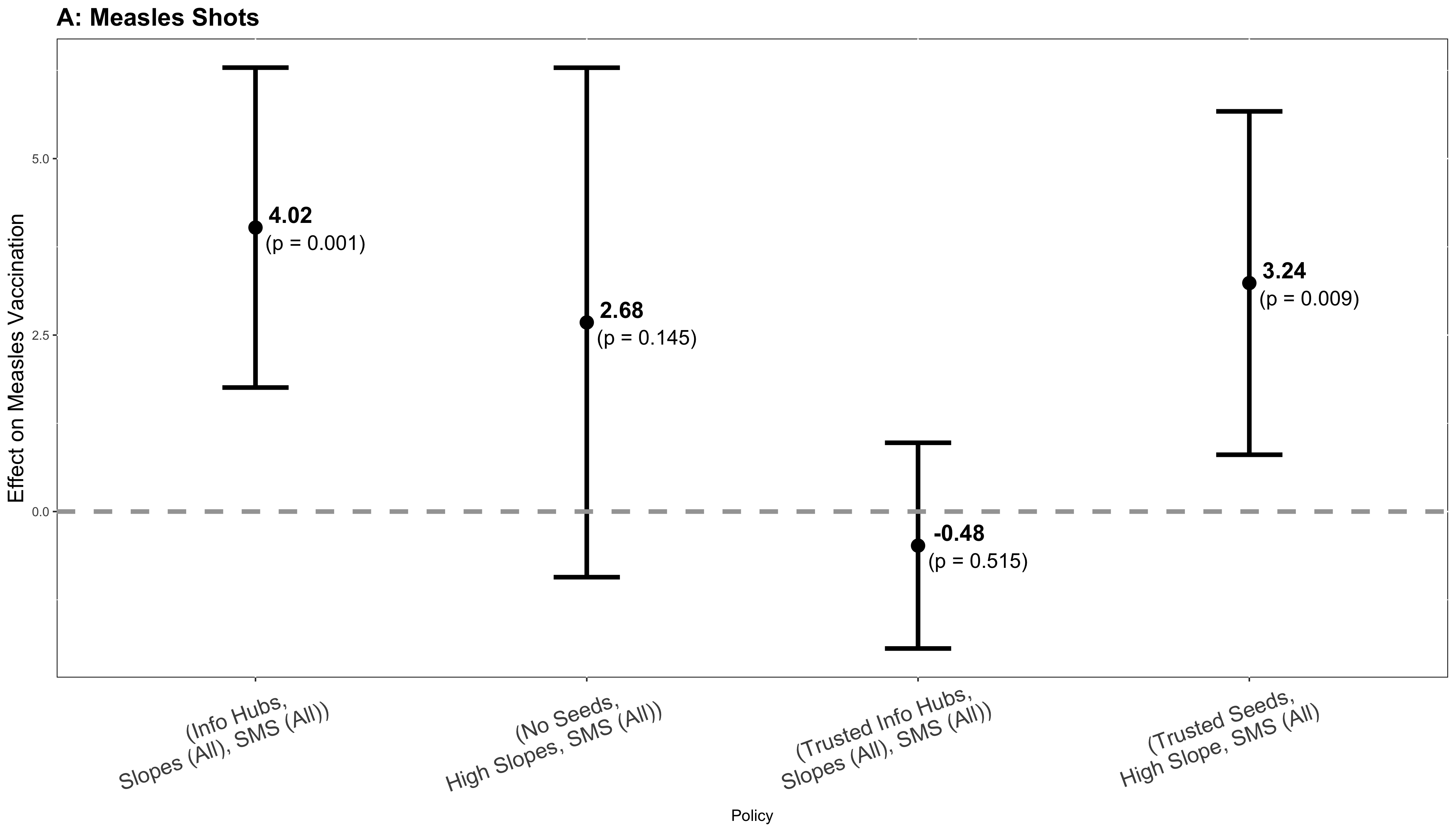}
	\end{subfigure}
	
	\begin{subfigure}[Panel B]{\linewidth}
		\includegraphics[scale=0.12]{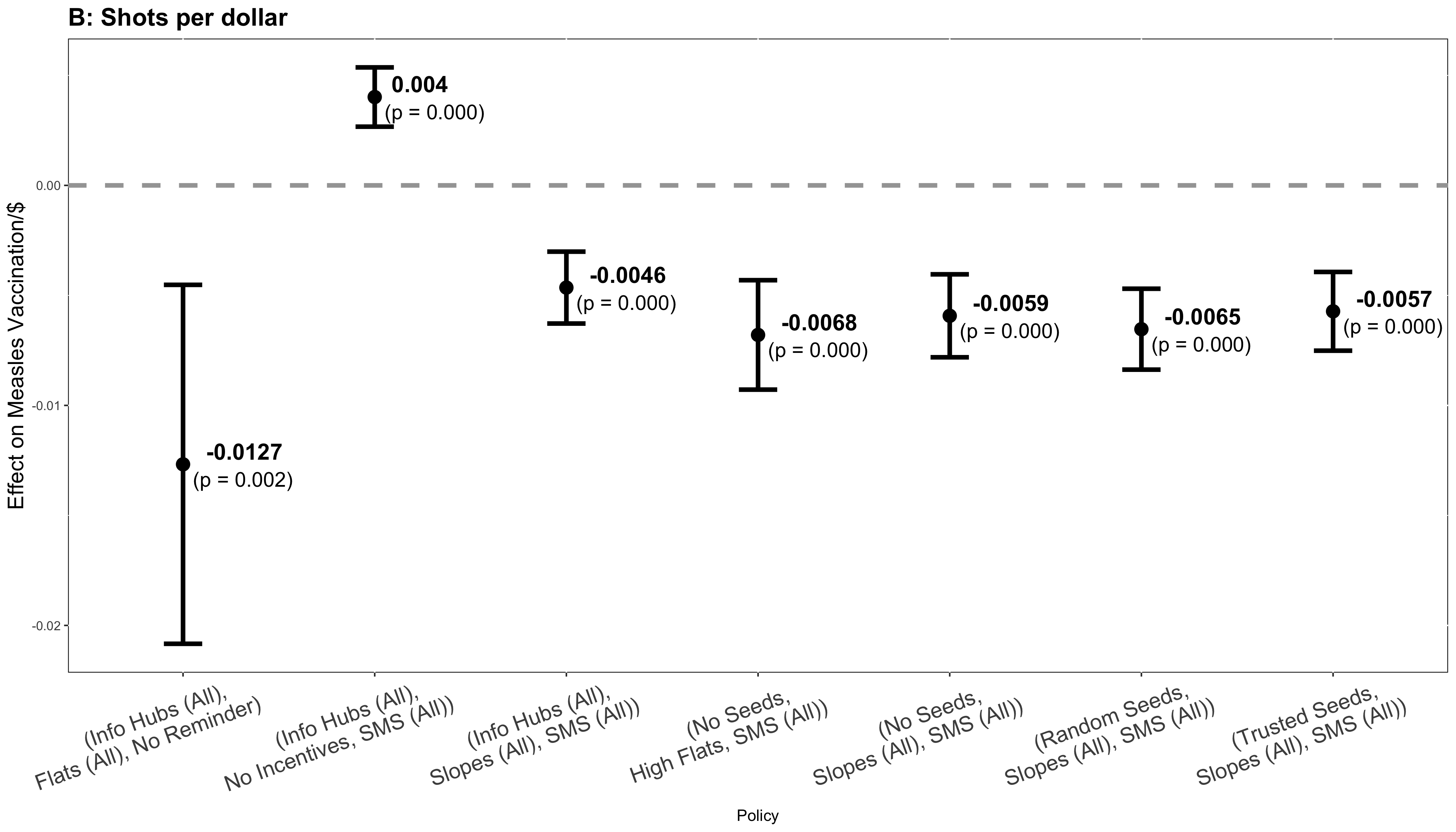}
	\end{subfigure}
	
\caption{TVA estimates of combinations of reminders, incentives, and seeding policies on the number of measles vaccines (Panel A) and the number of measles vaccines per \$ (Panel B) relative to control (7.32 and 0.0436 respectively). The specifications are weighted by village population and include controls for district-time fixed effects. Standard errors are clustered at the sub-center level. 95\% confidence intervals displayed.\label{fig:postlasso_combined}}
\end{figure}

\clearpage

\section*{Tables}


\begin{table}[!h]
	\centering
	\caption{Best Policies}\label{tab:max}
	\scalebox{0.8}{\begin{threeparttable}
			\begin{tabular}{lcccc} \hline
 & (1) & (2)  \\
  & \# Measles Shots & \# Measles Shots per \$1 \\ \hline
 &  &  &  &  \\
WC Adjusted Treatment Effect & 3.26   & 0.004  \\
 Confidence Interval (95\%) & [0.32,6.25]   &  [0.003, 0.005]  \\
Control Mean & 7.32 & 0.0435 \\
Observations & 204  & 814  \\
Optimal Policy & (Information Hubs, SMS, Slope) & (Information Hubs POOLED, SMS) \\ \hline
\end{tabular}


			\begin{tablenotes}
				Notes: Estimation using \cite{andrews2019inference}; hybrid estimation with $\alpha = 0.05, \beta = 0.005$. The specifications are weighted by village population and account for district-time fixed effects as well as variance clustered at the sub-center level.
			\end{tablenotes}
	\end{threeparttable}}
\end{table}

\clearpage

\appendix

\section{Proofs}\label{sec:proofs}

\

\begin{proof}  [Proof of Proposition \ref{prop: support}] According to Theorem 1 of \cite{jia2015preconditioning}, if $\min_{j \in S_\alpha} |\alpha_j| \geq 2 \lambda_n$, then $\widetilde{\alpha} =_s \alpha$ with probability greater than
	
	\begin{equation*}
		f(n) := 1- 2K\exp{ \big( - \frac{n \lambda_n^2 \xi^2_{\min}}{2 \sigma^2} \big)},
	\end{equation*}
	where   $\xi_{\min} = \xi_{\min}(\frac{X}{\sqrt{n}})$ is the minimum singular value of the $\sqrt{n}$-normalized design matrix.
	\noindent By Assumption \ref{assu: support}, there is a uniform lower bound $c > 0$ (independent of $n$) on the magnitude of the non-zero $\{\alpha\}$.   Since by Assumption \ref{assu: penalty}, $\lambda_n \rightarrow 0$, for sufficiently high $n$ $\min_{j \in S_\alpha} |\alpha_j| \geq 2 \lambda_n$. Theorem 1 applies and $\mathrm{sign}(\hat{\alpha}) =\mathrm{sign}( \alpha)$ with probability greater than $f(n)$. \\

 It will be convenient to re-express $f(n)$ as follows:
	\[
		f(n) = 1- 2\exp{ \big(\log(K)- \frac{n \lambda_n^2 \xi^2_{\min}}{2 \sigma^2} \big)}.
\]
	And applying Lemma \ref{lem:cmin}, for sufficiently high $n$:
	\[
		f(n) \geq 1- 2\exp{ \big(\log(K) - \frac{n \lambda_n^2}{2 \sigma^2 K^2} \big)}.
\]
	Per Assumption \ref{assu: design}, $K = O(n^\gamma)$, it follows that for sufficiently high $n$:
	\begin{align*}
		f(n) \geq 1- 2\exp{ \big(\gamma \log(n)- \frac{n^{1- 2 \gamma} \lambda_n^2}{2 \sigma^2} \big)}.
	\end{align*}

	\noindent By Assumption \ref{assu: design}, $0 < \gamma <\frac{1}{2} \implies n^{1 - 2 \gamma} = \omega(\log(n))$ and by Assumption $\ref{assu: penalty}$, since $\lambda_n^2 n^{1 - 2 \gamma} = \omega(\log(n))$, it follows that $\lim_{n\rightarrow \infty} f(n) \geq 1$. Since also $f(n) \leq 1$, it follows that $f(n) \rightarrow 1$ and the proof is complete. \end{proof}

\

\begin{lem} \label{lem:cmin} For the marginal effects design matrix $X$, for $R \geq 3$, wpa1 the lowest singular value of $\sqrt{n}$ normalized design matrix, i.e., {$\xi_{\min}(\frac{X}{\sqrt{n}})$}, has the value $\xi_{\min} (\frac{X}{\sqrt{n}}) = \Big( 4 R \sin ^2 \big( \frac{R-\frac{3}{2}}{R - \frac{1}{2}} \frac{\pi}{2} \big) \Big)^{-\frac{M}{2}}$. Thus, with probability approaching 1, $ \xi_{\min}(\frac{X}{\sqrt{n}}) > (\frac{1}{K})$.\footnote{This is a conservative bound; the optimal uniform lower bound is  $\big(\frac{1}{K} (\frac{1}{4^M})\big)^\frac{1}{2}$.}

\end{lem}

\begin{proof}[Proof of Lemma \ref{lem:cmin}] It will be useful to index the design matrix $X$ by $R$ and $M$, i.e., $X = \mathbf{X}_{R,M}$. Let $\mathbf{C}_{R,M} = \lim_{n \rightarrow \infty} \frac{1}{n} \mathbf{X}_{R,M}' \mathbf{X}_{R,M}$. Then $\lim_{n \rightarrow \infty}{\xi^2_{\min}}(\frac{\mathbf{X}_{R,M}}{\sqrt{n}}) = \lambda_{\min}(\mathbf{C}_{R,M})$, i.e., the lowest eigenvalue of $\mathbf{C}_{R,M}$. We will characterize this eigenvalue. \\
	
	\noindent The combinatorics of the limiting frequencies of ``1"s in marginal effects variables imply that $\mathbf{C}_{R,M}$ is a block diagonal matrix with structure
	
	\begin{equation*}\mathbf{\mathbf{C}_{R,M}} = \frac{1}{K} \begin{bmatrix}
			\mathbf{B}_{R,M} & 0             & \cdots & 0            \\
			0            & \mathbf{B}_{R,M-1}  & \cdots & 0            \\
			\vdots       & \vdots        & \ddots & \vdots       \\
			0            & 0             & \cdots & \mathbf{B}_{R,1}
		\end{bmatrix}
	\end{equation*}
	
	\noindent where $\mathbf{B}_{R,M-1}$ implies this block is found in $\mathbf{C}_{R,M-1}$ (pertaining to an RCT with one less cross-treatment arm), etc. More than one block of $\mathbf{B}_{R,M-1}$ , $\mathbf{B}_{R,M-2}$, ... $\mathbf{B}_{R,1}$ is found in $\mathbf{\mathbf{C}_{R,M}}$, but only $\mathbf{\mathbf{B}_{R,M}}$ determines
	the minimum eigenvalue.
	
	The combinatorics of variable assignments also implies that
	
	\begin{enumerate}
		\item  $\mathbf{B}_{R,M}$ is an $(R-1)^M \times (R-1)^M$ matrix with recursive structure $\mathbf{B}_{R,M} = \mathbf{B}_{R,1} \otimes  \mathbf{B}_{R,M-1}$, where $\otimes$ is the Kronecker product.\footnote{Thanks to Nargiz Kalantarova for noticing this Kronecker product and its consequent implication for $\lambda_{\min}(\mathbf{B}_{R,M})$.}
		\item $\mathbf{\mathbf{B}_{R,1}}$ is an $(R-1) \times (R-1)$ matrix with recursive structure
		
		$\mathbf{\mathbf{B}_{R,1}} = \begin{bmatrix}
			R-1 & R-2             & \cdots & 1           \\
			R-2            &  & &            \\
			\vdots            &  & \mathbf{B_{R-1,1}}  &             \\
			1            &              & &
		\end{bmatrix}$ and $\mathbf{B_{2,1}} = [1].$
	\end{enumerate}
	
	\
	
	\begin{sublemma}\label{sublem:BR1} $\lambda_{\min}(\mathbf{\mathbf{B}_{R,1}}) =  \Big( 4 \sin ^2 \big(\frac{R-\frac{3}{2}}{R - \frac{1}{2}} \frac{\pi}{2} \big) \Big)^{-1}$
	\end{sublemma}
	
	\begin{proof} The  key insight of the argument\footnote{The argument is provided on Mathematics Stackexchange \citep{stackexchange}.}
		  is that ${\mathbf{B}_{R,1}}^{-1}$ is the $(R-1) \times (R-1)$ tridiagonal matrix:
		
		\begin{equation*}
			{\mathbf{B}_{R,1}}^{-1}=\begin{bmatrix}
				1&-1\\
				-1&2&-1\\
				&-1&\ddots&\ddots\\
				&&\ddots&\ddots&-1\\
				&&&-1&2
			\end{bmatrix}
		\end{equation*}
		Which has known eigenvalues $\mu_j=4\sin^2\left(\frac{j-\frac12}{R-\frac12}\frac\pi2\right)$ for $j = 1,2,...,R-1$. Thus, given that the inverse of a matrix's eigenvalues are the inverse matrix's eigenvalues, $\lambda_{\min}(\mathbf{\mathbf{B}_{R,1}}) =  \Big( 4 \sin ^2 \big( \frac{R-\frac{3}{2}}{R - \frac{1}{2}} \frac{\pi}{2} \big) \Big)^{-1}$.
	\end{proof}
	
	\
	
	\noindent (Resuming the proof of Lemma \ref{lem:cmin}) Per the multiplicative property of the eigenvalues of a Kronecker product, together with the fact that all matrices in question are positive definite, it immediately follows that $\lambda_{\min}(\mathbf{B}_{R,M}) = \lambda_{\min}(\mathbf{B}_{R,1}) \lambda_{\min} (\mathbf{B}_{R,M-1})$, which in turn implies $\lambda_{\min}(\mathbf{B}_{R,M}) = \lambda_{\min}(\mathbf{B}_{R,1})^M$. Since by Sublemma \ref{sublem:BR1}, $\lambda_{\min}(\mathbf{\mathbf{B}_{R,1}}) < 1$, $\mathbf{B}_{R,M}$ is the
	block determining the rate  with the smallest eigenvalue, and therefore, given that the eigenvalues of a block diagonal matrix are the eigenvalues of the blocks:
	
	\begin{align*}
		\lambda_{\min} (\mathbf{C}_{R,M}) = \frac{1}{K} \lambda_{\min}(\mathbf{B}_{R,M})
		= \frac{1}{K} (\lambda_{\min}(\mathbf{B}_{R,1}))^M
		= \Big( 4 R \sin ^2 \big(\frac{R-\frac{3}{2}}{R - \frac{1}{2}}  \frac{\pi}{2} \big)  \Big)^{-M}
	\end{align*}
	
	\noindent where the last equality uses Sublemma \ref{sublem:BR1}. The Lemma follows.
\end{proof}

\

\begin{proof}[Proof of Proposition \ref{prop: TVA_consistency}]
The proof is found in \cite{javanmard2013model}, proof of Theorem
2.7, with minor modifications, which we reproduce for completeness.
Label the events $\mathcal{E}:=\left\{ \widehat{S}_{TVA}=S_{TVA}\right\} $
(that the treatment variants were aggregated correctly).  Define the pseudo-true value  $
\eta_{S}^{0}:=\argmin_{\eta}{\rm E}\left[\left\Vert y-Z_{S}\eta\right\Vert _{2}^{2}\right],
$
noting that $\eta_{S_{TVA}}^{0}$ satisfies this for $S=S_{TVA}$.  Finally, let $\mathcal{F}:=\left\{ \left\Vert \hat{\eta}_{\widehat{S}_{TVA}}-\eta_{\widehat{S}_{TVA}}^{0}\right\Vert _{\infty}>\epsilon\right\} $, so it is the event that the estimator exceeds the pseudo-true value
on the estimated support by $\epsilon$.

Then, we can write
\[
\Pr\left(\mathcal{F}\right)=\Pr\left(\mathcal{F}\cap\mathcal{E}\right)+\Pr\left(\mathcal{F}\cap\mathcal{E}^{c}\right)\leq\Pr\left(\mathcal{F}\cap\mathcal{E}\right)+\Pr\left(\mathcal{E}^{c}\right).
\]
By the proof of Proposition \ref{prop: support}, we have
\[
\Pr\left(\mathcal{E}^{c}\right)\leq2K\exp\left(-\frac{n\left(\lambda/K\right)^{2}}{2\sigma^{2}}\right) = 2\exp\left(\gamma\log n-\frac{n^{1-2\gamma}\lambda}{2\sigma^{2}}\right)
\]

Turning to $\Pr\left(\mathcal{F}\cap\mathcal{E}\right)$, on the event
$\mathcal{E}$,
\[
\hat{\eta}_{\widehat{S}_{TVA}}-\eta_{\widehat{S}_{TVA}}^{0}=\left(Z_{\widehat{S}_{TVA}}'Z_{\widehat{S}_{TVA}}\right)^{-1}Z_{\widehat{S}_{TVA}}'\epsilon,
\]
since $\widehat{S}_{TVA}=S_{TVA}$. Therefore, for every $j\in\left\{ 1,\ldots,K\right\} $,
$\hat{\eta}_{S_{TVA},j}-\eta_{S_{TVA},j}^{0}$ is normally distributed with
variance order bounded above by $\frac{\sigma^{2}}{n\cdot C_{\min}}$
where $C_{\min}:=\sigma_{\min}\left(n^{-1}Z_{S_{TVA}}'Z_{S_{TVA}}\right)$
is the minimum singular value of the design matrix. But $C_{\min}\geq \frac{1}{\sqrt{K}}$
by definition since each unit is assigned to a disjoint pooled policy, and each policy pools one or more variants. And so,
\begin{align*}
\Pr\left(\mathcal{F}\cap\mathcal{E}\right) & =\Pr\left( \left\Vert \hat{\eta}_{\widehat{S}_{TVA}}-\eta_{\widehat{S}_{TVA}}^{0}\right\Vert _{\infty}>\epsilon\cap\mathcal{E}\right)\\
 & \leq\Pr\left(\sup_{j}\left|\hat{\eta}_{S_{TVA},j}-\eta_{S_{TVA},j}^{0}\right|>\epsilon\right)\\
 & \leq2\exp\left(-\frac{n\cdot\epsilon^2\cdot K^{-1/2}}{2\sigma^{2}}\right),
\end{align*}
which uses a Gaussian tail bound and then a union bound for uniform
control over $j\in\left\{ 1,\ldots,\left|S_{TVA}\right|\right\} $.

Putting the pieces together we have
\[
\Pr\left(\mathcal{F}\right)\leq2\exp\left(-\frac{n\cdot\epsilon^2\cdot K^{-1/2}}{2\sigma^{2}}\right)+2\exp\left(\gamma\log n-\frac{n^{1-2\gamma}\lambda}{2\sigma^{2}}\right)
\]
This establishes that $\Pr(\mathcal{F}) \rightarrow 0$ for every $\epsilon$, i.e., the consistency of the estimator to the pseudo-true values. With probability $1-2\exp\left(\gamma\log n-\frac{n^{1-2\gamma}\lambda}{2\sigma^{2}}\right)\rightarrow1$, the event $\mathcal{E}$ is active and the pseudo-true value will be the true value. In this case consider $\epsilon=\sqrt{q\cdot\frac{\log n}{n^{1-\gamma/2}}\cdot2\sigma^{2}}$
for any positive $q$. Then
\begin{align*}
	\Pr\left(\mathcal{F}\cap\mathcal{E}\right) &= 2\exp\left(-\frac{n\epsilon^2 C_{\text{min}}}{2\sigma^{2}}\right) \\
& \leq2\exp\left(-q\cdot\frac{\log n}{n^{1-\gamma/2}}\cdot2\sigma^{2}\cdot\frac{nK^{-1/2}}{2\sigma^{2}}\right)\\
	& \leq2\exp\left(-q\cdot\log n\right)\\
	& =\frac{2}{n^{q}}\rightarrow0,
\end{align*}
i.e., $\left\Vert \hat{\eta}_{\widehat{S}_{TVA}}-\eta_{S_{TVA}}^{0}\right\Vert _{\infty} < \sqrt{q\cdot\frac{\log n}{n^{1-\gamma/2}}\cdot2\sigma^{2}}$ with high probability, completing the proof.
\end{proof}

\

\begin{proof}[Proof of Proposition \ref{prop:normalinf}] As in the proof of Proposition \ref{prop: TVA_consistency}, let $\mathcal{E}:=\left\{ \widehat{S}_{TVA}=S_{TVA}\right\} $. We can decompose\footnote{We thank Adel Javanmard for a helpful discussion.} $\sqrt{n}\left(\hat{\eta}_{\widehat{S}_{TVA}}-\eta_{\widehat{S}_{TVA}}^{0}\right)$ as

\begin{align*}
\sqrt{n}\left(\hat{\eta}_{\widehat{S}_{TVA}}-\eta_{\widehat{S}_{TVA}}^{0}\right) & ={\bf 1}\left\{ \mathcal{E}\right\} \cdot\sqrt{n}\left(Z_{S_{TVA}}'Z_{S_{TVA}}\right)^{-1}Z_{S_{TVA}}'\epsilon+{\bf 1}\left\{ \mathcal{E}^{c}\right\} \cdot \sqrt{n}\left(Z_{\widehat{S}_{TVA}}'Z_{\widehat{S}_{TVA}}\right)^{-1}Z_{\widehat{S}_{TVA}}'\epsilon\\
		& =\sqrt{n}\left(Z_{S_{TVA}}'Z_{S_{TVA}}\right)^{-1}Z_{S_{TVA}}'\epsilon-{\bf 1}\left\{ \mathcal{E}^{c}\right\} \cdot\sqrt{n}\left(Z_{S_{TVA}}'Z_{S_{TVA}}\right)^{-1}Z_{S_{TVA}}'\epsilon\\
		& +{\bf 1}\left\{ \mathcal{E}^{c}\right\} \cdot \sqrt{n} \left(Z_{\widehat{S}_{TVA}}'Z_{\widehat{S}_{TVA}}\right)^{-1}Z_{\widehat{S}_{TVA}}'\epsilon.
\end{align*}

The proof strategy is to see that central limit theorem ensures the asymptotic normality of the first term in this sum in the usual way, while the remaining terms will asymptotically vanish in probability.

Let us take the first term. First, by Assumption \ref{assu: design}, note that $K^2 \log(K) = O \left(n^{2\gamma} \log(n) \right)= o(n)$, using that $2\gamma < 1 $ and that $\log(n)$ grows more slowly than any polynomial in $n$. $K$ thereby satisfies a critical condition for asymptoptic normality of OLS in a regime where the number of parameters grow (Corollary 2.1 of \cite{he2000parameters}, using also that the score function for linear regression is smooth). Then, using that these are OLS estimates of $\eta_{S_{TVA}}^0$,
	\[
	\sqrt{n}\left(Z_{S_{TVA}}'Z_{S_{TVA}}\right)^{-1}Z_{S_{TVA}}'\epsilon\rightsquigarrow\mathcal{N}\left(0,H^{-1}JH^{-1}\right).
	\]

The second term is ${\bf 1}\left\{ \mathcal{E}^{c}\right\} \cdot \sqrt{n} \left(Z_{S_{TVA}}'Z_{S_{TVA}}\right)^{-1}Z_{S_{TVA}}'\epsilon$. We already showed that $\sqrt{n} \left(Z_{S_{TVA}}'Z_{S_{TVA}}\right)^{-1}Z_{S_{TVA}}'\epsilon$ is asymptotically normal and so $O_p(1)$. Since ${\bf 1}\left\{ \mathcal{E}^{c}\right\}$ is $o_p(1)$, the whole term is $o_p(1)$.

The less trivial term is the third one: ${\bf 1}\left\{ \mathcal{E}^{c}\right\} \cdot\sqrt{n} \left(Z_{\widehat{S}_{TVA}}'Z_{\widehat{S}_{TVA}}\right)^{-1}Z_{\widehat{S}_{TVA}}'\epsilon$.  The point is that $\left(Z_{\widehat{S}_{TVA}}'Z_{\widehat{S}_{TVA}}\right)^{-1}Z_{\widehat{S}_{TVA}}'\epsilon$, which potentially inherits omitted variable bias by including incorrect regressors, is nevertheless uniformly bounded in $n$ and $K$. In detail, first note that $\left\Vert\left(Z_{\widehat{S}_{TVA}}'Z_{\widehat{S}_{TVA}}\right)^{-1}Z_{\widehat{S}_{TVA}}'\epsilon\right\Vert_\infty  \leq \left\Vert Z_{\widehat{S}_{TVA}}'\epsilon\right\Vert_\infty $ because $\left(Z_{\widehat{S}_{TVA}}'Z_{\widehat{S}_{TVA}}\right)^{-1}$ is a positive definite block diagonal matrix with every entry $< 1$. Secondly $\left\Vert Z_{\widehat{S}_{TVA}}'\epsilon \right\Vert_\infty < K n \epsilon$ since $Z_{\widehat{S}_{TVA}}$ is a binary matrix. Since $\epsilon$ is $O_p(1)$, $ \left\Vert \left(Z_{\widehat{S}_{TVA}}'Z_{\widehat{S}_{TVA}}\right)^{-1}Z_{\widehat{S}_{TVA}}'\epsilon \right\Vert_\infty $ is uniformly $O_p(Kn)$ over all misspecifications $\widehat{S}_{TVA}$.

Thus, ${\bf 1}\left\{ \mathcal{E}^{c}\right\} \cdot \sqrt{n} \left(Z_{\widehat{S}_{TVA}}'Z_{\widehat{S}_{TVA}}\right)^{-1}Z_{\widehat{S}_{TVA}}'\epsilon$ is uniformly bounded in probability by $\Pr(\mathcal{E}^{c}) O_p(Kn^{\frac{3}{2}})$. But since $\Pr\left(\mathcal{E}^{c}\right)\leq2K\exp\left(-\frac{n\left(\lambda/K\right)^{2}}{2\sigma^{2}}\right) = 2\exp \left(\log(K) -\frac{n\left(\lambda/K\right)^{2}}{2\sigma^{2}} \right)$, recalling the proof of Proposition \ref{prop: TVA_consistency},

\begin{align*}
O_p(K n^\frac{3}{2})\Pr(\mathcal{E}^{c}) &= O_p \left(\frac{3}{2}\log(n) + 2 \log(K) - \frac{n\left(\lambda/K\right)^{2}}{2\sigma^{2}}  \right) \\
& =O_p \left(\frac{3}{2}\log(n) + 2 \gamma \log(n) - \frac{n^{1-2\gamma}\lambda^2}{2\sigma^{2}}  \right)\\
& = o_p(1).
\end{align*}

And thus ${\bf 1}\left\{ \mathcal{E}^{c}\right\} \cdot \sqrt{n} \left(Z_{\widehat{S}_{TVA}}'Z_{\widehat{S}_{TVA}}\right)^{-1}Z_{\widehat{S}_{TVA}}'\epsilon$ also vanishes in probability. Putting this all together
\begin{align*}
\sqrt{n}\left(\hat{\eta}_{\widehat{S}_{TVA}}-\eta_{\widehat{S}_{TVA}}^{0}\right) &\rightsquigarrow\mathcal{N}\left(0,H^{-1}J H^{-1}\right) + o_p(1) + o_p(1) \\
& \rightsquigarrow\mathcal{N}\left(0,H^{-1}J H^{-1}\right) ,
\end{align*}
which completes the proof.
\end{proof}

\

\begin{proof}[Proof of Lemma \ref{lem: localalter}]
The only non-trivial case is for $\kappa$ and $\kappa'$ that are policy variants, i.e. have the same treatment profile. Here the proof follows the basic intuition from the Hasse diagram. If $\kappa$ and $\kappa'$ were adjacent anywhere in the diagram (say for some treatment combinations $k, k'$ such that $\kappa$ pools $k$ and $\kappa'$ pools $k'$) then $\alpha_{\min\{k, k'\}} = \frac{r_{\kappa \kappa'}}{\sqrt{n}}$, violating Assumption 2. On the other hand, if $\kappa$ and $\kappa'$ are not adjacent anywhere, one can consider the policies "in between" them, i.e. the policies $\zeta_1,...,\zeta_n \in S_{TVA}$ such that $\kappa$ is adjacent to $\zeta_1$, $\kappa'$ is adjacent to $\zeta_n$, $\zeta_{i}$ is adjacent to $\zeta_{i+1}$ for $i \notin \{1,n\}$, and the union of these $\zeta_i$ pools all $z$ such that $\min\{k,k'\} < z < \max\{k,k'\}$ and which are neither pooled by $\kappa$ nor $\kappa'$ . Then $\kappa$ and $\kappa'$ can be made local alternatives with the relevant ``in between" marginals from $\alpha^0$ satisfying Assumption 2, by making the $|\eta_{S_{TVA},\zeta_i}^{0}|$ sufficiently large. Following this construction where applicable (in each stage of the iterative procedure where a new pair is made into local alternatives, to ensure that prior local alternatives remain local alternatives, the absolute magnitude of the policy effects of the prior alternatives may have to be increased), an entire set $\alpha^0$ satisfying Assumption 2 ensues.
\end{proof}

\


\clearpage

\setcounter{table}{0}
\renewcommand{\thetable}{B.\arabic{table}}
\setcounter{figure}{0}
\renewcommand{\thefigure}{B.\arabic{figure}}

\section{Pooling Procedure and Properties}\label{sec:pooling}

Here we show how to construct a set of pooled policies $\hat{S}_{TVA}$ from an estimated support of marginals
 $\hat{S}_{\alpha}$. In case the marginals are correctly estimated $\hat{S}_{\alpha} = S_{\alpha}$, we will show that the implied pooling $\hat{S}_{TVA} = S_{TVA}$ has the properties of $\Lambda$-admissibility and maximality mentioned in Section \ref{sec:Estimation}. 

Given $\hat{S}_\alpha$ let $[\hat{S}_\alpha]$ denote its partition into sets of support vectors with the same treatment profile.  Each $S \in [\hat{S}_\alpha]$ is thus a set of treatment combinations $\{k_1,...k_n\}$. For each $k_i \in \{k_1,...k_n\}$, define the set:

\begin{equation}
	A_{k_i} = \{k \in \mathcal{K} | P(k) = P(k_i) \text{ and } k \geq k_i\}
\end{equation}

In words, $A_{k_i}$ are the policies sharing policy $k_i$'s treatment profile, that weakly dominate $k_i$. 

Now consider a simple operator that selects, for every set $B$, either the set or its complement. Let us write this as $B^a$, where $a \in \{1,c\}$. When $a=1$, the operator selects $B$, while when $a =c$ the operator selects $B^c$. Following this, the pooled policies $\widehat{S}_{TVA}$ is defined as the collection of sets: 

\begin{equation}\label{eq:pooledpolicyatoms}
	\mathcal{A} = A_{k_1}^{a_1} \cap ... \cap A_{k_n}^{a_n}
\end{equation}

which are nonempty, and where furthermore at least one $a_i = 1$. Algorithm \ref{alg:pooling} describes this construction of $\widehat{S}_{TVA}$ procedurally.

We will now lend intuition for this construction, which will clarify its properties. For this discussion, assume the marginal support is correctly estimated, i.e. $\widehat{S}_\alpha = S_{\alpha}$. Considering the relationship between marginal effects $\alpha$ and policy effects $\beta$ given by \eqref{eq:alpha_beta_rel}, $A_k$ is the set of policies whose effects are partly determined or ``influenced'' by $\alpha_{k}$.

It is useful to depict this on the Hasse with a simple example. Take $M=2, R=4$, and the treatment profile where both arms are on. The Hasse diagram of the unique policies within this treatment profile is shown on Figure \ref{fig:Hasse_3}. 

Consider for example the set $A_{[2,1]}$, stemming from $\alpha_{[2,1]}$, depicted as all those policies within the blue contour on this figure. These are all the policies that are influenced by $\alpha_{[2,1]}$. Motivated by this visual depiction of influence, let us say that in general $A_k$ is the ``sphere of influence'' of $\alpha_k$. 

A policy's effect is entirely determined by the various ``spheres'' acting on it. We can restrict attention to those spheres stemming from $\alpha_k \neq 0$ because the other spheres are inactive (spheres of ``non-influence''). Each set $\mathcal{A}$ constructed per \eqref{eq:pooledpolicyatoms}, describes a set of policies subject to a unique set of active spheres. When $a_i = 1$, it means that the sphere $A_{k_i}$ stemming from $\alpha_{k_i}$ is active; when $a_i = c$ the sphere is inactive. Thus in Figure \ref{fig:Hasse_3},  $\mathcal{A}_1 := A_{[1,2]} \cap A_{[2,1]}$ are those policies influenced by both $\alpha_{[1,2]}$ and $\alpha_{[2,1]}$.  $\mathcal{A}_2 := A_{[1,2]} \cap A_{[2,1]}^c$ are those policies influenced by \emph{only} $\alpha_{[1,2]}$ (and not by $\alpha_{[2,1]}$); vice-versa for $\mathcal{A}_3 := A_{[1,2]}^c \cap A_{[2,1]}$.

The first main point is that policies subject to the same active spheres of influence have the same policy effects. For any pooled policy $\mathcal{A} \in S_{TVA}$ constsructed per \eqref{eq:pooledpolicyatoms}, and any policy $k \in \mathcal{A}$, the $\alpha \leftrightarrow \beta$ parameter relationship \eqref{eq:alpha_beta_rel} immediately implies $\beta_{k} = \sum_{i|a_i = 1} \alpha_{k_i}$. In words, the pooled policy effects is the sum of the the marginal effects whose spheres are active for it. Thus $S_{TVA}$ meets requirement (1) of $\Lambda$-admissibility (defined in Assumption \ref{assu: admissible} of Section \ref{sec:Estimation})

Requirement (2) of $\Lambda$-admissability, that only variants are pooled, is met by construction. Requirement (3), that contiguous intensities are pooled, stems from the fact that each sphere of influence $A_k$ is clearly a convex shape in the Hasse, and that the intersection of convex shapes is convex. It is easily seen to be met in the example of Figure \ref{fig:Hasse_3}. 

Finally requirement (4) is easily seen to be met in a quick proof by contradiction. Note that requirement (4), while a technical condition, has an elegant description in the Hasse diagram: for every parallelogram one draws in the Hasse diagram of a treatment profile, if the ``top'' segment is pooled together, so is the ``bottom segment''. Assume not, i.e. the bottom segment is cleaved. There must be some nonzero marginal effect along the bottom segment that was responsible for this cleaving. But the sphere of influence of this marginal cuts through the top segment too, cleaving the top segment into distinct pooled sets, a contradiction. 

\begin{obs} $S_{TVA}$ is $\Lambda$-admissible. 
\end{obs}

The second point is that contiguous policies subject to distinct active spheres of influence have distinct policy effects. While one can construct instances of distinct and adjacent active spheres resulting out to the same policy effects (exactly or statistically), these rather contrived constructing are precluded by Assumption \ref{assu: support} in Section \ref{sec:Estimation} of well separated nonzero marginal effect sizes. More precisely, if $\mathcal{A}_1$ and $\mathcal{A}_2$ are pooled policies adjacent at some $k \in \mathcal{A}_1, k' \in \mathcal{A}_2$), then $\eta_{\mathcal{A}_1}$ and $\eta_{\mathcal{A}_2}$ cannot be exactly or statistically indistinguishable unless 
$\alpha_{kk'}$ is zero or tending to zero, which is precluded by Assumption \ref{assu: support}. Thus, $S_{TVA}$ has no redundancies, in the sense of leaving contiguous intensities with the same policy effects unpooled or ``cleaved'' from one another. It is the coarsest $\Lambda$-admissible partition. 

\begin{obs} $S_{TVA}$ is maximal in $\mathcal{P}_{\Lambda}$.
\end{obs}

	\begin{figure}[h]
		\begin{center}
			\includegraphics[scale=0.5]{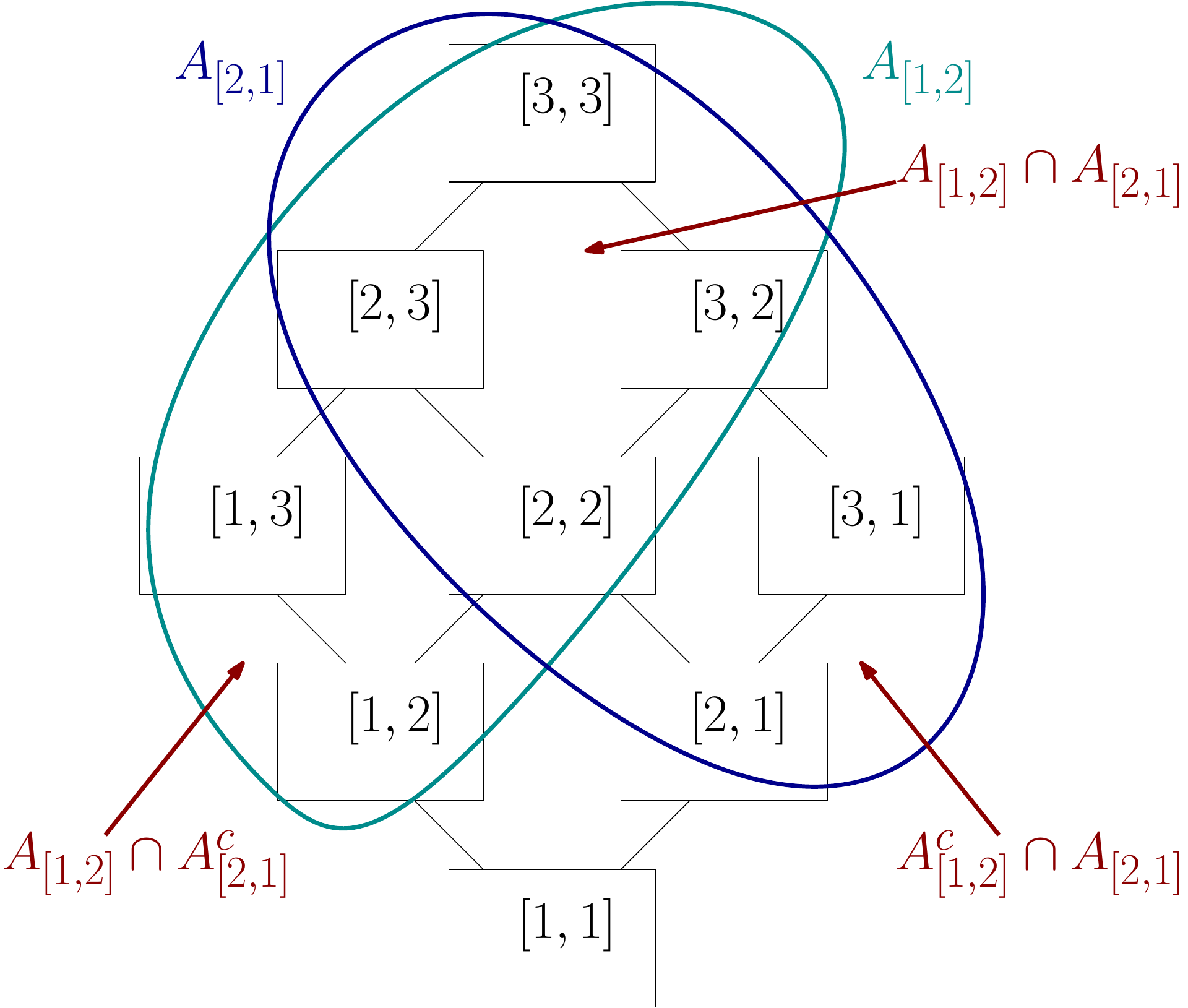}
			\caption{Hasse Diagram with $A_{[2,1]}$ and $A_{[1,2]}$ with complements and intersections. \label{fig:Hasse_3}}
		\end{center}
	\end{figure}

\begin{algorithm}
	\SetAlgoLined
	\SetKwInOut{Input}{input} \SetKwInOut{Output}{output}
	
	\Input{Estimated support $\widehat{S}_\alpha$ from the marginal specification \eqref{eq: TVA} }
	\Output{Estimated pooled policies $\widehat{S}_{TVA}$ for pooled specification \eqref{eq: pooled_pruned}}
	
	Partition $\widehat{S}_\alpha$ into $[\widehat{S}_{\alpha}]$ per the treatment profile mapping $P(.)$ \;
	Initialize $\widehat{S}_{TVA} \longleftarrow \emptyset$ \;
	\For{$S \in [\widehat{S}_{\alpha}]$}{
		discover $S = \{k_1,...,k_n\}$\;
		generate $\{A_{k_1},...,A_{k_n}\}$\;
		\For{each $(a_1,...,a_n) | a_i \in \{1,c\}$} {
			generate $\mathcal{A} = A_{k_1}^{a_1} \cap ... \cap A_{k_n}^{a_n}$ \;
			\If{$\mathcal{A} \neq \emptyset$ and $\mathcal{A} \neq A_{k_1}^{c} \cap ... \cap A_{k_n}^{c}$}{
				$\widehat{S}_{TVA} \longleftarrow \widehat{S}_{TVA}\cup \mathcal{A}$\;
			}
		}
	}
	\Return{$\widehat{S}_{TVA}$}
	\caption{Pooling Procedure}
	\label{alg:pooling}
\end{algorithm}


\setcounter{table}{0}
\renewcommand{\thetable}{C.\arabic{table}}
\setcounter{figure}{0}
\renewcommand{\thefigure}{C.\arabic{figure}}
\clearpage
\section{Simulations} \label{sec:simulations}

\subsection{TVA without Puffering,  \eqref{eq: TVA},  fails irrepresentability} \label{subsec:fail_irrep}

Consider the marginal effect covariate where all $M$ arms are ``on" with highest intensity, i.e., $X_{k^*}$ for $k^* = (R-1,...,R-1)$. We will show that this covariate is ``representable" by the other covariates. Intuitively, this means too much of this covariate is explained by the others. Formally, the $L_1$ norm of the coefficients (excluding intercept) from a regression of this covariate over the others is too great (it exceeds $>1$). That is, if an OLS regression finds:

\[X_{k^*} = \hat{\gamma}_0 + \sum_{k \in K, k \neq {k^*}} \hat{\gamma}_k X_k.\]

$X_{k^*}$ is representable by the other covariates if $\sum_{k \in K, k \neq {k^*}} |\hat{\gamma}_k|
>1$.

We demonstrate this through a proof by example. A simulation establishes that $X_{k^*}$ is representable (and therefore that the specification fails irrepresentability) for a computationally reasonable range of $R$ and $M$. We see that the patterns imply that irrepresentability fails even more dramatically for larger $R$ and $M$. 

In this simulation, we choose large $n = 10,000$ so that the propensities of ``1" within each covariate have stabilized. We consider two kinds of regressions: an ``unstandardized" regression where the raw marginal effects covariates are regressed, and a ``standardized" regression where the marginal covariates are first standardized by the $L_2$ norm. The latter corresponds to a preprocessing step that LASSO packages typically apply before LASSO selecting; we would like to know if irrepresentability fails even in this case. Indeed, we see the $L_1$  norms are greater than 1 in both cases, and irrepresentability fails.

\begin{table}[h]
\centering
\begin{tabular}{rrrr}
  \hline
$R$ & $M$ & $L_1$ norm (standardized vars) & $L_1$ norm (unstandardized vars) \\ 
  \hline
3 & 2 & 1.73 & 1.26 \\ 
  3 & 3 & 3.67 & 2.32 \\ 
  3 & 4 & 7.66 & 4.2 \\ 
  4 & 2 & 1.77 & 1.24 \\ 
  4 & 3 & 3.98 & 2.43 \\ 
  4 & 4 & 8.27 & 4.14 \\ 
  5 & 2 & 1.87 & 1.28 \\ 
  5 & 3 & 3.78 & 2.29 \\ 
  5 & 4 & 7.90 & 4.70 \\ 
   \hline
\end{tabular}
\end{table}

\setcounter{table}{0}
\renewcommand{\thetable}{D.\arabic{table}}
\setcounter{figure}{0}
\renewcommand{\thefigure}{D.\arabic{figure}}
\clearpage
\section{Robustness} \label{sec: robustness}

 \setcounter{table}{0}
\renewcommand{\thetable}{C.\arabic{table}}

In this section we consider robustness of TVA conclusions on our data. As motivated in the body, we consider the plausible environment of Section \ref{sec:Estimation} so that a data driven procedure can reveal a set of relevant pooled policies and estimate it together with an estimate of the single best policy.\footnote{We emphasize again that exact sparsity of Assumption \ref{assu: design} can be practically relaxed to some of the approximate sparsity regimes explored in Section  \ref{sec: simulations} subsection \ref{subsec: sparisty robustness}} The main issues are then the sensitivity of TVA results to (1) level of sparsity imposed and (2) the particular draw of the data. We discuss these below, and in doing so also elaborate on the ``admissible" LASSO penalties $\lambda$ for evaluating robustness as well as the choice of $\lambda$ emphasized in the body. We intend for this to be helfpul as a user's guide for future practitioners. 

\underline{Sparsity level}

The LASSO penalty level $\lambda$, which determines the level of sparsity imposed by TVA, trades off Type I and Type II error. A higher $\lambda$ implies less of the first at the price of the second. Our main recommendation is to err on the side of lower Type I error. For one, inclusion in $\widehat{S}_{TVA}$ of false positives might (misleadingly) over-attenuate the best policy for its winner's curse if one of them emerges as the ``second best". Secondly, and more seriously, one of these false positives might itself be selected as the best policy. From a policy standpoint this is a particularly dangerous error in the context of government advice. So, adopting a conservative stance, the admissible $\lambda$ will be higher than lower. 

We determine the sufficiently high $\lambda$ through the following exercise, which can be applied generally. Namely, we first consider raw sensitivity of best policy selection and winner's curse adjusted estimates to a range of $\lambda$, as in Figure \ref{fig:penalty_sensitivity} panel A.\footnote{Note that for immunizations, the support coverage (of the 75 policies) ranges from 3\% (right side of the diagram at the highest displayed $\lambda$) to 52\% (left side of the diagram at the lowest displayed $\lambda$). For immunizations/\$ the support coverage ranges from 4\% to 39\%.} Within this range, we see that for immunizations/\$, the policy (Info Hubs (All), No Incentives, SMS (All)) is robustly selected across the range. On the other hand, for immunizations the policy (Info Hubs, Slopes (All), SMS (All)) is selected for $\lambda \geq 0.47$ while (No Seeds, High Slopes, Low SMS) is selected for $\lambda < 0.47$. Furthermore, for immunizations/\$ the winners curse adjusted estimates of (Info Hubs, Slopes (All), SMS (All)) reject 0 except at $\lambda \leq 0.00045$. On the other hand, for the outcome of immunizations the winners curse adjusted estimate of (Info Hubs, Slopes (All), SMS (All))  rejects 0 for the range $\lambda \geq 0.47$. For $\lambda < 0.47 $ the other selected policy (No Seeds, High Slopes, Low SMS) always fails to reject 0.

Next, we interpret these preliminary findings in light on the aforementioned false positive risk. To do this, we plot the unadjusted post-LASSO estimates of the best policy together with the second best estimates, as in panel B of Figure \ref{fig:penalty_sensitivity}. This gives a sense of the the nominal estimates/confidence intervals and the `perpetrators' of winners curse attenuations. We evaluate these nominal values per the observation by  \cite{taylor2015statistical} that a LASSO based procedure greedily favors false positives over true negatives, and so actual type I error rates and $p$-values in post-selection inference are larger than nominal ones. Thus nominally insignificant policies are even more likely to be true negatives, and nominally borderline significant policies are likely to be insignificant as well. Our conservative principle thus tells us dial up the $\lambda$ to remove considerations of these. We demonstrate below. 

Consider first the immunizations/\$ outcome. For $\lambda < 0.0008$ these second bests (policies in green and pink) have nominal OLS confidence intervals that fail to reject 0. Thus $\lambda$ should be increased to the range $\lambda \geq 0.0008$. Since after $\lambda =0.00160$ both this policy and the policy support is entirely deselected, the admissible $\lambda$ is the interval $[0.0008,0.00160)$. Per Figure \ref{fig:penalty_sensitivity},  panel A1, the winner's curse estimates are robust in this admissible interval.

Proceeding to the immunizations outcome, for $\lambda \leq 0.19$ and $\lambda \in (0.415,0.452]$, even nominal post-LASSO estimates of (No Seeds, High Slopes, Low SMS) fail to reject 0. For the range $0.19 < \lambda < 0.358$, the noisy and nominally insignificant policy (Info Hubs, High Slope, No Reminder) is responsible for the winner's curse attenuations. We therefore dial up $\lambda$ further. Doing so gives an extremely narrow range ($\lambda \in [0.358,0.415]$) where (No Seeds, High Slopes, Low SMS) is significant, and that too barely ($p=0.048$ at the displayed $\lambda = 0.039$). Keeping in mind the aforementioned observation by \cite{taylor2015statistical}, this nominal significance is particularly unreliable at the threshold. Moreover, it is in stark contrast to the robust stability of the post-LASSO estimates of (Info Hubs, Slopes (All), SMS (All)) at $\lambda \geq 0.39$. Altogether this is strongly suggestive that (No Seeds, High Slopes, Low SMS) is a policy that we should disinclude as a false positive. This is achieved at $\lambda > 0.452$. Since after $\lambda \geq 0.53$ this and all policies are dropped, the admissible $\lambda$ for this interval is $(0.452,0.53)$. Per Figure  \ref{fig:penalty_sensitivity},  panel A2, the winner's curse estimates are robust in this admissible interval. 

The two sets of admissible $\lambda$ are for two entirely different scale of outcomes and cannot usually be directly compared. However, since our specific implementation of LASSO on $\text{Puffer}_N$ is bijective with a $p$-value threshold in a backwards elimination (following \cite{rohe2014note}), following this bijection the set of admisslbe $\lambda$ for immunization/\$ is wider. So there is some sense in which TVA estimator is more fragile. Another approach that evaluates this fragility is a boostrapping analysis, which is explained shortly. Also since we use this implementation with $p$ value thresholds, for the result we highlight in the paper we use ``bottleneck" $\lambda = 0.48$ for immunizations, namely $p = 5 \times 10^{-13}$. This maps to $\lambda = 0.0014$ for immunizations/\$. The corresponding Type I error level is thus constant.

Finally, we emphasize that checking a $\lambda$ against a saturated regression of 75 raw coefficients for finely differentiated policies -- a kind of ``heatmap" test -- is generally \emph{not} a reliable robustness check. In Figures \ref{fig:OLS_immunization} (number of shots) and \ref{fig:OLS_costs} (shots/\$) of our Online Appendix \ref{sec: extended_robustness} we show the fully unrolled saturated regression together with the pooling. But it is not clear how to interpret this with respect to pooling. The eye cannot sanity check the combinatorially large number of joint hypothesis tests for pooling that become relevant when more than one arm is on. Figure \ref{fig:power_comparison}, panel B from Section \ref{sec: simulations} considers a typical simulation where eye-balling fails to sanity check the right pooling choices.\\

\underline{Bootstrapping analysis}

It is natural to ask about the fragility of TVA to the particular draw of the data; precisely this concern motivates, for example, our implementation of winner's curse adjustments by \cite{andrews2019inference}, as well as simulations in Section \ref{sec: simulations} that directly speak to the variance of TVA. However, one might further wonder about just the observations in our dataset, with a concern akin to one about leverage of observations. An intuitive approach to address this is a bootstrapping analysis, where TVA is run on multiple bootstrapped samples. We can then speak to variation in both the set of supports selected as well the estimates of the pooled policies. Because this is a more exploratory analysis, its principal value lies in speaking to \emph{relative} stability of conclusions between the two policies for the two outcomes. \footnote{Note that we have slightly different goals here from the issue of bootstrapped standard errors.} Results are presented in Figures \ref{fig:boot_immunization} (immunizations) and \ref{fig:boot_costs} (immunizations /\$) of our Online Appendix and demonstrate a robust stability for the cost-effective outcome, where the original support is selected in 96\% of bootstrapped sample, against 77\% for the immunizations outcome.

\newgeometry{left=1cm, bottom = 1cm}
\begin{landscape}

\begin{figure}[!h]
	\begin{subfigure}[Panel A]{\linewidth}
		\includegraphics[scale=0.38]{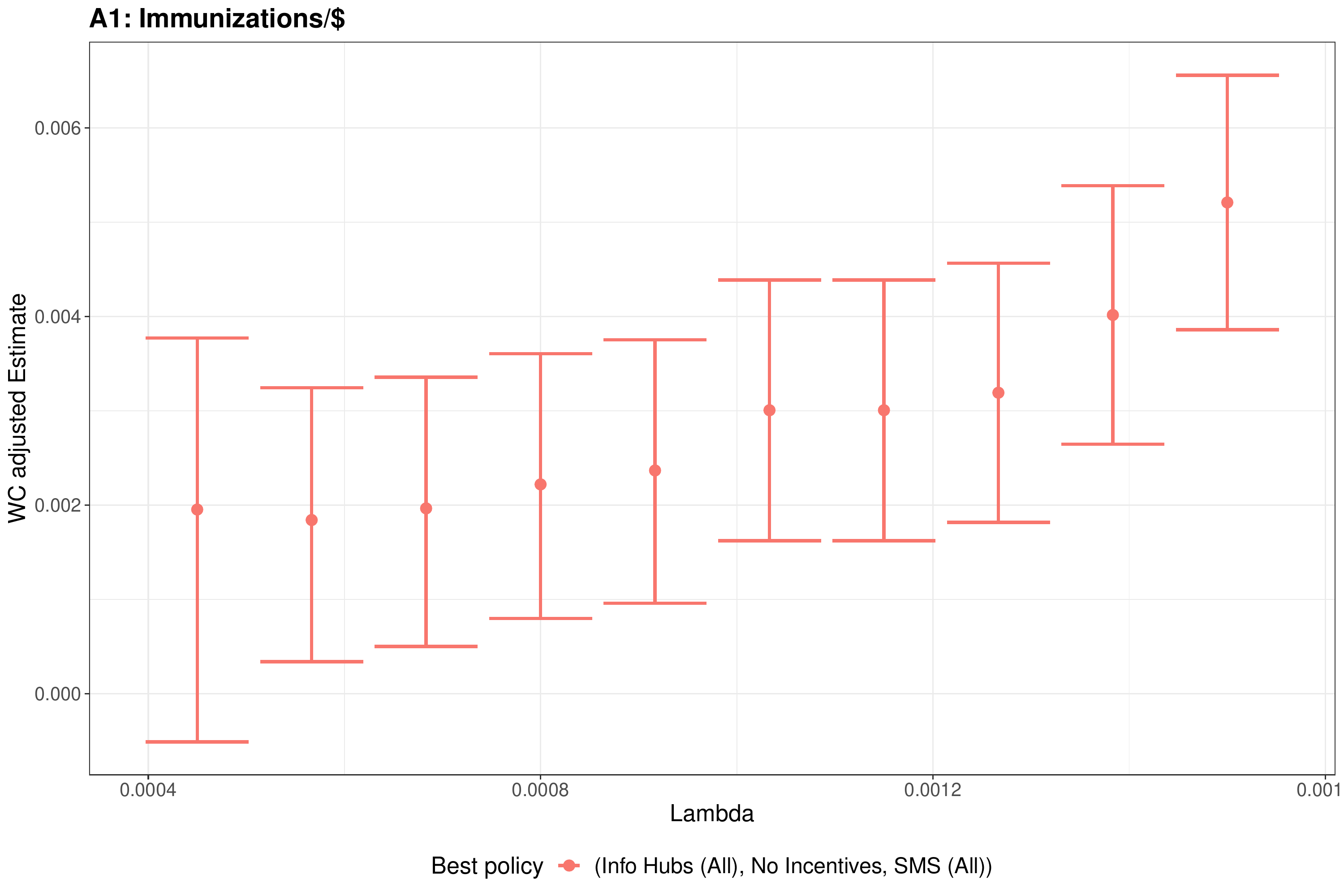}
		\qquad
		\includegraphics[scale=0.38]{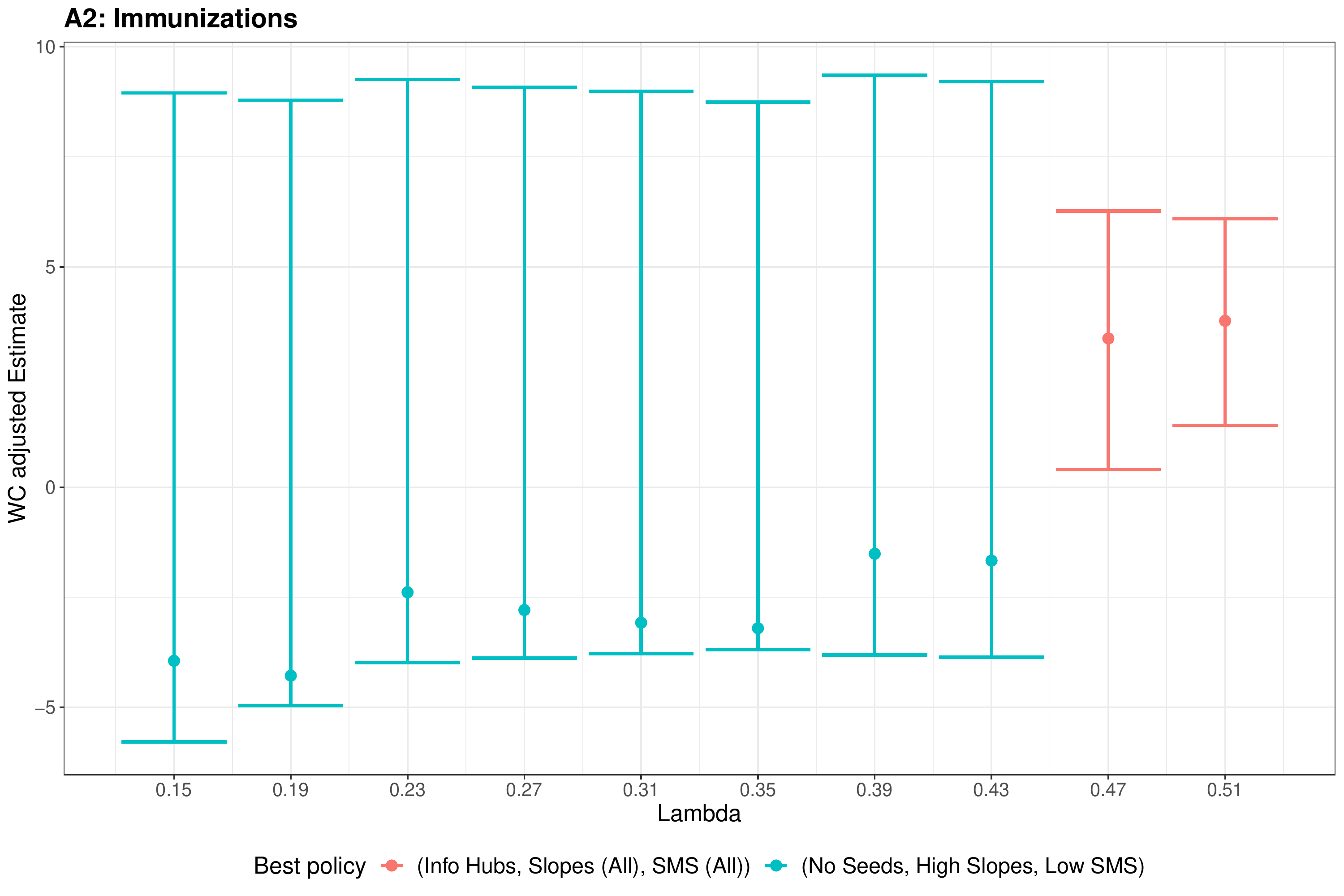}
		\caption{Best policy per LASSO penalty (after WC adjustment)}
	\end{subfigure}

	\begin{subfigure}[Panel B]{\linewidth}
		\includegraphics[scale=0.38]{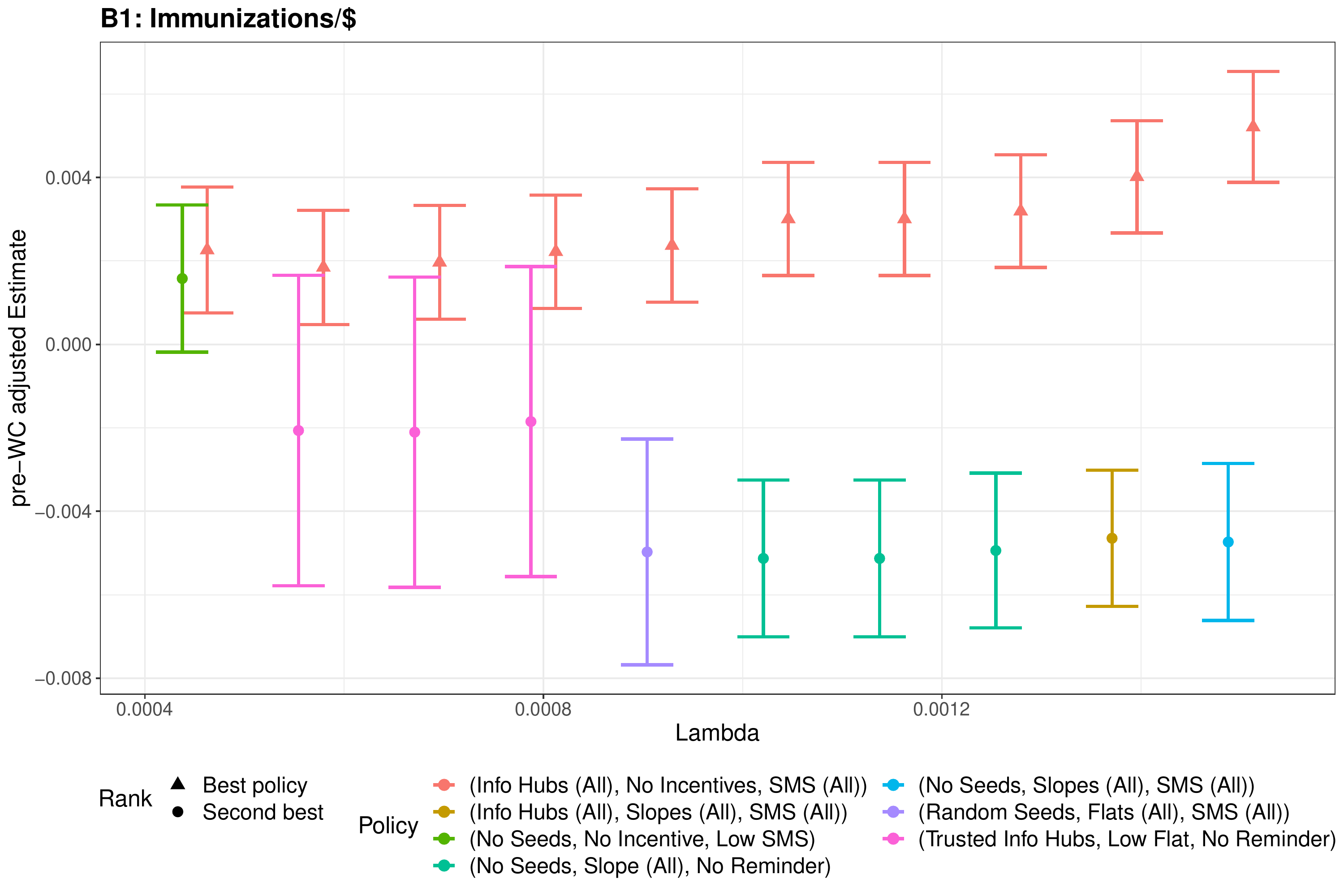}
		\qquad
		\includegraphics[scale=0.38]{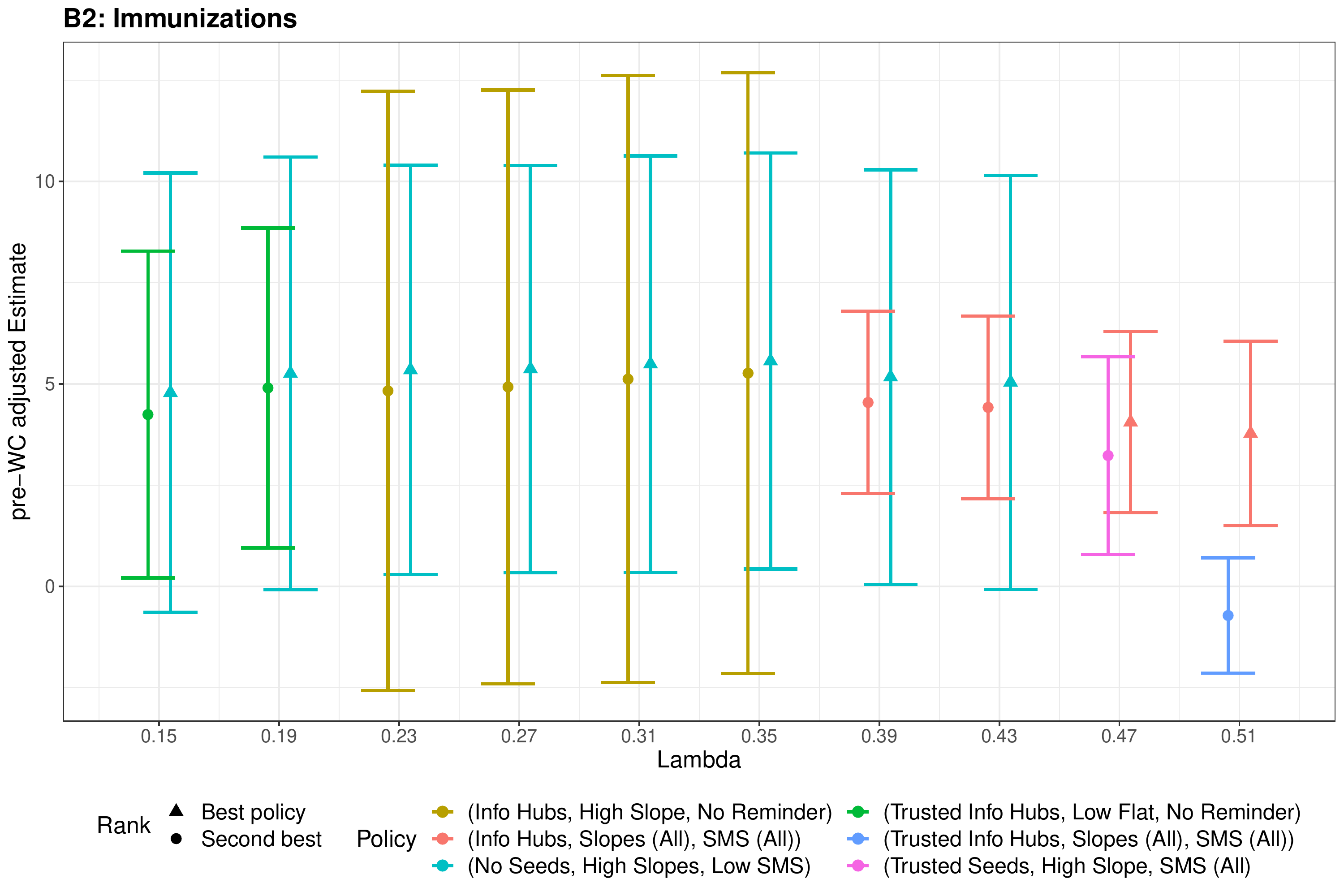}
		\caption{Best and Second Best policies per LASSO penalty (before WC adjustment)}
	\end{subfigure}
\caption{Panel A shows the sensitivity of best-policy estimation per LASSO penalty $\lambda$ for a sequential elimination version of LASSO on $\text{Puffer}_N (X,Y)$. Estimates are plotted after adjusting for the Winner's Curse. Panel B shows the sensitivity of both the best and second-best policy estimates. Here estimates are shown prior to winner’s curse adjustment i.e. when best policy selection happens. \label{fig:penalty_sensitivity}}
\end{figure}

\end{landscape}
\restoregeometry

\clearpage


\begin{center}
	\large{ONLINE APPENDIX : NOT FOR PUBLICATION}
\end{center}
\vspace{0.8cm}
\begin{center}
	\large{Selecting the Most Effective Nudge: Evidence from a Large-Scale Experiment on Immunization}
\end{center}
\vspace{0.8cm}
\begin{center}
Abhijit Banerjee,  Arun G. Chandrasekhar, Suresh Dalpath, Esther Duflo, John Floretta, Matthew O. Jackson, Harini Kannan, Francine Loza, Anirudh Sankar, Anna Schrimpf, and Maheshwor Shrestha
\end{center}

\vspace{1cm}

\setcounter{table}{0}
\renewcommand{\thetable}{E.\arabic{table}}
\setcounter{figure}{0}
\renewcommand{\thefigure}{E.\arabic{figure}}
\section{Extended Simulations}\label{sec:simulations_appendix}
This extended simulations section serves the purpose of providing complementary evidence and extending the discussion on TVA performance (A) relative to alternative estimation strategies and (B) under relaxation of sparsity and lower-bound assumptions.

\subsection{Performance relative to alternatives}\label{subsec: alternatives_appendix}
In this section we present a detailed comparison of TVA performance relative to the LASSO-based alternatives mentioned in Section \ref{subsec: alternatives}.

\subsubsection{Naive LASSO}

As outlined in section \ref{subsubsec: alternatives}, there are two ways one could ``naively" apply LASSO. The first is to disregard pooling, and, because sparse dosages might also mean a sparse set of policies, apply LASSO on the unique policy specification (\ref{eq:unique_policy}). While there is no theoretical issue with this procedure in terms of model consistency, using this for policy estimation leads to much the same performance limitations as direct OLS with regards to best policy estimation, namely a persistently high best policy MSE stemming from overly severe correction from \cite{andrews2019inference}'s winner's curse adjustment. 

Figure \ref{fig:puffer_vs_lasso1} which contrasts this ``No pooling, only pruning" version of LASSO to TVA on best policy estimation, documents this. Panel A shows that while the MSE for our TVA estimator quickly converges to 0, the one for the LASSO estimator persistently lies above 0.1 regardless of $n$. Panels B and C verify that this is driven by winner's curse attenuations, and not model selection issues. Panel B tracks the MSE conditional on both procedures selecting the right model of their respective specifications as the oracle; it is the same pattern as in Panel A. Panel C explicitly shows the shrinkage imposed by the best policy estimator; it is much higher when the model selection doesn't pool.

\begin{figure}[!ht]
		\begin{center}
			\includegraphics[scale=0.55]{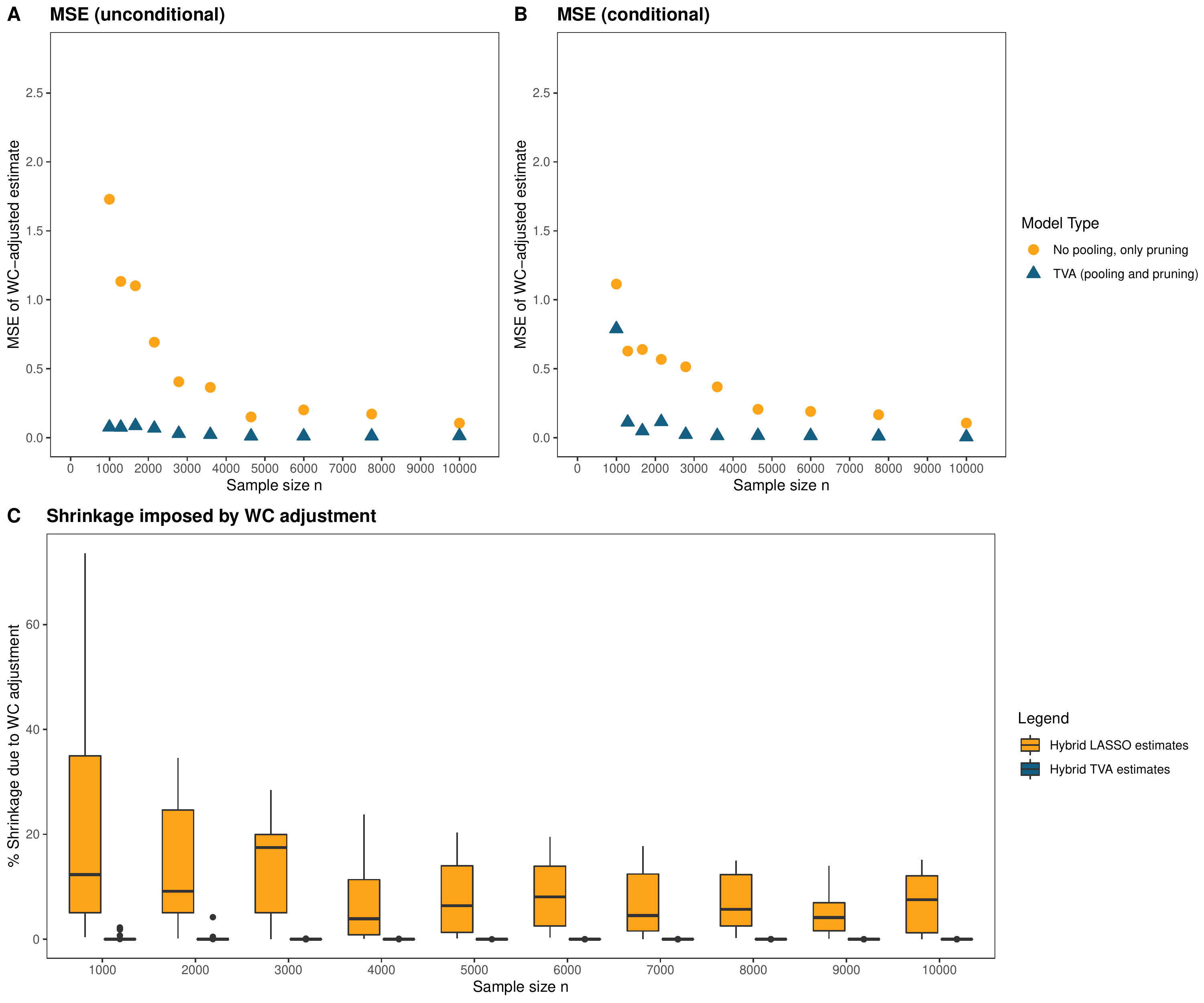}
			\caption{A plot comparing the performance of the TVA estimator to applying LASSO on the unique policy specification (\ref{eq:unique_policy}). Panel A compares the MSE of the hybrid estimators of \cite{andrews2019inference} for winner's curse-adjusted best policy estimation for both methods. Panel B is exactly the same but conditional on selecting the true support in (\ref{eq: TVA}) and (\ref{eq:unique_policy}). Panel C compares the amount of shrinkage imposed by the winner's curse adjustment for both methods, in percentage of the initial coefficient. In all simulations, there are 20 simulations per $n$\label{fig:puffer_vs_lasso1}}
		\end{center}
\end{figure}

The second way to ``naively" apply LASSO is to consider both pooling and pruning as important, but adopt a sign inconsistent model selection procedure by applying LASSO directly on \ref{eq: TVA} without a Puffer transformation. Figure \ref{fig:puffer_vs_lasso2} establishes the contrast with TVA. As expected, simulations attest to inconsistent support selection (Panel A). On the question of best policy estimation, it is more subtle. This naive LASSO does manage to identify at least some of the best policy, furthermore the actual MSE of the best policy is comparable to TVA (Panel C). However, a key deficiency from a policymaking perspective is that it fails to select the minimum dosage best policy with substantial probability relative to TVA (Panel D).\footnote{The reason why the naive LASSO will select some best policy stems from the fact that though it is not sign consistent, it is $l_2$ consistent, which means that it will select a strict superset of the correct support \cite{meinshausen2009lasso}. Basically, this will prune too little and pool too finely. In particular, as seen in our simulations, it will often cleave the best policy. The minimum dosage best policy is then at substantial risk of not being in the pool selected as the best in the data}.

\begin{figure}[!ht]
		\begin{center}
			\includegraphics[scale=0.6]{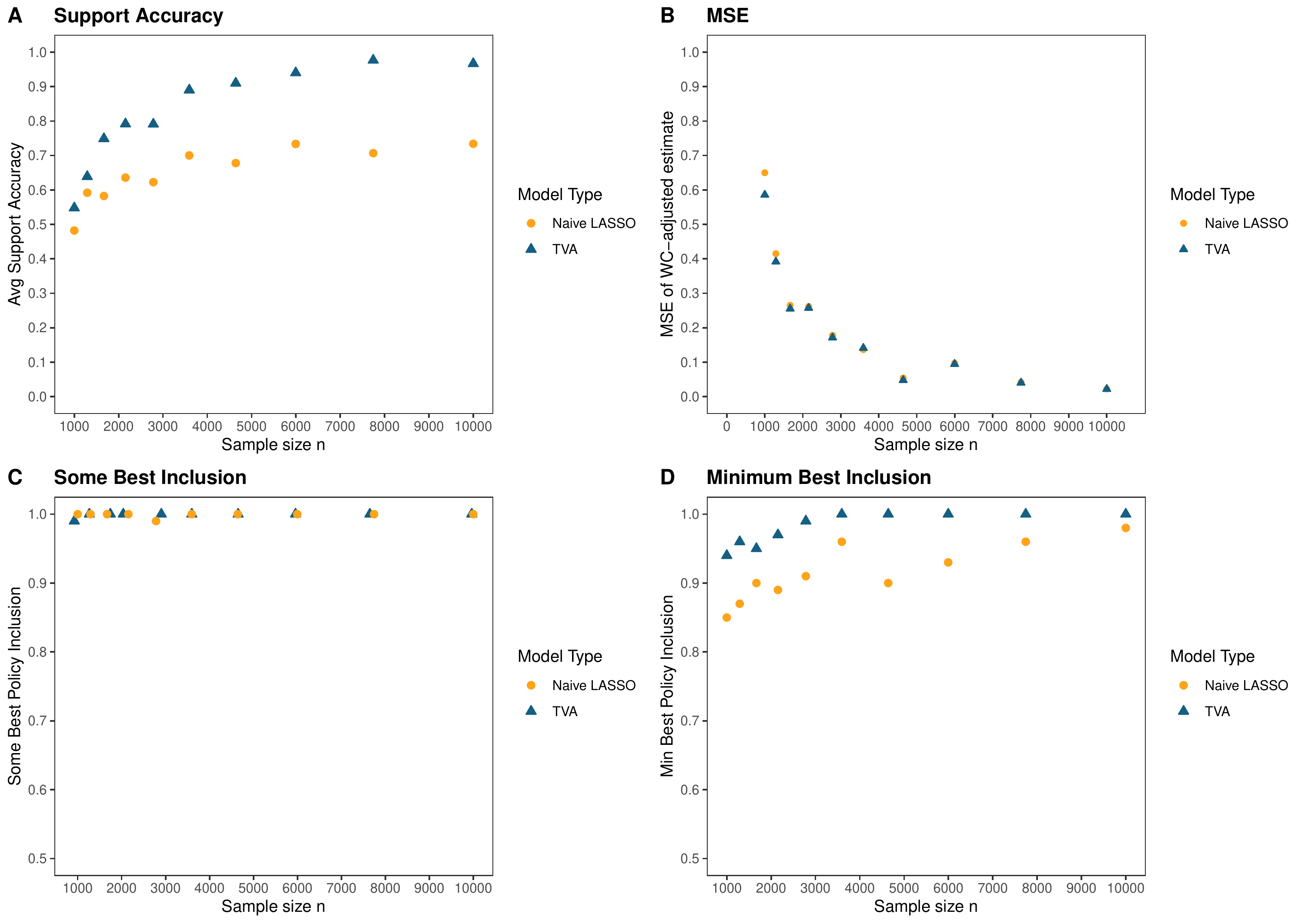}
			\caption{A plot comparing the performance of the TVA estimator to applying LASSO on (\ref{eq: TVA}) using \cite{chernozhukov2015post}. Panel A compares average support accuracy and Panel B compares MSE of the best policy treatment effect as a function of sample size $n$.  Panels C and D look at best policy selection accuracies. On panel C points are slightly jittered for better readability. There are 20 simulations per support configuration per $n$, for five support configurations. \label{fig:puffer_vs_lasso2}}
		\end{center}
\end{figure}

\subsubsection{Debiased LASSO}
Because at a higher level we are interested in high dimensional inference \footnote{albeit still in a $K< n$ regime} one alternative to a two step process of model selection and inference is the so-called debiased LASSO \cite{javanmard2014confidence}, \cite{javanmard2018debiasing}, \cite{van2019asymptotic}. The basic idea is that since LASSO permits high dimensional estimation but at the price of downwardly biased coefficients, but also since this bias is estimable, we can reverse the bias. In particular, to the LASSO coefficients we can add a debiasing term proportional to the subgradient at the $\ell_1$ norm of the LASSO solution $\widehat \theta^n$
\[\widehat \theta^u = \widehat \theta^n + (1/n)M \textbf{X}^T(Y- \textbf{X} \widehat \theta^n)\]
It is also possible to supply standard errors for these debiased coefficients. Note, however, that these debiased coefficients are almost surely never exactly zero, so that there is no question of sparsity. We thus only need to consider applying debiased LASSO to (\ref{eq:unique_policy}). 


\begin{figure}[!ht]

		\begin{center}
			\includegraphics[scale=0.48]{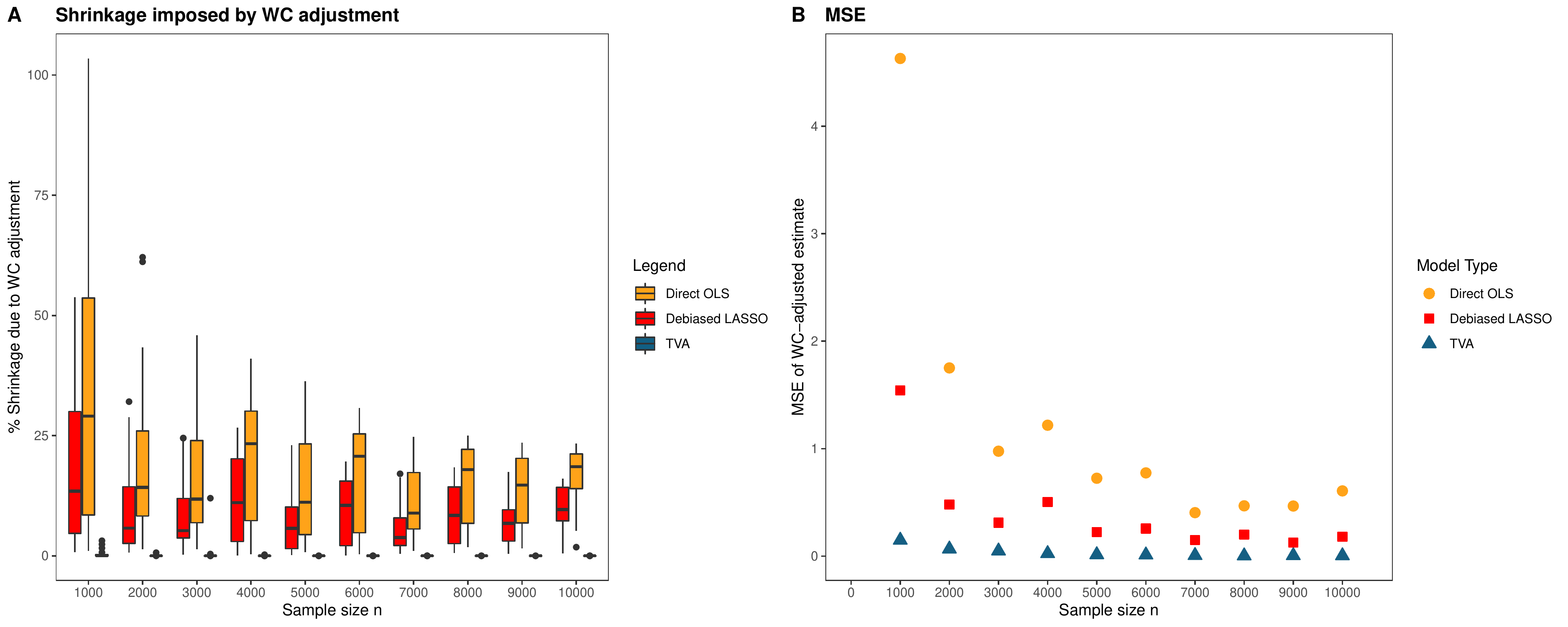}
			\caption{A plot comparing the performance of the TVA estimator to applying OLS on (\ref{eq:unique_policy}) and a debiased LASSO on (\ref{eq:unique_policy}) following \cite{javanmard2014confidence}. Panel A compares the amount of shrinkage imposed by the winner's curse adjustment as percentage of the initial coefficient. Panel B compares the mean squared error of the best policy estimate as a function of sample size $n$. There are 20 simulations per $n$. \label{fig:puffer_vs_dlasso}}
		\end{center}
\end{figure}

In Figure \ref{fig:puffer_vs_dlasso} we show that the debiased LASSO procedure suffers from the similar limitations as direct OLS estimation, especially with regards to best policy estimation. That it does better with winner's curse attenuations relative to the same adjustments on direct OLS possibly indicates that it might interact better with those adjustments despite the fact that Gauss-Markov theorem guarantees that the unadjusted direct OLS must dominate the unadjusted debiased LASSO.  Nevertheless, TVA sharply dominates both alternatives because it can pool.

\subsubsection{``Off the Shelf" Bayesian approaches: Spike and Slab LASSO}\label{subsubsec:bbssl}
Because we envision a world in which there are many potential treatment variant aggregations to be done through pooling and pruning, this attests to our prior about the environment. Indeed, LASSO estimates have a Bayesian interpretation in terms of Laplace priors. One can ask whether a more sophisticated, ``explicitly" Bayesian approach can address our final objectives. This paradigmatically different route is the topic of future work. Below, we just show that``off the shelf" Bayesian approaches are unlikely to help. In particular, we show that a direct application of spike and slab formulations -- the most intuitively relevant method -- underperforms relative to our TVA procedure. 
The Spike and Slab LASSO uses a prior of the form
$$
	\pi(\beta | \gamma) = \prod_{i=1}^p[\gamma_i \underbrace{\psi_1(\beta_i)}_{Slab} + (1-\gamma_i)\underbrace{\psi_0(\beta_i)}_{Spike}], \gamma \sim  \pi(\gamma)
$$
with $\gamma_0$ and $\gamma_1$, two Laplace distributions with very high ($\lambda_0$) and a very low ($\lambda_1$) scale parameters respectively (i.e. $\lambda_0 \gg \lambda_1$). This allows solving for the posterior of both the model parameters as well as the model itself (\cite{rovckova2018spike}).


\begin{figure}[!ht]
		\begin{center}
			\includegraphics[scale=0.5]{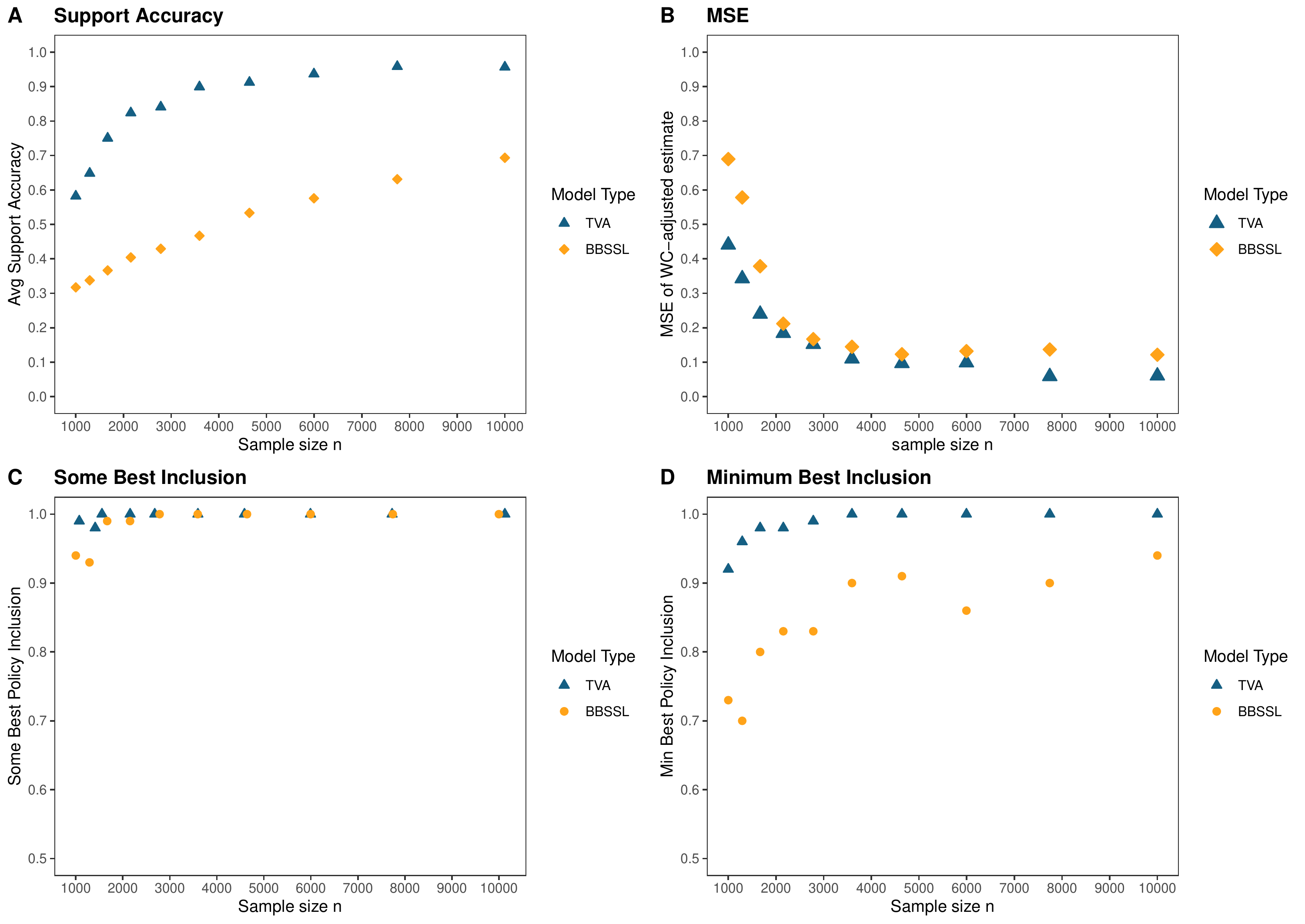}
			\caption{A plot comparing the performance of the TVA estimator to applying Bayesian Bootstrap Spike and Slab LASSO (\cite{nie2022bayesian}) on (\ref{eq: TVA}). Panel A compares support accuracy and panel B compares MSE on the final WC adjusted estimate. Panels C and D compare best policy inclusion measures as a function of $n$. For BBSSL the Laplace parameters are chosen as $\lambda_1 = 10^{-n/100}$ and $\lambda_0$ ranging from $1$ to $10^5$ (100 steps). This yields a solution path as a function of $\lambda_0$ and policies selected at least $95\%$ of times along the solution path compose the final support. There are 20 simulations per support configuration per $n$ for five support configurations. \label{fig:puffer_vs_sslasso1}}
		\end{center}
\end{figure}

In Figure \ref{fig:puffer_vs_sslasso1} we contrast a Puffer-transformed Bayesian Bootstrap Spike and Slab LASSO (BBSSL - \cite{nie2022bayesian}) to TVA. The performance of BBSSL is very similar to that when applying Naive LASSO to the marginal specification \eqref{eq: TVA}. BBSSL is support inconsistent (Panel A) and is clearly outperformed on minimum dosage best policy selection (Panel D).  It does however identify at least one best policy most of the time (Panel C) and has a best policy MSE close to that of TVA (Panel B).

\textcolor{red}{\subsection{Performance under Sparsity Relaxations}\label{subsec: sparsity_appendix}}
In this section we repeat the exercise of relaxing the sparsity requirements from section \ref{subsec: sparisty robustness}, to offer a comparison between the performance of TVA and the next best alternative of applying LASSO to the marginal specification \eqref{eq: TVA}. 

Below we reproduce the summary of the five regimes of sparsity and effect size relaxations we consider for these simulations: 


\begin{enumerate}
\item \textbf{Regime 1}: $M$ constant marginal effect sizes in $[1,5]$ \& $M$ rapidly diminishing remaining marginal effect sizes in ($[1,5]/n$)
\item \textbf{Regime 2}: $M$ constant marginals in $[1,5]$ \& $M$ moderately diminishing remaining marginals ($[1,5]/\sqrt{n}$) 
\item \textbf{Regime 3}: $M$ large constant marginals in $[5,10]$, $M$ medium marginals in $[1,2]$ \& $M$ rapidly diminishing remaining marginals in ($[1,5]/n$)
\item \textbf{Regime 4}: $M$ decreasing marginals in $([1,5]/n^{0.2})$ (and zero marginals everywhere else)
\item \textbf{Regime 5}: $M$ decreasing marginals in $[1,5]/n^{0.2}$ \&  $M$ moderately diminishing remaining marginals  in ($[1,5]/\sqrt{n}$)
\end{enumerate}

The main finding is that TVA strongly outperforms naive LASSO across regimes for support accuracy and minimum dosage best policy selection. Performance on selecting some best policy in $\widehat{S_\alpha}$ and final MSE of the best policy is very similar for both estimators, as it was in the exact sparse case, and also similar across regimes.  

TVA's performance on support accuracy is manifestly affected by sparsity violations: convergence is slower in regime 2 (Figure \ref{fig:regime2}) than regime 1 (Figure \ref{fig:regime1}) and for regimes 4-5 (Figures \ref{fig:regime4}, \ref{fig:regime5}) TVA is not even support consistent anymore. Notable, however, is the fact that naive LASSO is repeatedly underperforming TVA across all regimes and sometimes significantly so (e.g. Regime 3, Figure \ref{fig:regime3}).  This suggests that naive LASSO is an unhelpful way to proceed particularly if one were interested in the post-LASSO set in addition to the best policy.

Both methods then, are equally strong at identifying at least one best policy and with respect to the final MSE of the best policy with performance mostly ranging between 80\%-100\% for the former and never higher than 0.75 (R1, R3) for the latter. However only TVA helps with the joint best policy objectives of low MSE while reliably choosing the minimum dosage best policy. This is compelling in regimes 2 and 3 where LASSO's minimum best selection rate hovers around 60-70\%  against TVA reaching 90\% accuracy for a moderate sample size ($n = 3000$) already.

\begin{figure}[h]
		\begin{center}
			\includegraphics[scale=0.5]{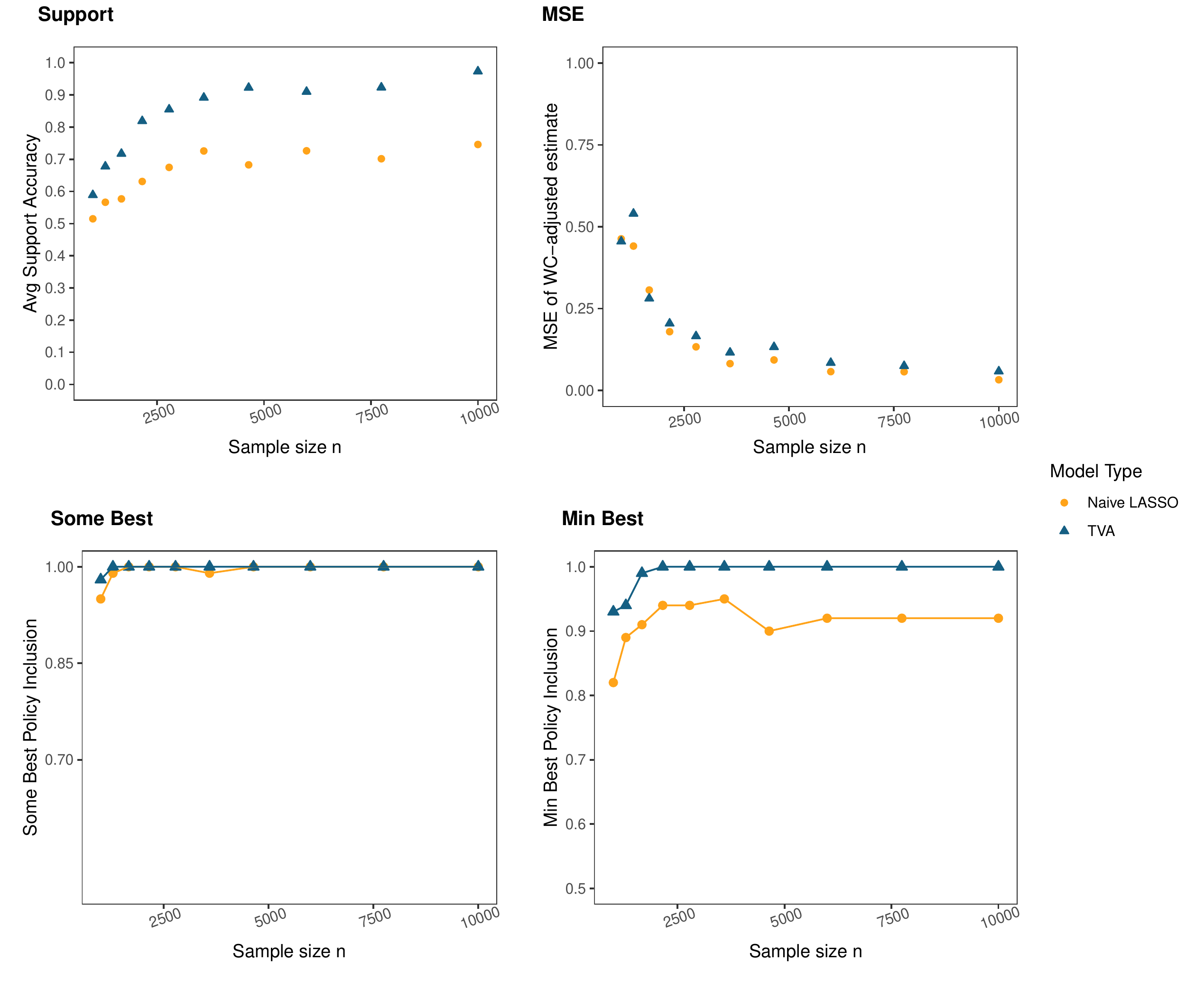}
			\caption{A comparison between the TVA estimator and a direct implementation of LASSO on (\ref{eq: TVA}) for support accuracy, MSE and best policy inclusion measures, under \textbf{regime 1}. There are $20$ simulations per support configuration per $n$, for five support configurations. \label{fig:regime1}}
		\end{center}
\end{figure}

\clearpage

\begin{figure}[h]
		\begin{center}
			\includegraphics[scale=0.5]{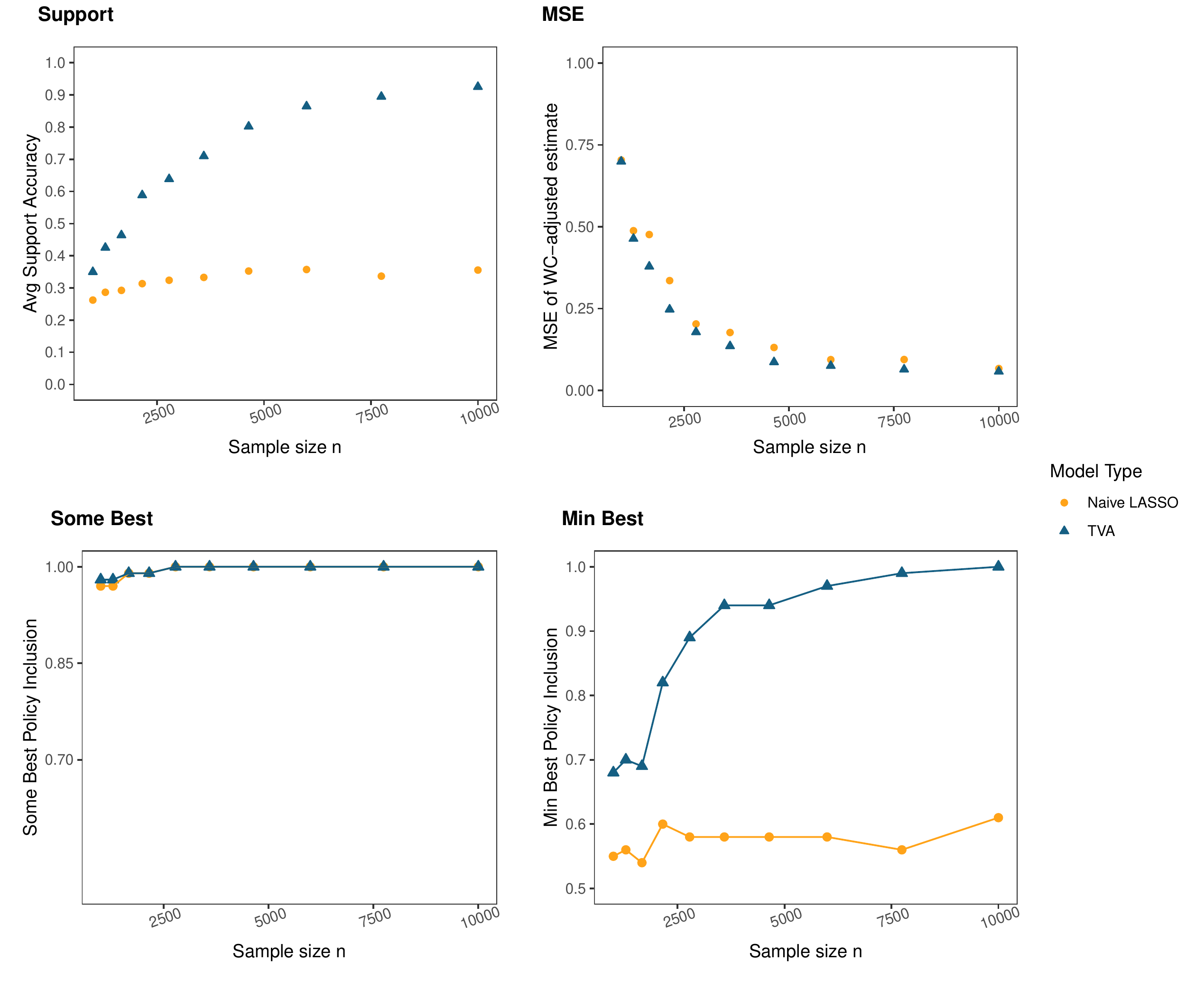}
			\caption{A comparison between the TVA estimator and a direct implementation of LASSO on (\ref{eq: TVA})  for support accuracy, MSE and best policy inclusion measures, under \textbf{regime 2}. There are $20$ simulations per support configuration per $n$, for five support configurations. \label{fig:regime2}}
		\end{center}
\end{figure}

\clearpage

\begin{figure}[h]
		\begin{center}
			\includegraphics[scale=0.5]{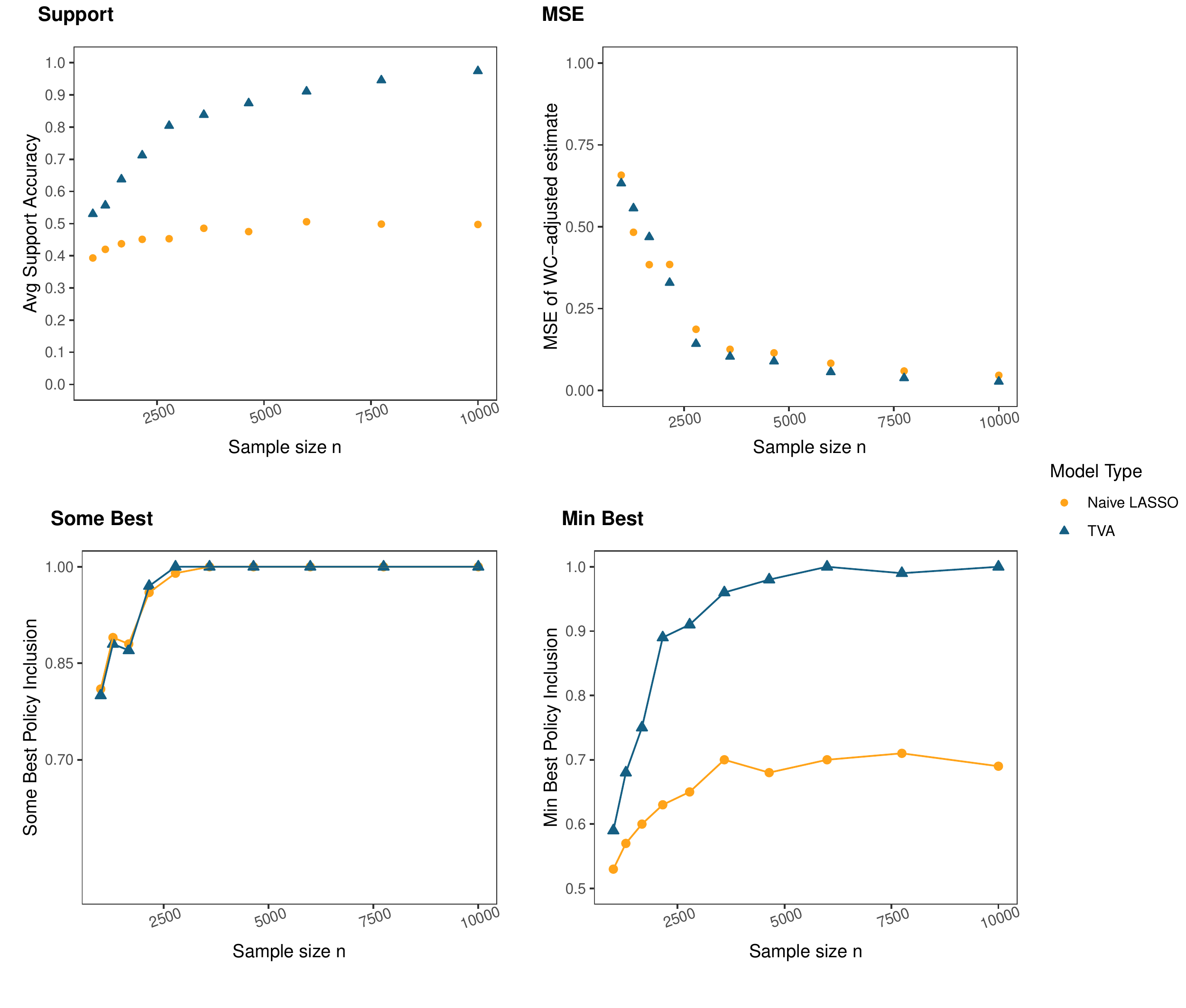}
			\caption{A comparison between the TVA estimator and a direct implementation of LASSO on (\ref{eq: TVA})  for support accuracy, MSE and best policy inclusion measures, under \textbf{regime 3}. There are $20$ simulations per support configuration per $n$, for five support configurations. \label{fig:regime3}}
		\end{center}
\end{figure}

\clearpage

\begin{figure}[h]
		\begin{center}
			\includegraphics[scale=0.5]{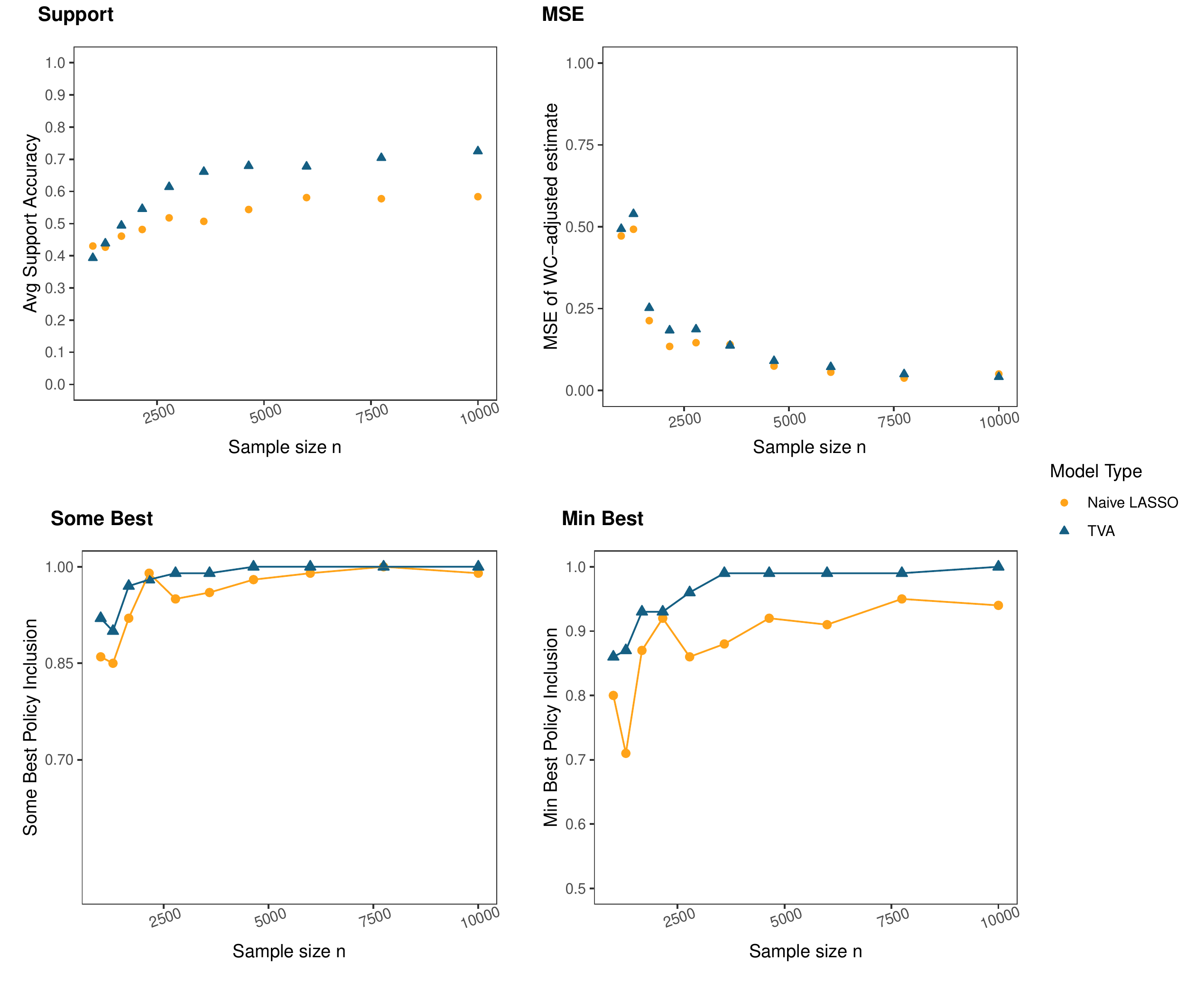}
			\caption{A comparison between the TVA estimator and a direct implementation of LASSO on (\ref{eq: TVA})  for support accuracy, MSE and best policy inclusion measures, under \textbf{regime 4}. There are $20$ simulations per support configuration per $n$,  for five support configurations. \label{fig:regime4}}
		\end{center}
\end{figure}

\clearpage

\begin{figure}[h]
		\begin{center}
			\includegraphics[scale=0.5]{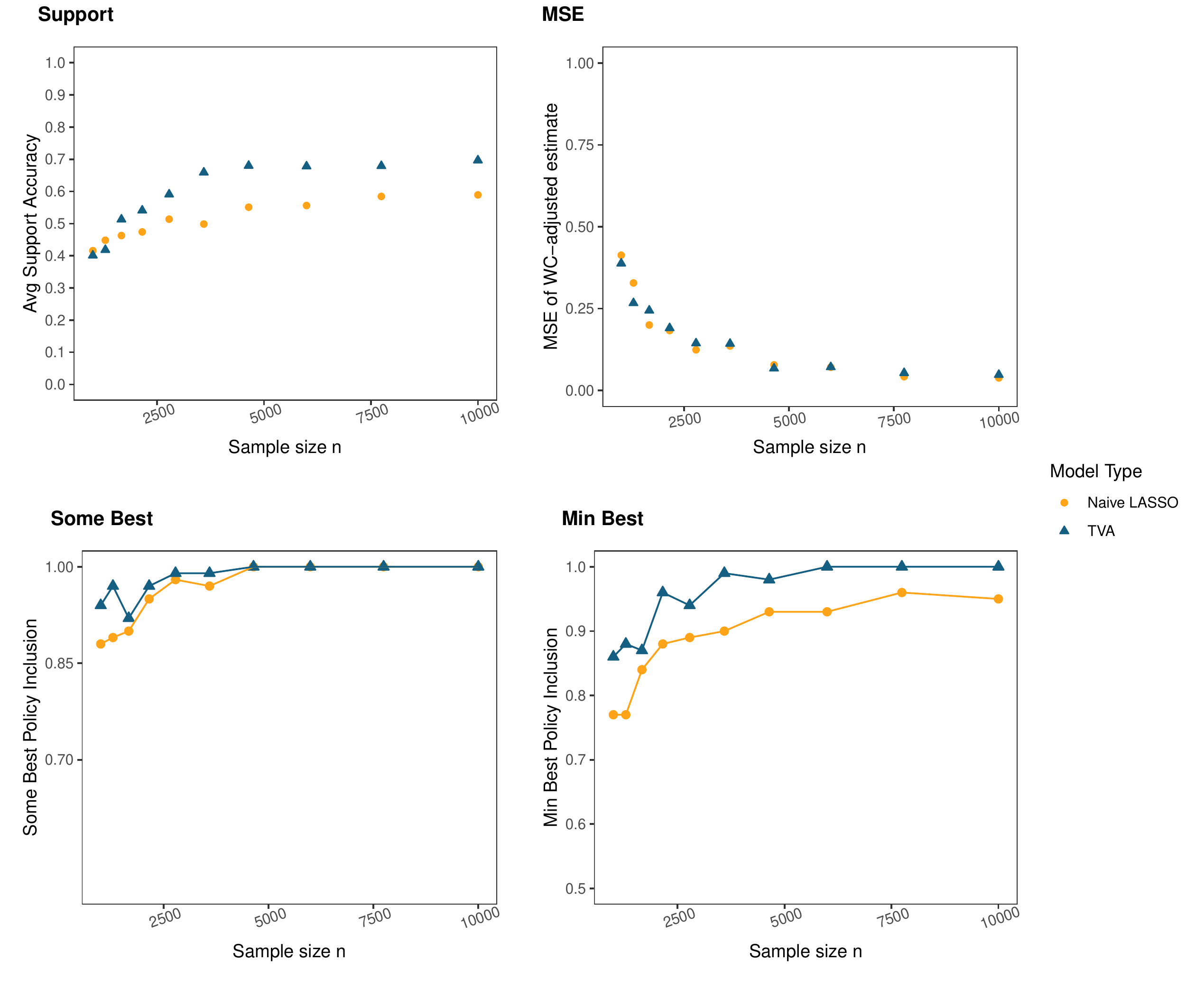}
			\caption{A comparison between the TVA estimator and a direct implementation of LASSO on (\ref{eq: TVA})  for support accuracy, MSE and best policy inclusion measures, under \textbf{regime 5}. There are $20$ simulations per support configuration per $n$,  for five support configurations. \label{fig:regime5}}
		\end{center}
\end{figure}

%
%
%
%

\clearpage
\setcounter{table}{0}
\renewcommand{\thetable}{F.\arabic{table}}
\setcounter{figure}{0}
\renewcommand{\thefigure}{F.\arabic{figure}}
\section{Extended Robustness}\label{sec: extended_robustness}
This section complements the discussion of Appendix \ref{sec: robustness}. We report the results from a fully saturated regression on the 75 unique policies to emphasize that some pooling choices made by TVA can be non-trivial making eyeballing a generally \textit{non}-reliable sanity check. We also elaborate more on the bootstrapping analysis suggested in Appendix \ref{sec: robustness} to explore TVA performance in the context of our experiment.  

\vspace{1cm}
\underline{Fully saturated regression}
Below we report the results from running a saturated regression of 75 raw coefficients for finely differentiated policies. We contrast this with the pooling and pruning choices made by the TVA estimator highlighted in different colors (red policies are pruned while green / blue policies are pooled together).  This makes the point very clearly that a simple eye-balling of these results would have been misleading in inferring the choices made by TVA.  Two examples are worth noting:

\begin{itemize}
\item Some seemingly efficient policies at standard confidence levels are pruned while others are pooled together: this is striking when comparing the policies in Figure \ref{fig:OLS_costs} (Shots/\$), Panel A (No Seed) in the last pooling profile (all were pruned) with those in Panel B (Trusted Seed) in the last pooling profile (all were pooled and kept in the support). 

\item Some policies that are individually underpowered (though positive) were pooled together (and kept in the support) suggesting that the pool ended up being powered and highly effective. An example can be seen in Figure \ref{fig:OLS_immunization} (Measles Shots), Panel C (Gossip Seed) in the last pooling profile.
\end{itemize}

\begin{figure}[!h]
\includegraphics[scale=0.45]{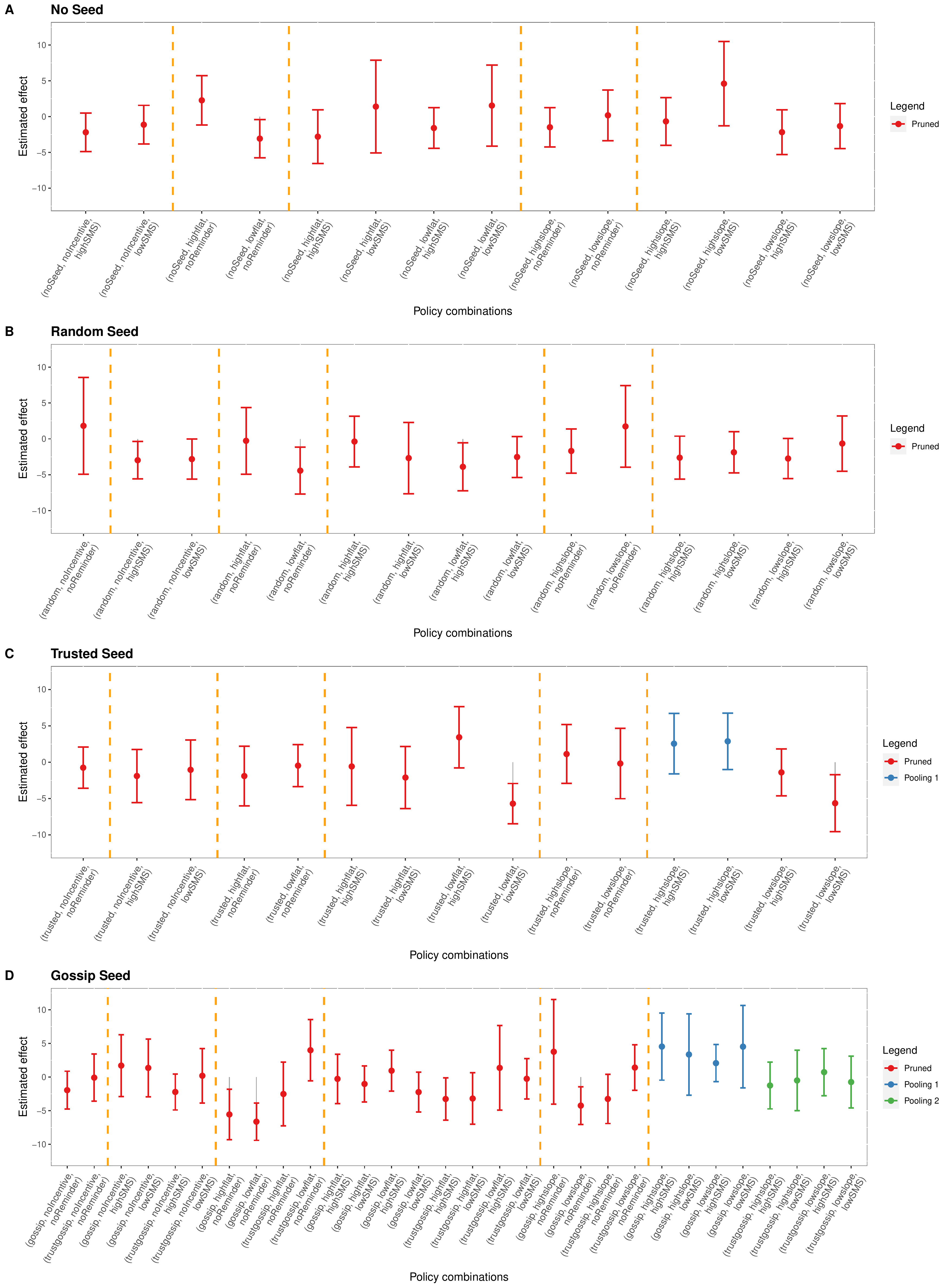}
\caption{A plot showing the coefficients from the OLS regression on all 75 unique policies for the Measles shot outcome. Panels are organized by seeds and within each panel the pooling profiles are delimited by dashed lines. Policies shown in red are policies pruned by the TVA Estimator while policies shown in other colors were pooled together.\label{fig:OLS_immunization}}
\end{figure}

\begin{figure}[!h]
\includegraphics[scale=0.45]{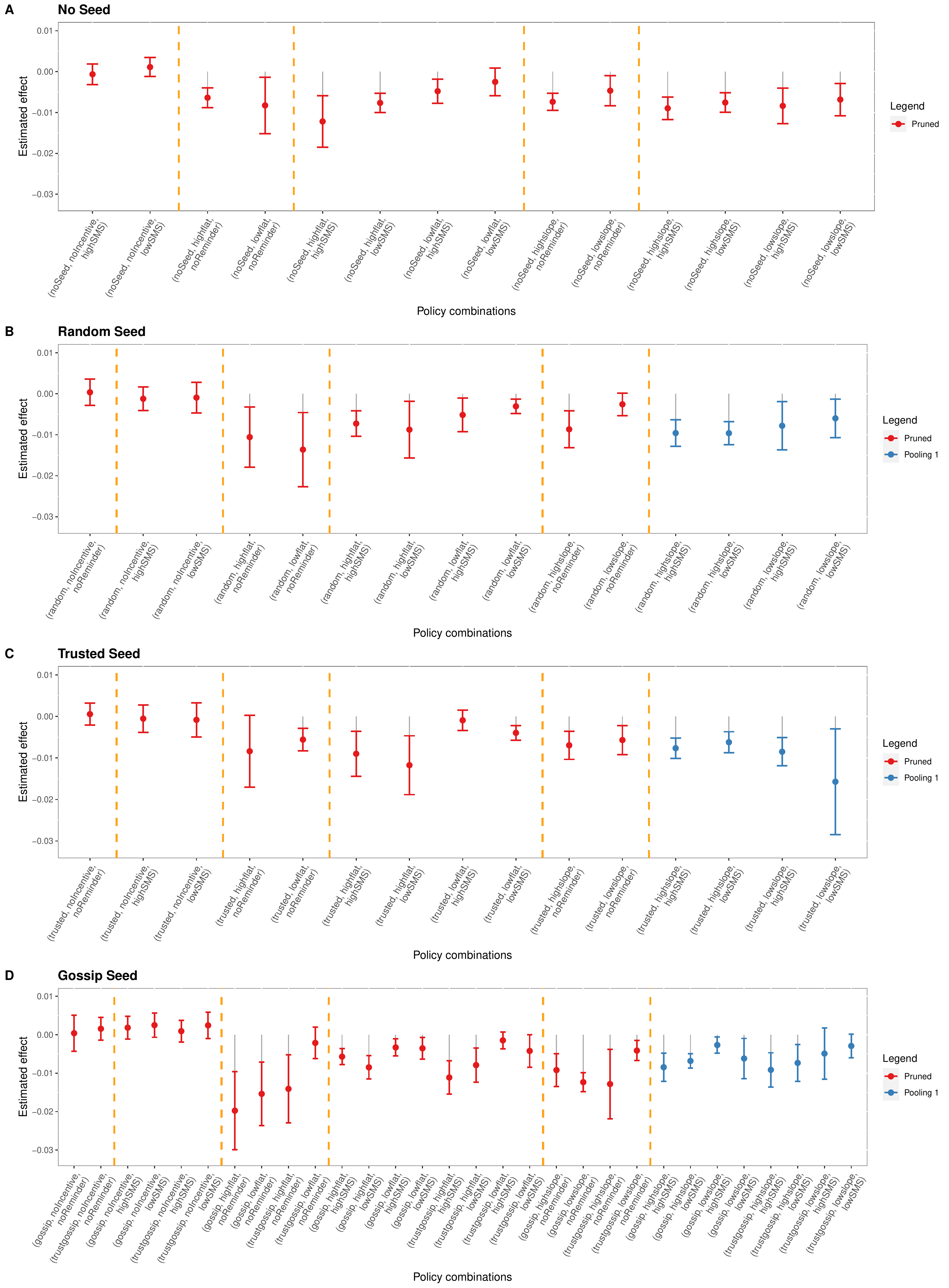}
\caption{A plot showing the coefficients from the OLS regression on all 75 unique policies for the Shots/\$ outcome. Panels are organized by seeds and within each panel the pooling profiles are delimited by dashed lines. Policies shown in red are policies pruned by the TVA Estimator while policies in blue were pooled together. \label{fig:OLS_costs}}
\end{figure}

\vspace{1cm}

\underline{Bootstrapping analysis}
It is natural to ask about the fragility of TVA to the particular draw of the data; precisely this concern motivates, for example, our implementation of winner's curse adjustments by \cite{andrews2019inference}, as well as simulations in Section \ref{sec: simulations} that directly speak to the variance of TVA. However, one might further wonder about just the observations in our dataset, with a concern akin to one about leverage of observations. An intuitive approach to address this is a bootstrapping analysis, where TVA is run on multiple bootstrapped samples. We can then speak to variation in both the set of supports selected as well the estimates of the pooled policies. Because this is a more exploratory analysis, its principal value lies in speaking to \emph{relative} stability of conclusions between the two policies for the two outcomes. \footnote{Note that we have slightly different goals here from the issue of bootstrapped standard errors} 

The 200 bootstrapped samples we run for each outcome are stratified at the policy level and results for each sample are displayed in Figures \ref{fig:boot_immunization} (immunizations) and \ref{fig:boot_costs} (immunizations /\$). For immunizations/\$, the support from the original sample was LASSO selected in 96\% of bootstrapped samples and the estimated coefficients are concentrated around our main estimate, including the best policy (Info Hubs (All), No Incentives, SMS (All)). Notably there is almost no winner’s curse adjustment since the best policy is consistently well separated from the second best. For the immunizations outcome we again observe little variation in the support, and the most effective policy  (Info Hubs, Slopes (All), SMS (All)) is identified as such in 77\% of the bootstrapped samples. However, there is considerably more variation in the winner's curse estimates, with some bootstrap samples sharply attenuating the best policy estimate. Taken altogether, this speaks to tighter competition and more sensitivity to leverage of certain observations for immunizations than immunizations/\$.


\begin{landscape}
\begin{figure}[!ht]
\includegraphics[scale=0.55]{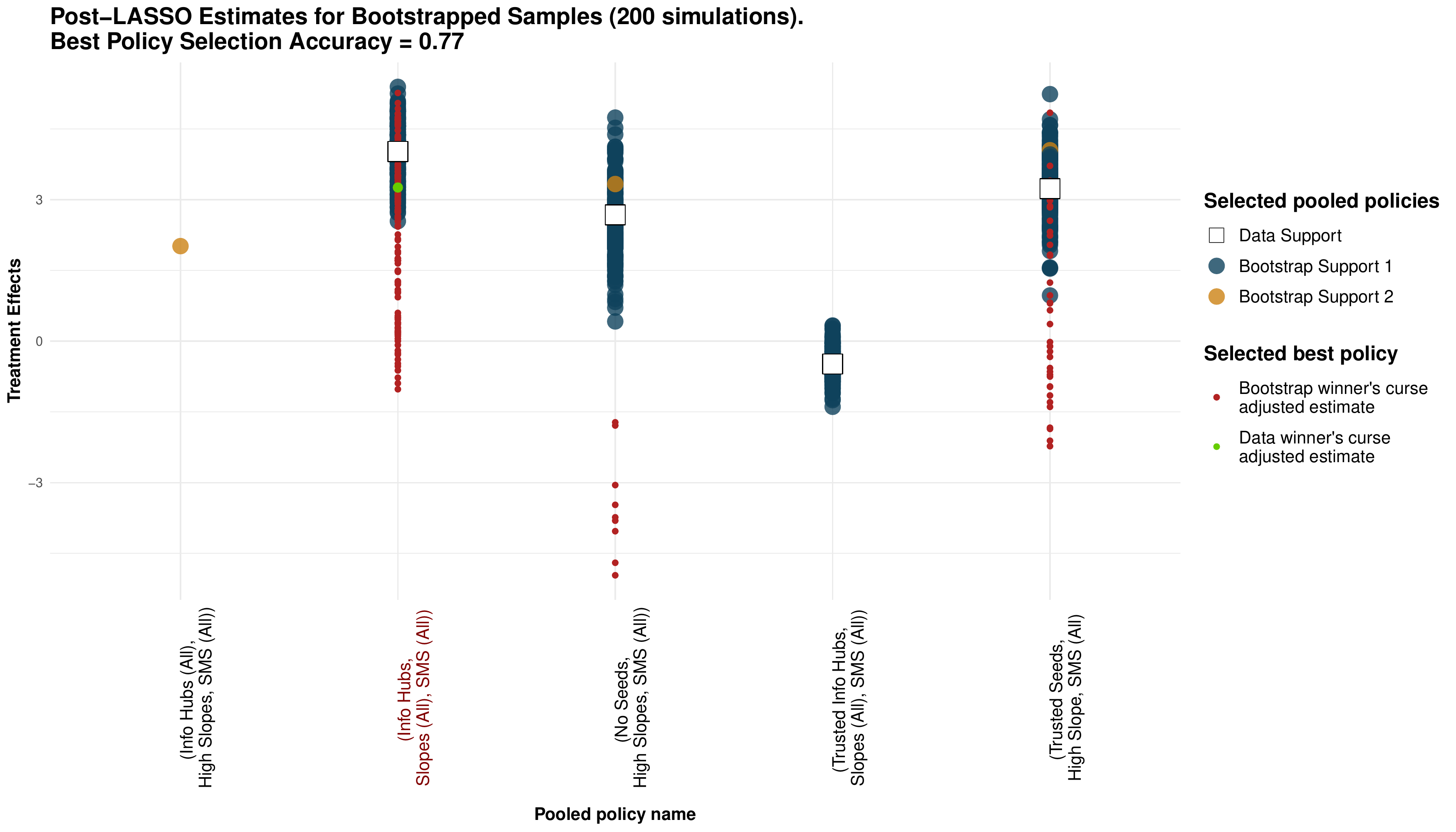}
\caption{Results of post-LASSO OLS estimates for the Measles Shot outcome are shown for 200 bootstrapped samples stratified at the policy level. The raw data marginal effects support is shown in white and “Bootstrap Supports” 1 and 2 are the two unique supports that were selected across all bootstrapped samples. We also report winner’s curse adjusted coefficients for the raw data (green) and the bootstrapped samples (red). On the x-axis, the raw data best pooled policy is highlighted in red and was selected in 77\% of the bootstrapped samples. Finally the minimum dosage best policy (Info Hubs, Low Slopes, Low SMS) is only selected when the TVA estimator actually selects the best pooled policy, hence 77\% of the time.
\label{fig:boot_immunization}}
\end{figure}

\begin{figure}[!ht]
\includegraphics[scale=0.55]{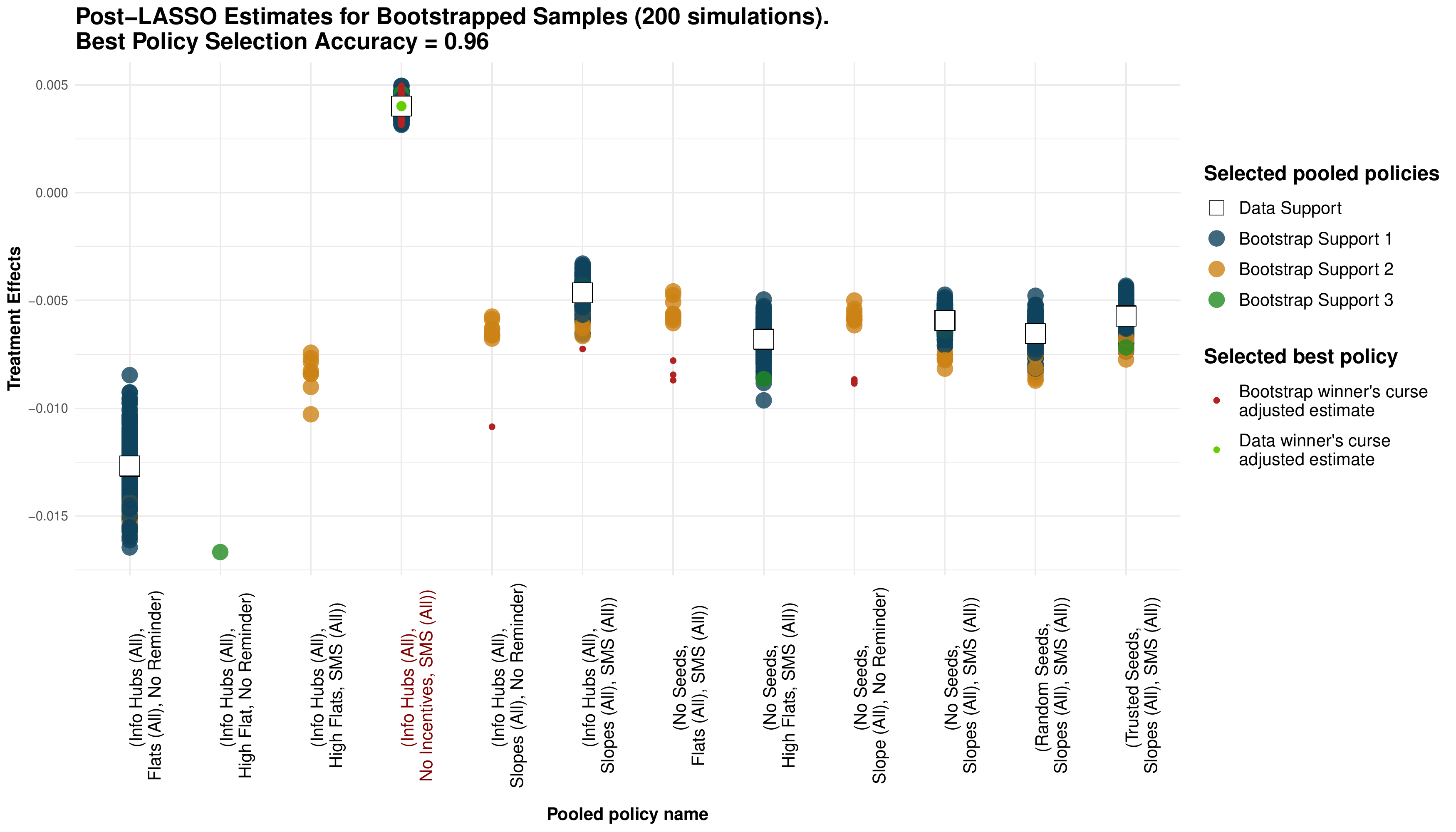}
\caption{Results of post-LASSO estimates for the Shots/\$ outcome are shown for 200 bootstrapped samples stratified at the policy level. The raw data marginal effects support is shown in white and “Bootstrap Supports” 1-3 are the three unique supports that were selected across all bootstrapped samples. We also report winner’s curse adjusted coefficients for the raw data (green) and the bootstrapped samples (red). On the x-axis, the raw data best pooled policy is highlighted in red and was selected in 96\% of the bootstrapped samples. Finally, the minimum dosage best policy (Info Hubs, No Incentives, Low SMS) is only selected when the TVA estimator actually selects the best pooled policy, hence 96\% of the time.
\label{fig:boot_costs}
}
\end{figure}
\end{landscape}

\clearpage
\setcounter{table}{0}
\renewcommand{\thetable}{G.\arabic{table}}
\setcounter{figure}{0}
\renewcommand{\thefigure}{G.\arabic{figure}}

\section{Appendix Figures}\label{sec:fig-table-appendix}


\begin{figure}[!h]
	\includegraphics[scale=0.35]{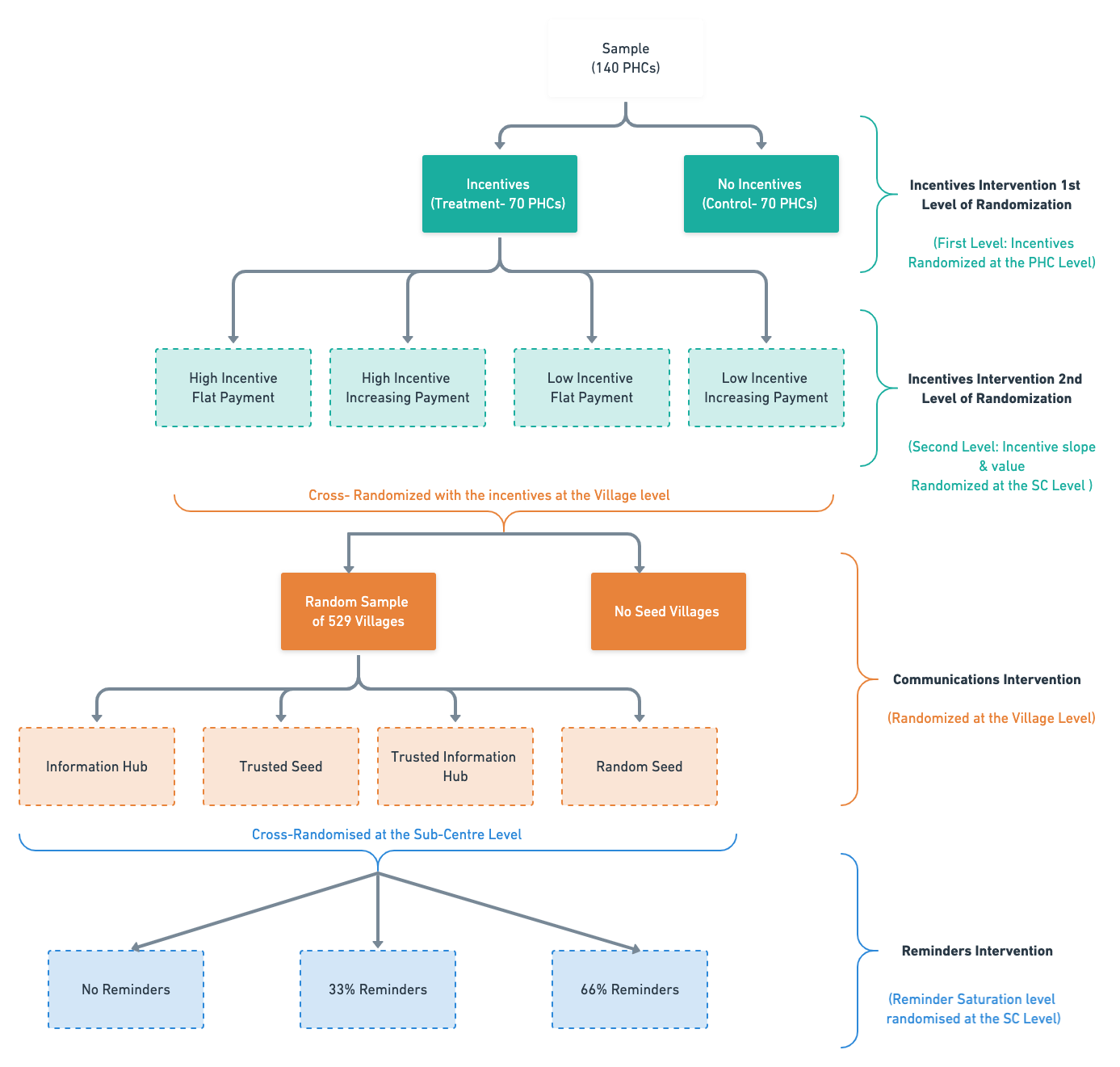}
	\caption{Experimental Design\label{fig:experimental_design}}
\end{figure}

\clearpage

\begin{figure}[!h]
\centerfloat
	\includegraphics[scale=0.61]{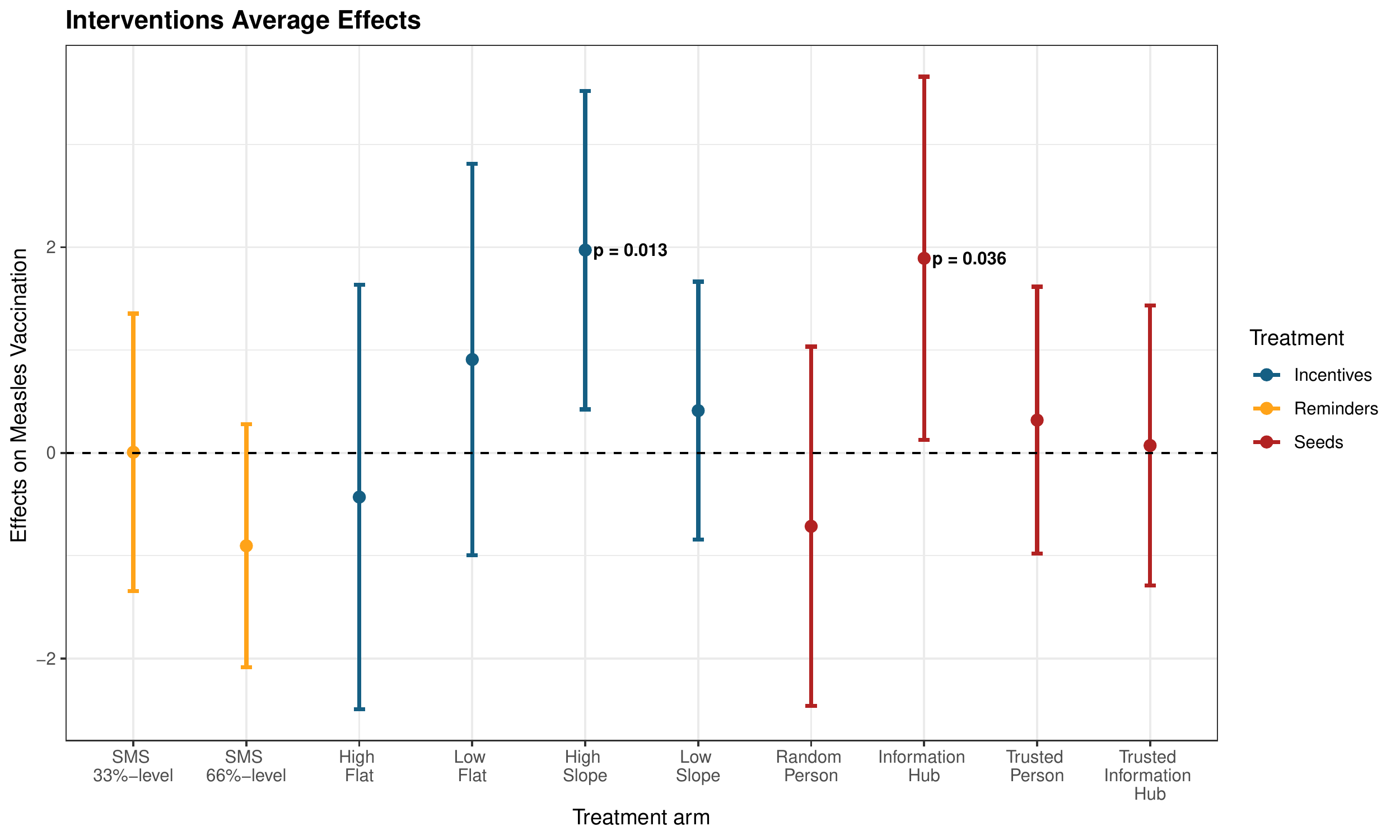}
\caption{Effects on the number of measles vaccinations relative to control (5.29) by reminders, incentives, and seeding policies,  for the full sample. The specification includes a dummy for being in the Ambassador Sample, is weighted by village population, controls for district-time fixed effects, and clusters standard errors at the sub-center level.\label{fig:agg-policy-effects_full} 
}
\end{figure}

\begin{figure}[!h]
	\includegraphics[scale=0.95]{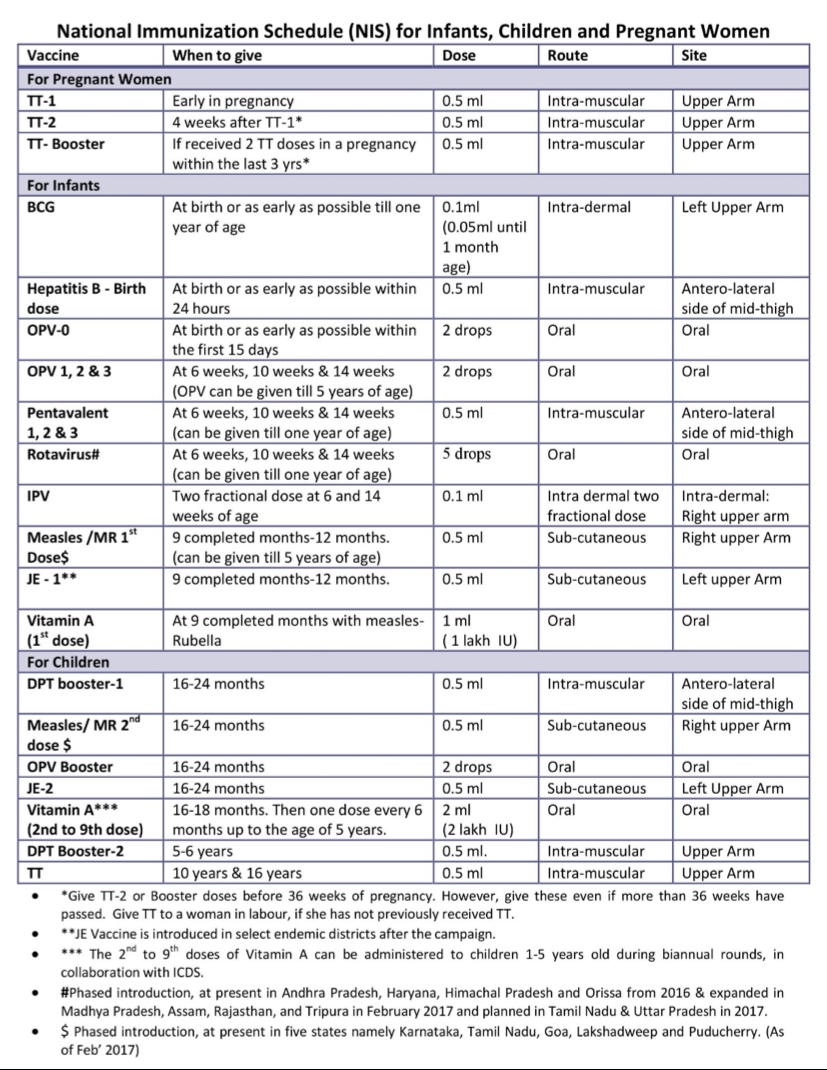}
\caption{National Immunization Schedule for Infants, Children, and Pregnant Women.\label{fig:immun_schedule}}
\end{figure}

\begin{figure}[!h]
	\includegraphics[scale=0.6]{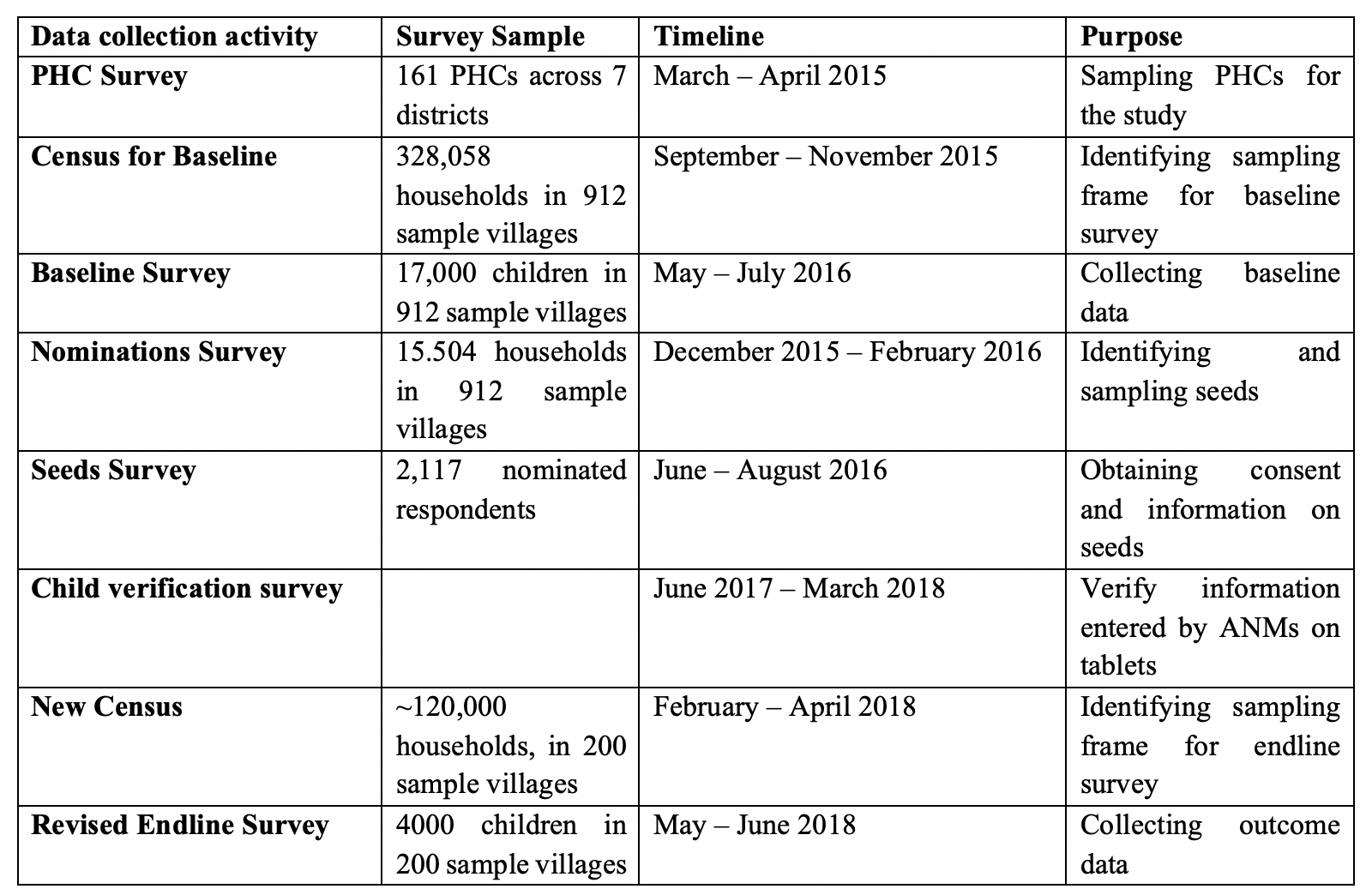}
	\caption{Overview of Survey Data Collection Activities. \label{fig:survey_data_collection}}
\end{figure}

\setcounter{table}{0}
\renewcommand{\thetable}{H.\arabic{table}}
\setcounter{figure}{0}
\renewcommand{\thefigure}{H.\arabic{figure}}
\clearpage
\section{Substitution Patterns}\label{sec:substitution}

\begin{table}[!htbp] \centering 
  \caption{Incentive Treatment Effects for Non-Tablet Children from Endline Data} 
  \label{table:subs-incentives} 
\scriptsize 
\begin{tabular}{@{\extracolsep{0.25pt}}lccccccc} 
\\[-1.8ex]\hline 
\hline \\[-1.8ex] 
 & \multicolumn{7}{c}{\textit{Dependent variable:}} \\ 
\cline{2-8} 
\\[-1.8ex] & At Least 2 & At Least 3 & At Least 4 & At Least 5 & At Least 6 & At Least 7 & Measles 1 \\ 
\\[-1.8ex] & (1) & (2) & (3) & (4) & (5) & (6) & (7)\\ 
\hline \\[-1.8ex] 
 High Slope & $-$0.158$$ & $-$0.052 & $-$0.076 & $-$0.196$$ & $-$0.187$$ & $-$0.027 & $-$0.135 \\ 
  & (0.062) & (0.072) & (0.093) & (0.106) & (0.101) & (0.105) & (0.108) \\ 
  & & & & & & & \\ 
 High Flat & $-$0.021 & $-$0.024 & $-$0.091 & $-$0.078 & $-$0.053 & 0.102 & 0.185 \\ 
  & (0.088) & (0.063) & (0.078) & (0.155) & (0.152) & (0.167) & (0.143) \\ 
  & & & & & & & \\ 
 Low Slope & 0.090 & 0.175$$ & 0.104 & $-$0.026 & $-$0.152$$ & $-$0.079 & 0.051 \\ 
  & (0.064) & (0.060) & (0.085) & (0.099) & (0.080) & (0.077) & (0.100) \\ 
  & & & & & & & \\ 
 Low Flat & 0.004 & 0.069 & $-$0.010 & $-$0.110 & $-$0.005 & $-$0.079 & $-$0.102 \\ 
  & (0.076) & (0.096) & (0.122) & (0.176) & (0.160) & (0.173) & (0.148) \\ 
  & & & & & & & \\ 
\hline \\[-1.8ex] 
Control Mean & 0.69 & 0.54 & 0.4 & 0.31 & 0.17 & 0.11 & 0.39 \\ 
Total Obs. & 1179 & 1165 & 1165 & 1042 & 1042 & 706 & 613 \\ 
Zeros Replaced & 0 & 0 & 0 & 0 & 0 & 0 & 0 \\ 
\hline 
\hline \\[-1.8ex] 
\textit{Note:}  & \multicolumn{7}{r}{\parbox[t]{\textwidth}{Specification includes District Fixed Effects, and a set of controls for seeds and reminders. Control mean shown in levels, and standard errors are clustered at the SC Level }} \\ 
\end{tabular} 
\end{table}

\begin{table}[!htbp] \centering 
  \caption{Seeds Treatment Effects for Non-Tablet Children from Endline Data} 
  \label{table:subs-seeds} 
\scriptsize 
\begin{tabular}{@{\extracolsep{0.25pt}}lccccccc} 
\\[-1.8ex]\hline 
\hline \\[-1.8ex] 
 & \multicolumn{7}{c}{\textit{Dependent variable:}} \\ 
\cline{2-8} 
\\[-1.8ex] & At Least 2 & At Least 3 & At Least 4 & At Least 5 & At Least 6 & At Least 7 & Measles 1 \\ 
\\[-1.8ex] & (1) & (2) & (3) & (4) & (5) & (6) & (7)\\ 
\hline \\[-1.8ex] 
 Random & $-$0.101$$ & $-$0.045 & 0.0003 & $-$0.143 & 0.015 & 0.058 & $-$0.017 \\ 
  & (0.059) & (0.066) & (0.088) & (0.122) & (0.102) & (0.102) & (0.101) \\ 
  & & & & & & & \\ 
 Information Hub & $-$0.040 & $-$0.121 & $-$0.025 & $-$0.112 & 0.018 & $-$0.105 & $-$0.092 \\ 
  & (0.082) & (0.080) & (0.113) & (0.135) & (0.123) & (0.073) & (0.119) \\ 
  & & & & & & & \\ 
 Trusted & 0.034 & $-$0.033 & 0.111 & 0.027 & 0.147 & $-$0.011 & 0.106 \\ 
  & (0.070) & (0.073) & (0.100) & (0.128) & (0.118) & (0.107) & (0.111) \\ 
  & & & & & & & \\ 
 Trusted Information Hub & $-$0.103 & $-$0.082 & $-$0.031 & $-$0.209$$ & $-$0.099 & $-$0.057 & $-$0.344$$ \\ 
  & (0.075) & (0.079) & (0.100) & (0.113) & (0.086) & (0.082) & (0.099) \\ 
  & & & & & & & \\ 
\hline \\[-1.8ex] 
Control Mean & 0.78 & 0.66 & 0.5 & 0.58 & 0.3 & 0.22 & 0.62 \\ 
Total Obs. & 469 & 461 & 461 & 389 & 389 & 251 & 231 \\ 
Zeros Replaced & 0 & 0 & 0 & 0 & 0 & 0 & 0 \\ 
\hline 
\hline \\[-1.8ex] 
\textit{Note:}  & \multicolumn{7}{r}{\parbox[t]{\textwidth}{Specification includes District Fixed Effects, and a set of controls for incentives and reminders. Control mean shown in levels, and standard errors are clustered at the SC Level }} \\ 
\end{tabular} 
\end{table}

\begin{table}[!htbp] \centering 
  \caption{Reminders Treatment Effects for Non-Tablet Children from Endline Data} 
  \label{table:subs-reminders} 
\scriptsize 
\begin{tabular}{@{\extracolsep{0.25pt}}lccccccc} 
\\[-1.8ex]\hline 
\hline \\[-1.8ex] 
 & \multicolumn{7}{c}{\textit{Dependent variable:}} \\ 
\cline{2-8} 
\\[-1.8ex] & At Least 2 & At Least 3 & At Least 4 & At Least 5 & At Least 6 & At Least 7 & Measles 1 \\ 
\\[-1.8ex] & (1) & (2) & (3) & (4) & (5) & (6) & (7)\\ 
\hline \\[-1.8ex] 
 33\% & 0.079 & 0.093 & 0.070 & 0.033 & 0.019 & $-$0.138$$ & 0.011 \\ 
  & (0.058) & (0.066) & (0.085) & (0.104) & (0.094) & (0.080) & (0.086) \\ 
  & & & & & & & \\ 
 66\% & 0.031 & 0.096$$ & 0.042 & $-$0.074 & $-$0.073 & $-$0.044 & $-$0.097 \\ 
  & (0.053) & (0.055) & (0.069) & (0.091) & (0.079) & (0.067) & (0.084) \\ 
  & & & & & & & \\ 
\hline \\[-1.8ex] 
Control Mean & 0.64 & 0.48 & 0.34 & 0.28 & 0.19 & 0.13 & 0.46 \\ 
Total Obs. & 1179 & 1165 & 1165 & 1042 & 1042 & 706 & 613 \\ 
Zeros Replaced & 0 & 0 & 0 & 0 & 0 & 0 & 0 \\ 
\hline 
\hline \\[-1.8ex] 
\textit{Note:}  & \multicolumn{7}{r}{\parbox[t]{\textwidth}{Specification includes District Fixed Effects, and a set of controls for seeds and incentives. Control mean shown in levels, and standard errors are clustered at the SC Level }} \\ 
\end{tabular} 
\end{table}

\clearpage

\section{Data Validation}\label{sec:data_validation}

A household survey was conducted to monitor program implementation at the child-level---whether the record entered in the tablet corresponded to an actual child, and whether the data entered for this child was correct. This novel child verification exercise involved J-PAL field staff going to villages to find the households of a set of randomly selected children which, according to the tablet data, visited a session camp in the previous four weeks. Child verification was continuous throughout the program  implementation, and the findings indicate high accuracy of the tablet data. We sampled children every week to ensure no additional vaccine was administered in the lag between them visiting the session camp and the monitoring team visiting them. Data entered in the tablets was generally of high quality. There were almost no incidences of fake child records, and the child's name and date of birth were accurate over 80\% of the time. For 71\% of children the vaccines overlapped completely (for all main vaccines under age of 12 months). Vaccine-wise, on average, 88\% of the cases had matching immunization records. Errors seem genuine, rather than coming from fraud: they show no systematic pattern of inclusion or exclusion and are no different in any of the treatment groups.

\clearpage
\setcounter{table}{0}
\renewcommand{\thetable}{J.\arabic{table}}
\setcounter{figure}{0}
\renewcommand{\thefigure}{J.\arabic{figure}}

\section{Baseline Statistics} \label{sec: baseline}

\begin{table}[ht]\caption{Selected Baseline Statistics of Haryana Immunization}\label{table:Haryana_desc}
\centering
\small{
\begin{tabular}{lc}
  \hline\hline
 & Population-Weighted Average \\ 
 \hline \\  [-4mm]

\textbf{Baseline Covariates--Demographic Variables}  \\  [0.3mm]

\emph{(Village Level Averages)}  \\  [0.3mm]

Fraction participating in Employment Generating Schemes & 0.045 \\ 
  Fraction Below Poverty Line (BPL) & 0.187 \\ 
  Household Financial Status (on 1-10 scale) & 3.243 \\ 
  Fraction Scheduled Caste-Scheduled Tribes (SC/ST)  & 0.232 \\ 
  Fraction Other Backward Caste (OBC)  & 0.21 \\ 
  Fraction Hindu & 0.872 \\ 
  Fraction Muslim & 0.101 \\ 
  Fraction Christian & 0.001 \\ 
  Fraction Buddhist & 0 \\ 
  Fraction Literate & 0.771 \\ 
  Fraction Unmarried & 0.05 \\ 
  Fraction of Adults Married (living with spouse) & 0.504 \\ 
  Fraction of Adults Married (not living with spouse) & 0.002 \\ 
  Fraction of Adults Divorced or Seperated & 0.001 \\ 
  Fraction Widow or Widower & 0.039 \\ 
  Fraction who Received Nursery level Education or Less & 0.17 \\ 
  Fraction who Received Class 4 level Education & 0.086 \\ 
  Fraction who Received Class 9 level Education & 0.158 \\ 
  Fraction who Received Class 12 level Education & 0.223 \\ 
  Fraction who Received Graduate or Other Diploma level Education  & 0.081 \\ 

\textbf{Baseline Covariates--Immunization History of Older Cohort} \\  [0.3mm]
\emph{(Village Level Averages)}  \\  [0.3mm]

 Number of Vaccines Administered to Pregnant Mother & 2.271 \\ 
 Number of Vaccines Administered to Child Since Birth & 4.23 \\ 
 Fraction of Children who Received Polio Drops & 0.998 \\ 
 Number of Polio Drops Administered to Child & 2.989 \\ 
 Fraction of Children who Received an Immunized Card & 0.877 \\ 

\textbf{Number of Observations} \\  [0.3mm]
 Villages & 903 \\ 
   \hline\hline

\end{tabular}}
\end{table}
 \label{table:baseline}

\clearpage

\section{Information Hub Questions}\label{sec:hub_questions}

\begin{enumerate}
	\item Random seeds: In this treatment arm, we did not survey villages. We picked six ambassadors randomly from the census. 
	
	\item 	Information hub 
	seed: Respondents were asked to identify who is good at relaying information. 
	
	We used the following script to ask the question to the 17 households: 
	\begin{quote}
	``Who are the people in this village, who when they share information, many people in the village get to know about it. For example, if they share information about a music festival, street play, fair in this village, or movie shooting many people would learn about it. This is because they have a wide network of friends, contacts in the village and they can use that to actively spread information to many villagers. Could you name four such individuals, male or female, that live in the village (within OR outside your neighbourhood in the village) who when they say something many people get to know?”
	\end{quote}
	
	\item ``Trust'' seed: Respondents were asked to identify those who are generally trusted to provide good advice about health or agricultural questions (see appendix for script)
	
	We used the following script to elicit who they were: 
	
	\begin{quote}
	``Who are the people in this village that you and many villagers trust, both within and outside this neighbourhood? When I say trust I mean that when they give advice on something, many people believe that it is correct and tend to follow it. This could be advice on anything like choosing the right fertilizer for your crops, or keeping your child healthy. Could you name four such individuals, male or female, who live in the village (within OR outside your neighbourhood in the village) and are trusted?" 
	\end{quote}
	
	\item ``Trusted information hub'' seed: 
	Respondents were asked to identify who is both trusted and good at transmitting information
	
	\begin{quote}
	``Who are the people in this village, both within and outside this neighbourhood, who when they share information, many people in the village get to know about it. For example, if they share information about a music festival, street play, fair in this village, or movie shooting many people would learn about it. This is because they have a wide network of friends/contacts in the village and they can use that to actively spread information to many villagers. Among these people, who are the people that you and many villagers trust? When I say trust I mean that when they give advice on something, many people believe that it is correct and tend to follow it. This could be advice on anything like choosing the right fertilizer for your crops, or keeping your child healthy. Could you name four such individuals, male or female, that live in the village (within OR outside your neighbourhood in the village) who when they say something many people get to know and are trusted by you and other villagers?"
	\end{quote}

\end{enumerate}

\end{document}